\documentclass[11pt]{article}
\usepackage{amssymb,amsmath}
\usepackage{epsfig}
\usepackage{graphicx}
\usepackage{hyperref}
\usepackage{mathpazo}

\newtheorem{theorem}{Theorem}[section]

\newtheorem{definition}[theorem]{Definition}
\newtheorem{algorithm}[theorem]{Algorithm}
\newtheorem{claim}[theorem]{Claim}
\newtheorem{lemma}[theorem]{Lemma}

\newtheorem{corollary}[theorem]{Corollary}
\newtheorem{fact}[theorem]{Fact}

\newcommand{\qedsymb}{\hfill{\rule{2mm}{2mm}}}
\newenvironment{proof}[1][]{\begin{trivlist}
\item[\hspace{\labelsep}{\bf\noindent Proof#1:\/}] }{\qedsymb\end{trivlist}}

\def\bfly{\triangleright\triangleleft}

\newcommand{\be}{\begin{eqnarray}}
\newcommand{\ee}{\end{eqnarray}}

\newcommand\abs[1]{{\left| {#1} \right|}}

\newcommand\ket[1]{{ |{#1} \rangle }}
\newcommand\bra[1]{{ \langle {#1} | }}
\newcommand\ketbra[1]{{\ket{#1}\bra{#1}}}

\newcommand{\ignore}[1]{}

\renewcommand{\epsilon}{\varepsilon}

\begin{document}

\title{\bf  On the complexity of Commuting Local Hamiltonians, 
and tight conditions for Topological Order in such systems}

\author{Dorit Aharonov\thanks{School of Computer Science and
Engineering, The Hebrew University,
Jerusalem, Israel}
 \and Lior Eldar\thanks{School of Computer Science and
Engineering, The Hebrew University,
Jerusalem, Israel.}}

\date{\today}

\maketitle

\begin{abstract}
The local Hamiltonian problem plays the equivalent role of SAT in quantum 
complexity theory. Understanding the complexity of the 
intermediate case in which the constraints are quantum 
but all local terms in the Hamiltonian commute, is of importance 
for conceptual, physical and computational complexity reasons.
Bravyi and Vyalyi showed in 2003 \cite{BV}, 
using clever applications of the representation theory of C*-algebras,  
that if the terms in the Hamiltonian 
are all two-local, the problem is in NP, and the entanglement in the 
ground states is local. The 
general case remained open since then. 
In this paper we extend the results of Bravyi and Vyalyi
beyond the two-local case, to the case of three-qubit interactions.
We then extend our results even further, and show that NP verification 
is possible for three-wise interaction between qutrits as well, as long as 
the interaction graph is embedded on a planar lattice, or more generally, 
"Nearly Euclidean" (NE). 
The proofs imply that in all such systems, the entanglement in the 
ground states is local. 
These extensions imply an intriguing sharp transition 
phenomenon in commuting Hamiltonian systems: 3-local NE 
systems based on qubits and qutrits cannot be used to construct
Topological order, as their entanglement is local, 
whereas for higher dimensional qudits, or for interactions of at least $4$ qudits,
Topological Order is already possible, via Kitaev's Toric Code construction.  
We thus conclude that Kitaev's Toric Code construction 
is optimal for deriving topological order 
based on commuting Hamiltonians.
\end{abstract}

\section{Introduction}\label{sec:intro}
The problem of approximating the ground energy of a 
local Hamiltonian describing a physical system
is one of the major problems in condensed matter
physics; in the area of quantum computation this problem is 
called the local Hamiltonian problem \cite{Kit1}.  
Formally, in the $k$-local Hamiltonian problem, we are given
a Hamiltonian $H$ which is a sum of positive semidefinite terms,
each acting on a set of at most $k$ out of $n$ qubits, 
where $k$ is of order $1$, and each term is of bounded norm.
Moreover, we are given two numbers, $b>a$ such that $b-a\ge \frac{1}{poly(n)}$.
We are asked whether $H$ has an eigenvalue below $a$ 
or all its eigenvalues are above $b$, and we are promised that 
the instance belongs to one of the two cases.

It turns out that the problem of understanding ground states and 
ground values of local Hamiltonians, central to condensed matter physics, 
is the quantum generalization of 
one of the most important problems in classical 
computational complexity, namely, SAT. 
Indeed, in a seminal work, Kitaev has shown that in parallel to 
the important of the SAT problem in NP theory, the local Hamiltonian 
problem is complete for
the quantum analogue of NP (denoted QMA) in which both witness and verifier
are quantum rather than classical.  
The analogy between the quantum and the classical 
problems is derived by viewing the terms of the Hamiltonians as 
generalizing the notion of classical constraints; 
energies are viewed as a penalty for a constraint violation.  
For example, to view the local constraints for the classical SAT 
as a special instance of local Hamiltonians, 
we assign for each clause a projection on the assignment 
forbidden by this clause. The projections we derive are of course 
all projections in the computational basis; 
in the general quantum case, the terms need not be 
diagonal in any particular basis, and the ground state can   
be highly entangled. 
This connection linking the physics and the 
computational complexity  problems has drawn much attention 
over the past few years, and has led to many exciting results and insights 
(eg., \cite{Kit1, KKR,Aha,Oli,BSAT,BV}). 

The computational view of the local Hamiltonian problem and 
its connection to classical NP problems, has led Bravyi and Vyalyi 
in \cite{BV} to the following very natural question: what would 
happen if we only
generalize from classical to quantum ``half way'':
we allow the terms
in the Hamiltonian to be projections in any basis, 
but we restrict them in that all the projections pairwise 
commute. 
We are asked to decide whether the ground energy is $0$ 
(namely, there exists a state which is in the ground space of all 
projections) or it is larger than $0$ 
(for pairwise-commuting projections, the overall 
energy, namely eigenvalue, of such a state must be at 
least $1$). 
This problem is the commuting local Hamiltonian problem\footnote{
We note that this problem is equivalent to the more general case when the terms
can be taken as positive-semidefinite operators, since for such an input
one can replace each local term with a projection on the 
non-zero eigenspaces of that term}.  

The interest in the commuting Hamiltonian problem is related to several 
important issues in quantum computational complexity. 
The first is conceptual:
a common intuition is that the counter intuitive phenomena
in quantum mechanics stems from the fact that non-commuting operators are
involved (cf the Heisenberg's uncertainty principle). 
One might conjecture, using this intuition, that the
commuting local Hamiltonian problem is far weaker than the general local 
Hamiltonian problem,
and might be of the same complexity as SAT, namely, lie in NP. 
However, a counter intuition exists: The intriguing strictly 
quantum phenomenon of 
topological order, which is exhibited for example in 
Toric codes \cite{Kit2}, can be achieved by ground 
states of commuting Hamiltonians. 
It is thus natural to ask where does the computational
complexity of the commuting Hamiltonian problem lie: is it in NP, 
is it perhaps
quantum-NP complete (where here the relevant quantum analogue of NP
is in fact, $QMA_1$, where there is only one sided error) 
or maybe the commuting local Hamiltonian problem
defines an intermediate computational
class of its own? 

The study of this problem can also be viewed as tightly related to 
an exciting major open problem in quantum 
Hamiltonian complexity: the question of whether a PCP-like  
theorem holds in the quantum setting or not \cite{Aha2}. 
Embarrassingly, this problem is still open 
even for the seemingly much easier case of commuting local Hamiltonians. 
Clearly, a PCP-type theorem would follow trivially if the commuting local 
Hamiltonian problem were in NP,  
but even if this were not true, one might still hope to prove a 
PCP-type theorem for the restricted problem before proceeding to the 
more general case. 
We recall that several results in quantum Hamiltonian complexity, 
such as the area law in 1Dim \cite{Has}
the decay of correlations in gapped 
Hamiltonians \cite{Has2}, and quantum gap amplification 
\cite{Aha2}
were all proven exactly in this way, by starting from the easier 
commuting case, and generalizing from there; it seems reasonable 
to hope that better understanding of the commuting case would 
help clarify the quantum PCP conjecture in general.   
More generally,  
understanding the complexity of the commuting 
local Hamiltonian problem might 
not only shed light on the role of commutativity in  
quantum Hamiltonian complexity, and possibly lead to progress on 
open problems in quantum complexity theory, but also, it seems 
that an answer to this question would necessarily require 
new insights regarding the nature of multi-particle entanglement. 
 
In \cite{BV} an important step was made towards resolving the computational
complexity of the commuting local Hamiltonian problem.
Bravyi and Vyalyi showed that for $k=2$, 
namely for two-body interactions, regardless
of the dimensionality $d$ of the particles involved, the problem
lies in NP. 
The method they use is interesting by itself; 
They cleverly apply the theory of 
representations of $C^*$-algebras to the problem.  
However, their methods break down for three-wise interactions.
The general problem was thus left open by \cite{BV}, and no progress 
was noted on this problem
since its inception in 2003.   

\subsection{Results: The Complexity of Commuting Hamiltonians} 
In this paper we extend the results of \cite{BV}
to three-local interactions with the following two results. 
We prove: 
\begin{theorem}\label{MainClaim}
The problem of $3$-local commuting Hamiltonian on qubits is in NP.
\end{theorem} 
\begin{theorem}\label{MainClaim2} (Roughly stated)
The problem of 
$3$-local commuting Hamiltonian on qutrits is in NP, 
as long as the interaction graph is planar, and moreover, 
{\emph nearly Euclidean}. 
\end{theorem} 

The notion of Nearly Euclidean will be defined later (see 
Definition \ref{def:NEdef}); roughly, it formalizes the 
requirement that the embedding makes sense physically: 
no area on the plane can have a particularly high 
density of particles, and only close-by particles can interact. 
This of course includes also the interesting 
special case of periodic lattices, or small 
perturbations of those.

Unlike what might be expected, the extension does not seem to 
follow easily from the result of \cite{BV}.
When attempting to apply the proof of $\cite{BV}$ to the case in which 
three local interactions are involved, (even when the particles are assumed 
to be qubits) severe problems occur. 
We will provide the overview of the proof of Bravyi and Vyalyi, 
why it breaks down for three-wise interactions and how we overcome this, 
later in the introduction. 
Let us first describe some 
implications of our results to the seemingly unrelated topic of 
conditions for topological order. 

\subsection{Results: Implications for tight conditions on Topological Order} 
Topological order is a purely quantum phenomenon; 
Roughly, a state exhibits a topological order if 
there exists a state orthogonal to it, and the two cannot be 
distinguished or connected by a local operator.
Such characteristics are extremely valuable in the context of 
fault tolerance, and topological order has attracted much attention 
for theoretical and implementation purposes for that reason. 
A celebrated example of a system exhibiting 
Topological Order in the context of quantum computation 
is the Toric Code, due to Kitaev \cite{Kit2}; 
it can be defined as the ground space of a set of $4$-local commuting 
operators on qubits arranged on the two dimensional grid, 
or alternatively, by $3$-local commuting 
interactions between $4$-dimensional particles. 
Topological order defined via 
commuting local Hamiltonians is particularly interesting: 
for example, recently
it has been shown (\cite{BH}) that such systems are also 
resilient to local perturbations.
It is therefore natural to ask whether it is possible to 
achieve topological order 
in ground states of local commuting Hamiltonians, 
with smaller dimensionality or with less particles interacting.   
Using the above results, we can resolve this problem to the negative.
We show that Kitaev's construction is optimal in a well 
defined sense.

\begin{theorem} {\bf Tight conditions for Topological order} \label{thm:TO}  
(Roughly) 
Consider a system of particles with 
commuting interactions which are either $2$-local, or
they are $3$-local and the dimensionality of the particles is at most $3$.
Moreover, assume the interaction graph is Nearly Euclidean planar.   
Then this system cannot exhibit Topological Order, 
and moreover, in  a well defined sense, the entanglement in the 
ground space is local. 
On the other hand, there exist nearly Euclidean
planar systems of $3$-local interactions with
particles of dimensionality $4$ that exhibit Topological order, 
hence we derive a tight boundary between local entanglement and Topological order. 
\end{theorem}

To explain how our results are related to conditions on topological order, 
observe the following. 
A key property of topological order states is that their entanglement 
is non-local. 
In particular, Bravyi, Hastings and Verstraete showed in \cite{BHV} (see 
Theorem (\ref{thm:BHV}), that if a nearest 
neighbor quantum circuit generates a state with Topological 
Order on the $n\times n$ grid, the circuit 
has to be of depth $\Omega(\sqrt{n})$. 
The methods we use, as well as those of \cite{BV}, however, 
imply that the ground space has a basis of states with localized
entanglement. 
More precisely, the proofs of 
Theorems \ref{MainClaim} and \ref{MainClaim2} also imply that 
any nearly Euclidean planar commuting Hamiltonian 
system, that is either $2$-local, or $3$-local with particle dimensionality 
$d\le 3$, has a basis of 
eigenstates all of which can be generated by a  
constant depth  
quantum circuit whose gates act on nearest neighbor particles on the plane 
(where by nearest neighbor, we mean neighbors in the interaction graph of the 
Hamiltonian). This means that such systems cannot exhibit 
topological order in all the states in their groundspace.  

The Toric Codes, however, 
can be easily seen as a the ground space of an
instance of the commuting $3$-local 
Hamiltonian problem on qudits of dimension $4$ (by gluing pairs of qubits 
together - see Section \ref{sec:TO} and in particular 
Figure (\ref{fig:TC})). 
Thus, this is an instance of $CLH(3,4)$ in which any 
state in the groundspace exhibits topological order. 

Our results thus imply that in the context of 
 "physical" planar systems, there exists a tight 
boundary between Topological Order systems and constant-depth systems:
For $k>3$ and all $d\ge 2$ or $k=3$ and $d\ge 4$ there exist 
nearly Euclidean planar
systems which exhibit Topological order, whereas
for $k=3$ and $d<4$, or for $k=2$ and any $d$, 
all nearly Euclidean planar systems have a constant-depth diagonalizing circuit 
and cannot exhibit Topological order.  
We deduce that Kitaev's construction cannot be simplified either in terms 
of particle dimensionality or number of particles in each interaction, 
and so it is optimal for commuting Hamiltonians. 

\subsection{Overview of the proofs of Theorems \ref{MainClaim} and 
\ref{MainClaim2}} 
\subsubsection{Bird's eye view on the proof of Bravyi and Vyalyi for two-local case}
Let us embark on trying to explain the proof, by first explaining the main 
idea in the proof of Bravyi and Vyalyi of the two-local case, for any 
particle dimensionality. To do that, we consider the hypergraph 
describing the interactions in the Hamiltonian. 
We observe that in the two-local case, 
every particle is the center of a ``star'' of interactions - 
the interactions acting on $q$ intersect only on $q$. 
This is not true when interactions are three-local, as one can see 
in Figure (\ref{fig:hyper}).  

\begin{figure}[ht]
\center{
 \epsfxsize=3in
 \epsfbox{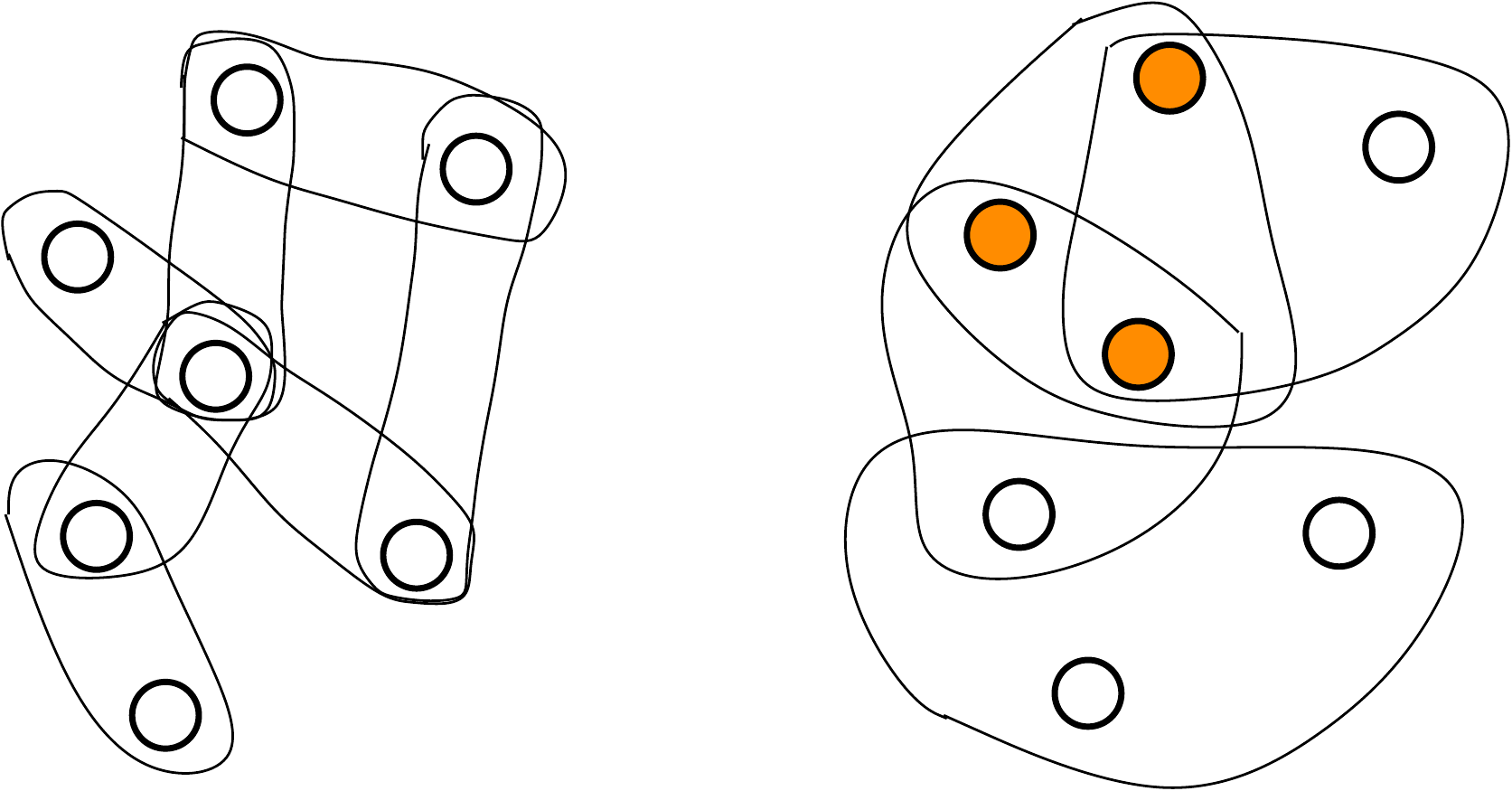}} 
 \caption{\label{fig:hyper} 
 An example of a $2$-local interaction hypergraph (left) and $3$-local interaction hypergraph (right).
 The particles which are star centers are unfilled, and those that are not, are filled in.}
\end{figure}

Bravyi and Vyalyi prove a lemma (restated here as  
Lemma \ref{lem:BV}), which shows that 
particles which are centers of ``stars'', are what we call ``separable''. 
This means that if $q$ is such a center of a star, 
its Hilbert space ${\cal H}_q$ can be 
{\it decomposed} to a direct sum of subspaces, which are all preserved 
by all interactions involving $q$:  
$${\cal H}_q = \bigoplus_{\alpha}{\cal H}_{\alpha}^q.$$
Moreover, each subspace ${\cal H}_{\alpha}^q$ can be written as a 
tensor product of sub-particles,  
$${\cal H}_{\alpha}^q = \otimes \left(\bigotimes_{k}{\cal H}_{\alpha}^{q.k}\right)
$$ where $k$ runs over all particles that interact with $q$, 
and the interaction between $q$ and $k$ is non-trivial only on 
the relevant sub-particle, ${\cal H}_{\alpha}^{q.k}$.  
This way, the restriction to one of the subspaces 
implies a decoupling of the interactions involving $q$ to 
interactions that act on separate sub-particles! 
When all particles are center of stars as is the case for the two-local, 
after each particle is restricted to one of its subspaces the restricted Hamiltonian is a
set of disjoint edges.

From this \cite{BV} derive a proof that the two-local 
problem lies in NP - essentially, the witness is the specification 
of the choice $\alpha$ of the correct subspace of each particle, 
in which the groundstate lies. 

The above proof also implies that in the two-local case, there is an 
eigenbasis of the Hamiltonian
in which any eigenstate (and in particular any ground state) has
a very limited and local structure of entanglement - 
the state can be generated by a depth-two
quantum circuit which uses only two-local gates. 
Of course, a natural question is whether these techniques can be
applied for the more general case, namely, for higher values of $k$. 

\subsubsection{What fails when trying to apply \cite{BV} to three-wise interactions}

Trivially, when generalizing from $2$-local interactions to $3$-local interactions we immediately loose the star topology
which was a crucial component in (\cite{BV}). See, for example, Figure 
\ref{fig:ex1}. 

\begin{figure}[ht]
\center{
 \epsfxsize=2in
 \epsfbox{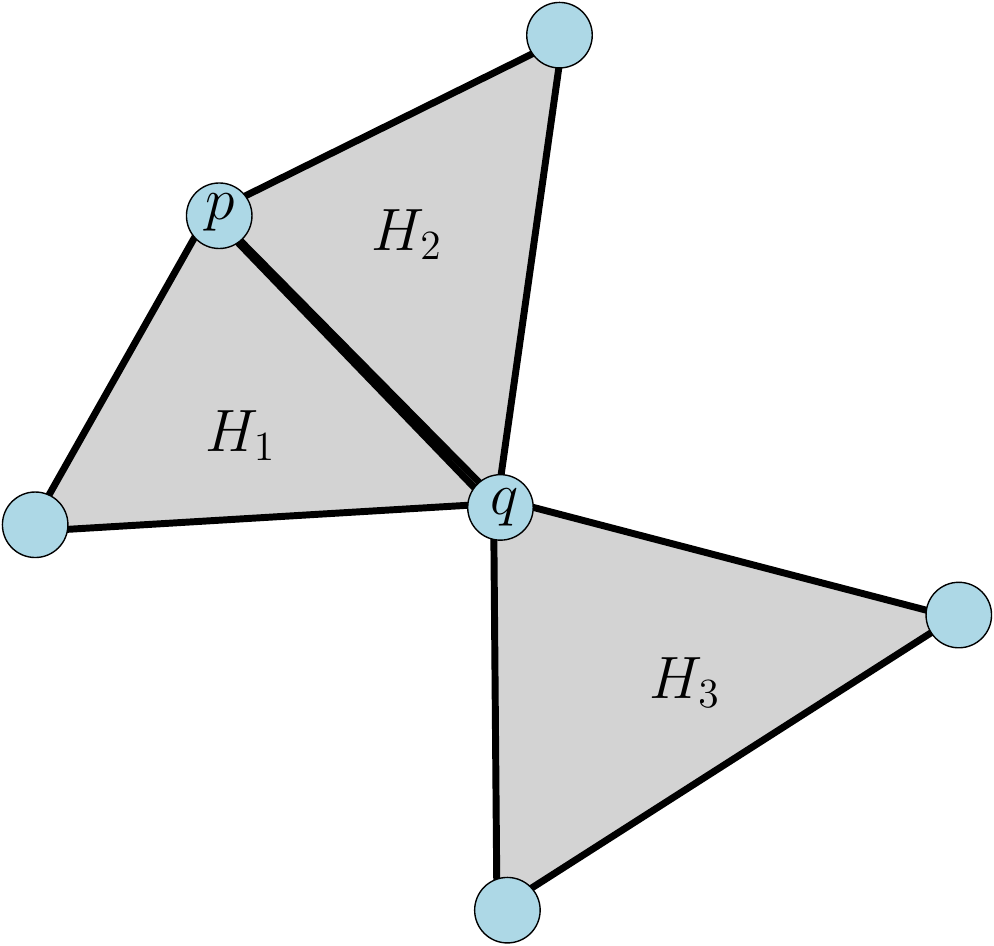}}
 \caption{\label{fig:ex1} In the example both $H_1$ and $H_2$ share a single 
qubit $q$ with $H_3$.
By Lemma (\ref{lem:BV}) $H_1$, and $H_3$ 
agree on some decomposition of $q$, and so do $H_2$ and $H_3$.
Yet, because $H_1$ and $H_2$ share two qubits $p$ and $q$, 
they do not agree necessarily on the same decomposition of $q$.} 
\end{figure}

However, this example is not truly a problem. 
Because we restrict our attention to qubits, 
the low dimensionality implies that one cannot "block-diagonalize" an 
operator on a qubit $q$ in more than one way. 
Thus it turns out that in the example above, there is indeed a "consensus" 
decomposition of $q$ preserved by all $3$ operators on $q$.
However, consider the example of $4$ operators on 
$4$ qubits in Figure \ref{fig:ex2}.

\begin{figure}[ht]
\center{
 \epsfxsize=3in
 \epsfbox{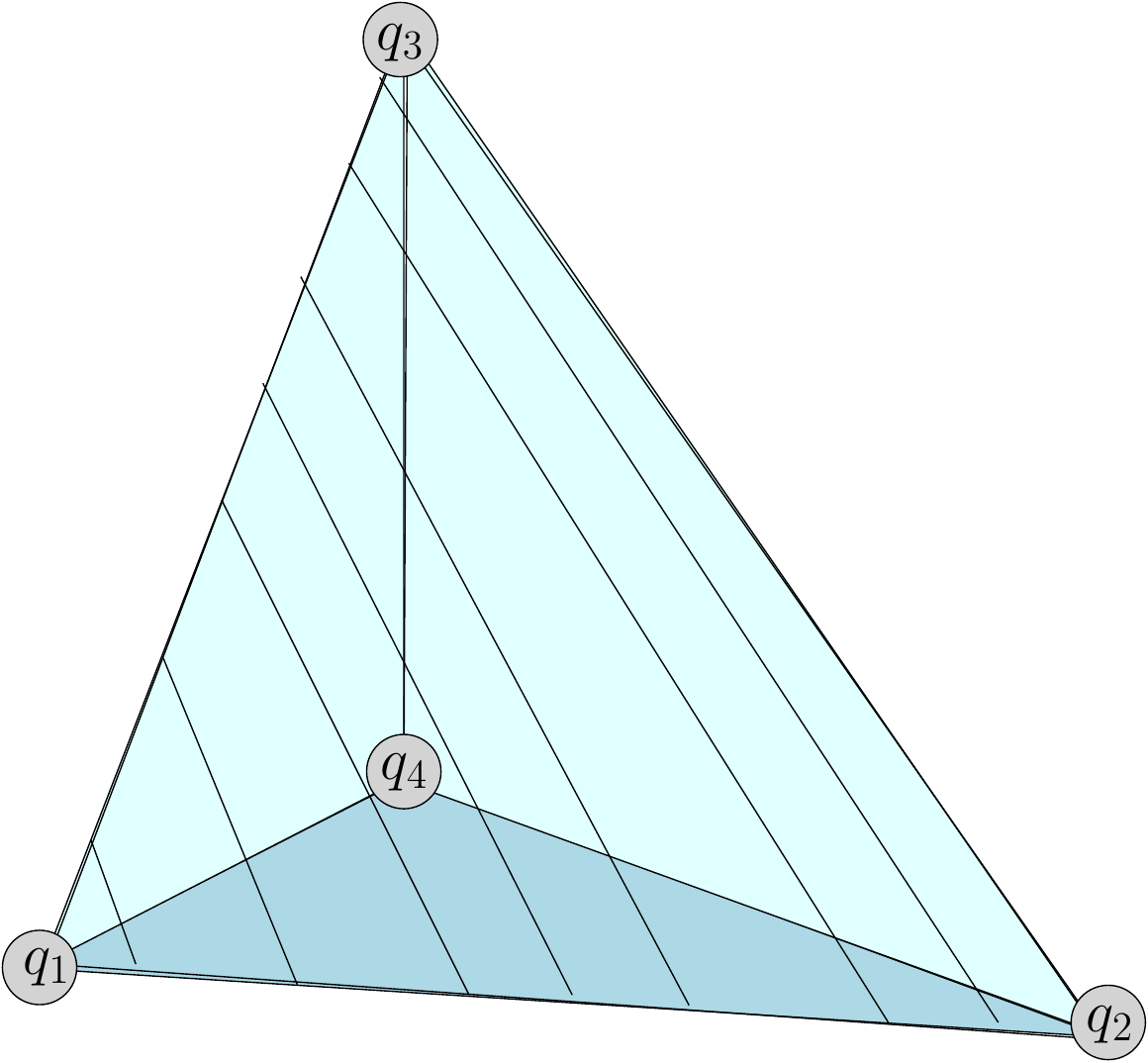}} 
 \caption{\label{fig:ex2} An example of a topology of interactions which can 
be defined in such a way that, say, for $q_1$, 
no decomposition exists, 
which is preserved by all operators acting on it.}
\end{figure}

Since any pair of operators share $2$ qubits, it may be that for none 
of the qubits does there exist a direct-sum
decomposition which is preserved by all operators on that qubit.
This in fact emanates from the nature of the commutativity 
relation for $3$-local terms: it can be generated using 
not just one particle as in $\cite{BV}$ but 
may involve more complex relations involving two particles.

\subsubsection{General idea of the proof}
The way we overcome this obstacle is by showing that truly complex 
structures, such as the example of Figure \ref{fig:ex2}, 
can only be of local nature, after we remove all separable qubits from the 
system. In other words, any attempt to expand such examples  
by adding more interactions with additional qubits 
inevitably makes are least one qubit separable. 

The proofs of the two theorems turn out to require quite different 
tools to achieve this goal; 
Below we provide overviews of those proofs. 

\subsubsection{Proof sketch for the case of three-wise interactions of qubits}
Let us start with explaining the proof of Theorem \ref{MainClaim}, 
namely the case of qubits. 
We build upon the analysis of
\cite{BV}, identifying qubits which behave "classically".
Those are qubits for which there exists an orthogonal basis, such 
that all operators are block-diagonal w.r.t. this basis.
We call those qubits ``separable''. 
If the original system has a zero eigenstate, 
so does a restricted version of the system when 
each separable qubit is restricted to one of those subspaces. 
This restriction of each of the qubits can be provided by 
the prover, which implies
that those qubits can be removed from the problem altogether. 
Unlike in the case of two-body interactions, however, not all qubits can 
be removed this way. Most of the proof of 
Theorem \ref{MainClaim} is focused
on handling the residual problem. 

The main point is that in the residual problem, 
configurations such as that in Figure \ref{fig:ex2} may exist, 
however, when the qubits are non-separable, 
they can only grow to some bounded constant size. 
More precisely, we prove that the fact that the remaining qubits are 
not separable, implies severe restrictions on the geometry of the problem.  
We show that any connected component of 
the residual problem after removing the separable qubits 
contains a certain one dimensional structure in the interaction graph,  
which we call a \emph{backbone}.  
The backbone is a long sequence of local Hamiltonian 
terms, as in Figure (\ref{fig:bkbn}). 

\begin{figure}[ht]
\center{
 \epsfxsize=3in
 \epsfbox{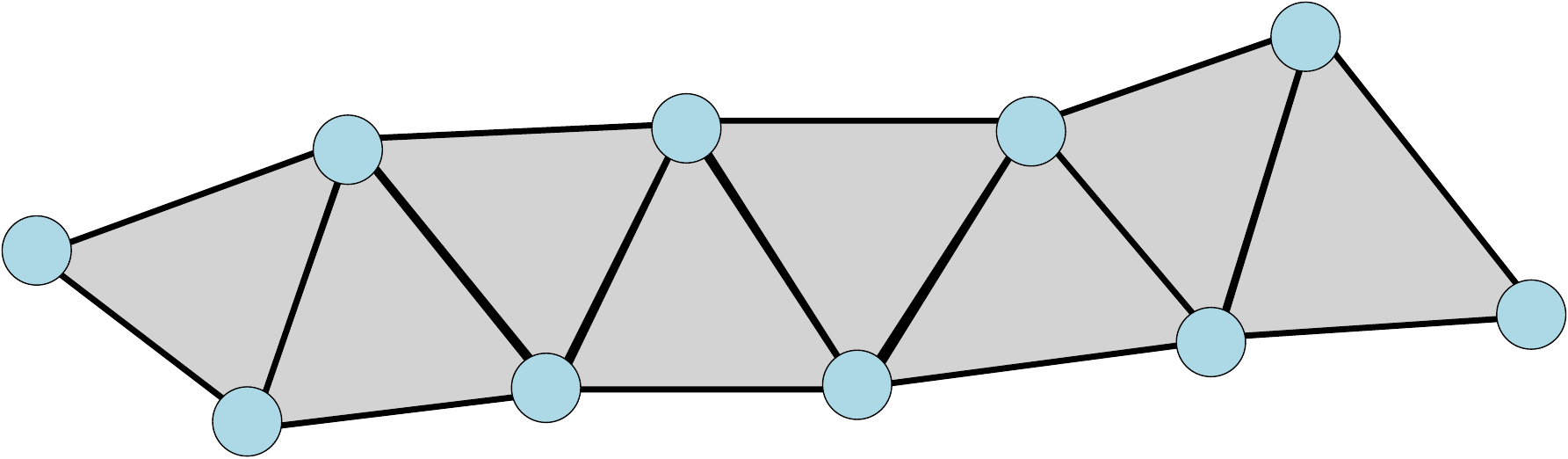}}
 \caption{\label{fig:bkbn} After removing all classically-behaving qubits, the residual Hamiltonian contains a certain $1$-dimensional structure, denoted as the "backbone".  We prove that any term in the Hamiltonian
must act on at least two "close" backbone qubits.}
\end{figure}

The key property is that all Hamiltonian terms in one connected component
in the residual problem  
must act on at least \emph{two} qubits in the backbone of that component. 
Moreover, the backbone qubits that a term acts on, must be 
within a constant distance, in terms of edges of the backbone. 
Thus, Hamiltonian terms cannot connect distant particles in 
the backbone;    
Moreover, there cannot be other ``shortcuts''
connecting distant parts of the backbone 
through interactions with qubits outside of the backbone. 
We call this an \emph{almost one dimensional structure}. 
We use this structure to prove that we can combine 
sets of constantly many nearest neighbor particles in the backbone into
large particles, which are still of constant dimension, 
but in which interactions only act on nearest neighbors 
large particles. The methods of \cite{BV} can be used here again to show that 
the large particles can then be sliced, in such a way that all 
interactions become two-local. 

\subsubsection{Proof sketch for the planar three-wise interactions of qutrits}

The proof of Theorem \ref{MainClaim2} which 
proves a similar result with qutrits on the plane is quite different. 
The first step is similar:  
we identify qutrits  that behave "classically" (i.e., separable), 
which enables us to remove them, or at least reduce their dimensionality to $2$; 
We then proceed to prove strong geometric restrictions on 
the residual problem, in which no particle can be further reduced. 

To deduce the geometrical constraints, we would like to make use of
the fact that
the remaining particles are not separable. However, the phenomenon of
separability of a qutrit is somewhat more involved than that of a
qubit. We say a qutrit is separable if there exists a non-trivial
decomposition of its space to orthogonal  subspaces, and each of these
subspaces is preserved by all operators acting on the qutrit. But
whereas for qubits, due to the low dimensionality, such a
decomposition, if exists, is unique, for qutrits this is no longer the
case, as there is freedom in the choice of basis of a two dimensional
space. This prevents the proof for the case of qubits from going
through.

We can overcome this difficulty in the case of 
nearly Euclidean planar interactions.
For such instances, we show that the fact that all qudits are inseparable,
implies that each qudit can be acted upon by at most a constant number
of operators; More precisely, the degree of each vertex in the
interaction graph is at most $5$. This is a completely geometrical
constraint on the interaction graph of the residual problem, and it is
in fact a very strong constraint which suffices for our purposes.
We show that if one attempts to cover large regions of the plane with 
faces with only three edges (corresponding to the three-local interactions) 
and such that no vertex has degree $6$ or more, then 
``holes'' must be formed, i.e., faces 
which are not distorted-triangles but have more edges, 
and thus do not correspond to any 
term in the Hamiltonian. 
The point is that those holes have a constant density; 
in other words, each vertex in the graph has a 
hole within constant distance (measured in number of faces separating 
the vertex from the hole).  
We use those holes in order to cut 
the interaction graph to smaller (constant size) pieces
which when combined to particles, induce two-local interactions. 
The problem thus lies in NP by \cite{BV}. 

We provide more detailed overviews of the two proofs in Subsections (\ref{sec:3lcl2})
and (\ref{sec:3lcl3}), this will be easier 
after we provide some of the necessary definitions and 
notations in Section \ref{sec:notation}.  

\subsection{Conclusions and Further work} 
The results in this paper focus on two aspects of commuting 
Hamiltonians: the first is extending the containment in NP 
also for three body interactions, where
 a fundamental barrier is encountered exactly when topological order can 
be present in the ground space.  Three body interactions seemed before 
as the barrier standing between \cite{BV} and the extension towards 
a proof of containment in NP of the general case; here we show that 
the barrier is far more intriguing, and has to do with the 
appearance of topological order. 

The second aspect is the proof that 
Kitaev's celebrated construction of topological order using Toric codes 
is optimal, a statement which is of interest in various contexts, 
such as physical implementations of topological order states, 
the understanding of topological quantum codes,
and our general understanding of quantum multiparticle 
entanglement and quantum states. 

The barrier exposed in this paper by no means implies that 
we should neglect the hope to prove that 
the general commuting Hamiltonian problem is in NP. 
In fact, we hope that the barrier encountered here would clarify as to 
how we need to proceed in order to prove (or possibly disprove) 
the conjecture that the general problem lies in NP. 

One possible direction to explore in order to attack the 
general commuting case is the following. As is well known, 
topological order states such as Toric codes do have 
short classical descriptions, which are in fact classical descriptions 
of small depth quantum circuits, except those circuits are
{\it non-local} (i.e., not nearest neighbor on the grid). These are called  
MERA \cite{Vid,AV}. Those descriptions  
allow computing local observables efficiently using a classical computer.
From the point of view of $NP$ verification, this is clearly sufficient.
If one can show that such MERA descriptions exist for any ground state 
of commuting Hamiltonian, this would imply that the problem lies in NP. 
It is possible that the algebraic methods of \cite{BV} can be used in an
innovative way (perhaps by recursion or by other means) 
to imply that there exist such MERA-type poly-size classical
descriptions of eigenstates for any $k$-local commuting Hamiltonian on a grid.

Our proofs are quite involved.
An indication to this complexity is that though our results imply 
that the ground states in the systems we study can be generated by 
constant depth quantum circuits, and thus can only create local entanglement, 
the locality we derived is quite large - 
the scale of entanglement involves a
number of qubits or qutrits which is of the order of a few tens.  
An intriguing open question is whether the 
complexity of our results is essential to the problem, 
or it can be removed. 
If not, the proofs indicate that 
in the three-local systems we study, though the 
entanglement structure is restricted to being local, 
still quite complicated (though local) structures 
of entanglement can emerge, which span large though 
constant sets of particles, with the constant much larger than 
the natural scales in the system (say, $2$ or $3$).  
It would be interesting to understand this aspect further. 

Finally, a technicality in the proof of the qutrit case is that the graph 
is required to be NE, rather than just planar. 
We speculate that in fact the NE restriction is not necessary; 
it is used only in the proof of 
Claim \ref{cl:cutting} 
and we believe the claim holds for the general planar case. 
Indeed, this restriction does not have strong implications for the 
main message of the paper (namely, the tight boundary between TO and 
local entanglement), since TO is in any case studied only in such NE 
systems; still it would be nice to close
that corner and it would make our statements somewhat cleaner.   

{~}

\noindent{\bf Organization of Paper:} 
The structure of this paper is as follows: 
In Section (\ref{sec:notation}) we lay 
out some notations and definitions, and in section (\ref{sec:2lcl}) we restate 
an important lemma in the representation theory of $C^*$-algebra that was at the center of \cite{BV}, and use it to reprove their result 
that the $2$-local Hamiltonian problem is in NP 
(perhaps providing a slightly simpler representation of the proof).  
Then, in section (\ref{sec:3lcl2}) we prove that $CLH$ for $3$-local operators on qubits has a classical verification protocol.
In section (\ref{sec:3lcl3}) we extend this result for qutrits, in cases where the interaction graph is a nearly Euclidean planar graph.
We then use these two proofs, and the result of $\cite{BHV}$ to derive 
tight conditions for Topological Order.

\section{Background, Notations and Definitions}\label{sec:notation}

\subsection{Hamiltonians and Hilbert Spaces} 
We use the following standard notation: 
\begin{itemize}

\item We denote Hilbert spaces by graphical symbols: 
${\cal H}, {\cal H}_i$, etc. 
 The set of linear operators over the complex numbers, 
acting on a given Hilbert space 
${\cal H}$, is denoted by $L({\cal H})$. 

\item 
Unless otherwise noted, we denote by ${\cal H}$ the Hilbert space
 of $n$ qudits:
$$ {\cal H} = {\cal H}_1 \otimes \hdots \otimes {\cal H}_n. $$
The dimensionality of the qudits is denoted by $d$. 

\item 
A $k$-local operator $h$ is an operator which acts 
on a subset of size $k$ of the $n$ 
qudits $S \subseteq \left\{1,\hdots,n\right\}$,
hence $|S| = k$, and we have 
$H\in \mathbf{L}
(\otimes_{j\in S}{\cal H}_j) \otimes \left(\otimes_{j\notin S} I_j \right)$.

\end{itemize}

\noindent Less standard notation includes: 
\begin{itemize} 

\item 
To specify that an operator
$H$ acts non-trivially on some specific qubit $q$, we write $H(q)$.
We say an operator acts non-trivially on a particle $q$ if the 
operator cannot be written as a scalar on $q$ tensor some operator on 
the remaining particles.  

\item
The set of qudits examined non-trivially by an operator $H_i$ is denoted 
by $A_i$; The set of particles examined non-trivially 
by a set of operators $B$, is denoted by $A_B$. 

\end{itemize}

\subsection{The Local Hamiltonian Problem and its interaction graph} 
\begin{definition}

\textbf{The $(k,d)$ local Hamiltonian problem for commuting operators, 
$CLH(k,d)$}

\noindent 
The $(k,d)$ local commuting Hamiltonian problem 
on $n$ qudits of dimension $d$, denoted by $CLH(k,d)$, is defined 
as follows. We are given a set $S$ of $poly(n)$ $k$-local projections, $H_i$, 
acting on $n$ particles each of dimension 
$d$, such that all terms in $S$ 
pairwise commute. We are asked whether there exists an 
eigenstate of $H=\sum_{i\in S} H_i$ with eigenvalue $0$.
\end{definition}

\begin{definition}

\textbf{ $G_S$: Interaction graph of a $CLH$ instance}

\noindent 
The interaction graph of an instance $S$ of $CLH(k,d)$ 
is the graph $G_S = (V,E)$, where $V$ is a set of $n$ nodes, each 
corresponding to a qudit, and an edge connecting nodes $i$ and $j$ is in $E$ 
(namely, $(i,j)\in E$) if there exists some $H_m\in S$
such that both $i$ and $j$ belong to $A_m$. 
\end{definition}

\begin{definition}

\textbf{Neighborhood of a qudit in an instance $S$ of $CLH(k,d)$}

\noindent
We denote by ${\cal N}_S(q)$ the neighboring set of a 
qudit $q$ w.r.t. $S$, namely,  
the set of all qudits $p$ adjacent to $q$ in $G_S$.
\end{definition}

\subsection{Operators Preserving Subspaces} 
An operator $A$ is said to preserve a subspace $S$ 
if $A(S)\subseteq S$. 
The following facts are trivial to prove: 

\begin{fact}\label{fact:trivial1}
If $A$ is Hermitian, if it preserves a subspace $S$ it also 
preserves the orthogonal complement of $S$. 
\end{fact}

\begin{fact}\label{fact:trivial2}
If a linear 
operator $A$ commutes with a projection on a subspace $S$, then 
$A$ preserves $S$. 
\end{fact}

\subsection{Algebras}
In this paper we consider finite dimensional $c*$-algebras. 
For the purposes of this paper, these are 
complex algebras (denoted ${\cal A}$) 
of linear operators (described by matrices) 
with the additional restriction that 
${\cal A}$ is closed under the operation of taking adjoints of operators (i.e,
the {\emph dagger, $\dagger$} operation). 

We often refer to the algebra {\emph generated by a given set of linear 
operators}, referred to as generators. The generators are 
always a set of matrices of the same dimensionality; 
and the algebra generated by them is defined
 either as the minimal algebra that contains the linear subspaces
spanned by the generators, or equivalently, 
the algebra generated by the set of generators union with the identity 
matrix. 
 
\subsection{Algebras induced by operators}

\begin{definition}\label{def:induced} 
\textbf{Algebra induced by an operator}

\noindent
Let $H=H(q)$ be an operator on $q$, and let us write 
\begin{equation}\label{eq:deq}
H = \sum_{\alpha} A_{\alpha} \otimes B_{\alpha}
\end{equation}
such that
$A_{\alpha}$ acts on $q$, and $B_{\alpha}$ acts on the 
rest of the environment, and the set $\left\{B_{\alpha}\right\}$ is
linearly independent.
Then the {\it algebra induced by $H$ on $q$} is 
the algebra inside ${\cal L}({\cal H}_q)$ generated 
by $\left\{A_{\alpha}\right\}_{\alpha} \cup \left\{I\right\}$.
\end{definition}

\begin{fact}\label{fact:independent}
Given an operator $H(q)$, the induced algebra on $q$, 
${\cal A}_q^H$ is independent of our choice 
of how to write $H$ as a sum as in Equation \ref{eq:deq},
so long as the $B_{\alpha}$ operators are linearly independent.
\end{fact}
\begin{proof} 
Let us decompose $H(q)$ in two different ways:
$$ H = \sum_{\alpha} A_{\alpha} \otimes B_{\alpha} = 
\sum_{\beta} \hat{A}_{\beta} \otimes \hat{B}_{\beta} $$
where the sets $\left\{B_{\alpha}\right\}_{\alpha}$ and $\left\{\hat{B}_{\beta}\right\}_{\beta}$ are each linearly independent.
Let ${\cal A},{\cal \hat{A}}$ denote the $C^*$-algebra of $H$ on $q$ 
induced by the first and second decompositions, respectively. 
We show that ${\cal A}\subseteq {\cal \hat{A}}$, by symmetry 
this suffices to show that the two algebras are equal. 
Let us complete the set $B$ into a basis of the second subsystem, 
by the matrices $\{B_{\alpha'}\}_{\alpha'}$.  
We can thus write the $\hat{B}$ matrices in terms of this basis:  
$$ \hat{B}_{\beta} = \sum_{\alpha} \gamma_{\beta}^{\alpha} B_{\alpha}+ 
\sum_{\alpha'} \gamma_{\beta}^{\alpha'} B_{\alpha'} $$
with $\gamma_{\beta}^{\alpha}$, $\gamma_{\beta}^{\alpha'}$ complex 
numbers.  
Setting the equation above, in the second decomposition, we get:
$$ H = \sum_{\beta}  \hat{A}_{\beta} \otimes \left( 
\sum_{\alpha} \gamma_{\beta}^{\alpha} B_{\alpha} + \sum_{\alpha'} \gamma_{\beta}^{\alpha'} B_{\alpha'} \right) $$
which means that
$$ H = \sum_{\alpha} \left( \sum_{\beta} \gamma_{\beta}^{\alpha} 
\hat{A}_{\beta} \right) \otimes B_{\alpha}+ 
\sum_{\alpha'} \left( \sum_{\beta} \gamma_{\beta}^{\alpha'} \hat{A}_{\beta} \right) \otimes B_{\alpha'}. $$ 
We now recall that
the decomposition in terms of a basis of linearly independent 
matrices is unique. Comparing the last equation with  
 $$ H = \sum_{\alpha} A_{\alpha} \otimes B_{\alpha}+ \sum_{\alpha'} A_{\alpha'} 
\otimes B_{\alpha'}, $$ the matrix in front of each  
$B_{\alpha}$ must be the same. Hence, $A_\alpha$ is contained in the algebra 
${\cal A}$ spanned by $\hat{A}_{\beta}$. 
\end{proof} 

\begin{fact}\label{fact:closed} 
Given a Hermitian operator $H(q)$, the induced algebra on $q$
is closed under the adjoint operator. 
\end{fact} 

\begin{proof} 
We write 
$H= \sum_{\alpha} A_{\alpha} \otimes B_{\alpha}$ with $B_\alpha$ linearly independent; 
then the induced algebra is the one generated by 
$\left\{A_{\alpha}\right\}_{\alpha} \cup \left\{I\right\}$.
But since $H$ is Hermitian,  
$H=\sum_{\alpha} A^\dagger_{\alpha} \otimes B^\dagger_{\alpha}$ and
so the induced algebra is also the algebra generated by 
$\left\{A_{\alpha}^\dagger\right\}_{\alpha} \cup \left\{I\right\}$ by 
Fact \ref{fact:independent}. 
This means that the induced algebra also contains the adjoint of the 
generators, and hence is closed under the adjoint. 
\end{proof}  

A simple but crucial fact to this paper is that if two 
terms that intersect only on a qubit 
commute, then the two algebras they induce on $q$ commute: 

\begin{fact}\label{fact:algcomm}
Consider two Hamiltonian terms $H_{j,k}$ intersecting only on the qudit 
$j$. Then the algebras ${\cal A}_{j.k}$ induced by these operators on 
$j$ commute with each other. 
\end{fact}

\begin{proof}
Let $H_{j,1}$ and $H_{j,2}$ be two commuting operators that share only qudit $j$, and let their
decomposition be as follows:
$$ H_{j,1} = \sum_{\alpha} A_{\alpha}^1 \otimes B_{\alpha}^1, H_{j,2} = \sum_{\alpha} A_{\alpha}^2 \otimes B_{\alpha}^2 $$
where $A_{\alpha}^1$ and $A_{\alpha}^2$ are operators on qudit $j$,
and the sets $\left\{B_{\alpha}^1\right\}$ and $\left\{B_{\alpha}^2\right\}$ are each linearly independent, and act on different qudits.
We get:
$$ [H_{j,1},H_{j,2}] = \sum_{\alpha,\beta} [A_{\alpha}^1,A_{\beta}^2]\otimes B_{\alpha}^1 \otimes B_{\beta}^2 = 0. $$
By the linear independence of $\left\{B_{\alpha}^1\right\}$ and $\left\{B_{\alpha}^2\right\}$ we have that the set
$ \left\{B_{\alpha}^1 \otimes B_{\beta}^2\right\}_{\alpha,\beta} $
is also linearly independent, thus
$ [A_{\alpha}^1,A_{\beta}^2] = 0 $
for all $\alpha,\beta$, meaning that the algebras ${\cal A}_{j.1}$ and ${\cal A}_{j.2}$ commute.
\end{proof} 

\subsection{Representation theory of algebras} 
The proof of \cite{BV}, as well as our proofs,
rely on fundamental
facts from the representation theory of $C^*$ algebra.  
For the purposes of this paper, we restrict attention to 
algebras ${\cal A}\subseteq L({\cal H})$, such 
that ${\cal A}$ is closed under the $^\dagger$ operation.  
All the algebras we will consider are of this form.

\begin{definition} 

\textbf{Center of a $C^*$-algebra, $Z({\cal A})$.}

\noindent
The center of a sub-algebra ${\cal A}$ is defined to be 
the set of all operators in ${\cal A}$ which commute with all the elements in 
${\cal A}$. It is denoted by $Z({\cal A})$. 
\end{definition} 

\begin{definition}

\textbf{A reducible / irreducible $C^*$-algebra}

\noindent
An algebra ${\cal A}$ is said to be irreducible if its center is trivial, i.e. $\mathbf{Z}\left({\cal A}\right)= c\cdot I$, and otherwise it is reducible.
\end{definition}

\begin{fact}\label{fact1}
Let ${\cal A}$ be an irreducible subalgebra of $L({\cal H})$, i.e. $\mathbf{Z}\left({\cal A}\right) = c\cdot I$. 
Then ${\cal H}$ can be written as a tensor product of two subsystems 
${\cal H}_1\otimes {\cal H}_2$ such that 
$$ {\cal A} \approx L({\cal H}_1)\otimes I({\cal H}_2). $$
\end{fact} 

\noindent
A generalization of the previous fact is a well-known decomposition theorem from the representation theory
of $C^*$-algebras:
\begin{fact}\label{fact2}
Let ${\cal A}$ be a $C^*$-algebra on some Hilbert space ${\cal H}$.
Then, there exists a decomposition of ${\cal H}$ into a direct sum 
of orthogonal subspaces ${\cal H}_{\alpha}$, 
where each ${\cal H}_{\alpha}$ is a tensor product of two Hilbert spaces,  
${\cal H}_{\alpha}={\cal H}_{\alpha}^1\otimes {\cal H}_{\alpha}^2$ 
such that  
$$ 
{\cal A} \approx \bigoplus_{\alpha} \mathbf{L}\left({\cal H}_{\alpha}^1\right)
\otimes 
\mathbf{I}\left({\cal H}_{\alpha}^2 \right). 
$$

\noindent
The projections on the subspaces ${\cal H}_{\alpha}$
generate $\mathbf{Z}\left({\cal A}\right)$, 
and for each subspace ${\cal H}_\alpha$ the algebra ${\cal A}_{\alpha}$ 
(which is defined to be the algebra ${\cal A}$ restricted to ${\cal H}_\alpha$), 
is irreducible.
\end{fact}

\begin{fact}\label{fact:presdec}
Let ${\cal A}_1$ and ${\cal A}_2$ be two commuting algebras on a Hilbert space ${\cal H}$, and let
${\cal A}_1$ be decomposed by its center as in fact (\ref{fact2}) - i.e.
a decomposition ${\cal H}=\bigoplus_{\alpha} {\cal H}_{\alpha}$ such that
$ 
{\cal A}_1 \approx \bigoplus_{\alpha} \mathbf{L}\left({\cal H}_{\alpha}^1\right)
\otimes 
\mathbf{I}\left({\cal H}_{\alpha}^2 \right). 
$
Then ${\cal A}_2$ preserves each subspace ${\cal H}_{\alpha}$ in the decomposition above.
\end{fact}

\begin{proof}
For each subspace ${\cal H}_{\alpha}$ there exists a projection $\Pi_{\alpha}\in \mathbf{Z}({\cal A}_1)$ whose image is ${\cal H}_{\alpha}$.
Since ${\cal A}_1$ and ${\cal A}_2$ commute, then each projection 
$\Pi_{\alpha}$ commutes with all operators in ${\cal A}_2$. 
Thus, by Fact \ref{fact:trivial2}
 ${\cal A}_2$ preserves ${\cal H}_{\alpha}$ for all $\alpha$.
\end{proof}

\section{2-local CLH is in NP (Revised from \cite{BV})}\label{sec:2lcl}

The basic facts from the theory of representations of algebras 
presented in Section \ref{sec:notation} can be used to prove 
an important lemma, which is the basis of the 
proof of \cite{BV}. We start with the following definition and a preliminary claim: 

\begin{definition}\label{def:algdec}
Let $\{{\cal A}_j\}_{j=1}^k$ be $k$ mutually commuting algebras 
on some Hilbert space ${\cal H}$.
A separating decomposition is a direct-sum decomposition of ${\cal H}$:
$$ {\cal H} = \bigoplus_{\alpha} {\cal H}_{\alpha}$$
that is preserved by all algebras, such that 
inside each subspace ${\cal H}_{\alpha}$ 
there appears a tensor product structure
$${\cal H}_{\alpha} = 
{\cal H}_\alpha^0 \otimes
{\cal H}_{\alpha}^1 \otimes 
{\cal H}_{\alpha}^2\otimes \cdots \otimes {\cal H}_\alpha^k $$
such that
$$ {\cal A}_j|_{{\cal H}_{\alpha}} \approx 
I_{{\cal H}_{\alpha}^0} \otimes 
I_{{\cal H}_{\alpha}^1} \cdots \otimes I_{{\cal H}_{\alpha}^{j-1}} \otimes
\mathbf{L}({\cal H}_{\alpha}^j) \otimes I_{{\cal H}_{\alpha}^{j+1}}\cdots \otimes I_{{\cal H}_{\alpha}^{k}}. $$
\end{definition}

\begin{claim}\label{cl:algdec}
Let $\{{\cal A}_j\}_{j=1}^k$ be $k$ mutually commuting algebras 
on some Hilbert space ${\cal H}$.
There exists a separating decomposition of ${\cal H}$.
\end{claim}

\begin{proof}
Suppose that the algebra ${\cal A}_j$ are all
 irreducible algebras - i.e. have trivial centers.
Using fact (\ref{fact1}) we have that each ${\cal A}_j$ is isomorphic 
to the full set of linear operators on some 
subsystem of ${\cal H}$. In other words, ${\cal H}= {\cal H}^j 
\otimes {\cal H}^{rest}$ and 
$$ {\cal A}_j \approx \bigoplus_{\alpha} \mathbf{L}\left({\cal H}_{\alpha}^j\right)
\otimes 
\mathbf{I}\left({\cal H}_{\alpha}^{rest} \right). 
$$
Consider first ${\cal A}_1$ and ${\cal H}_\alpha^1$.  
Since the algebras commute, each ${\cal A}_j$ for $j>1$ must act 
as the identity on the ${\cal H}_{\alpha}^1$, 
and hence acts not trivially only on ${\cal H}_{\alpha}^{rest}$; 
We proceed by induction, to derive that each ${\cal A}_j$ is isomorphic 
to the full set of linear operators on a 
separate sub-particle and thus the lemma follows in this case. 
  
Now we generalize to the case where at least one algebra is reducible.
Let us examine the algebra ${\cal A}$ generated by the set of all 
operators on the particle ${\cal H}$, that commute with any $A\in {\cal A}_j$ 
for all $j$. 
It is easy to check
that since the ${\cal A}_j$ are closed under adjoint (by Fact 
\ref{fact:closed}) then so is ${\cal A}$. 
By fact (\ref{fact2}) algebra ${\cal A}$ admits a decomposition, 
such that inside 
each subspace, it is isomorphic to the full set of 
operators on some subsystem, tensor with identity.
By fact (\ref{fact:presdec}) this decomposition is preserved by 
all ${\cal A}_j$ since they commute with ${\cal A}$.

We consider the algebras  ${\cal A}_j$ 
restricted to these subspaces. We want to show that these restricted 
algebras are all irreducible.    
This follows since it turns out that the center of the algebra ${\cal A}$ 
in fact contains the centers of the algebras ${\cal A}_j$.
To show this, first notice that $\mathbf{Z}({\cal A}_j) \subseteq {\cal A}$ 
since any element of $\mathbf{Z}({\cal A}_j)$ commutes with any element of
 ${\cal A}_j$ by definition of a center of an algebra, 
and also commutes with any 
element of all the other algebras ${\cal A}_{j'}$ for $j\neq j'$ since 
it consists of elements from ${\cal A}_j$ and 
the algebras ${\cal A}_j$ and ${\cal A}_{j'}$ commute.  
In fact, since $\mathbf{Z}({\cal A}_j)$ commutes with all the generators of 
${\cal A}$, it is also contained in the center of ${\cal A}$. 
So $\mathbf{Z}({\cal A}_j) \subseteq \mathbf{Z}({\cal A})$.

Therefore, since ${\cal A}$ is irreducible inside each 
of the subspaces, so are ${\cal A}_j$.
So the decomposition of the algebra ${\cal A}$ 
decomposes each algebra ${\cal A}_j$,
into irreducible commuting algebras, which by the first paragraph 
must act on separate subsystems inside each subspace.
\end{proof}

We are now ready to prove the following crucial fact, which is the basis for 
the result of Bravyi and Vyalyi \cite{BV} as well as the current paper:  

\begin{lemma}\label{lem:BV} 

{\bf Decomposition of the center of a star (adapted from \cite{BV})}

\noindent
Let $S$ be an instance of $CLH(2,d)$ 
whose interaction graph is a star: 
this means that there is a particle $j$, and each $2$-local $H_{j,k}$ 
examines $j$ and another particle $k$ (where different terms act on 
different $k$'s). 
Then there exists a direct sum decomposition
\begin{equation}\label{eq:dec1}
{\cal H}^j = \bigoplus_{\alpha}{\cal H}_{\alpha}^j
\end{equation}
such that inside each subspace ${\cal H}_{\alpha}^j$ 
there appears a tensor product structure
\begin{equation}\label{eq:dec2}
{\cal H}_{\alpha}^j = {\cal H}_{\alpha}^{j.j} \otimes \left(\bigotimes_{(j,k)\in E}{\cal H}_{\alpha}^{j.k}\right)
\end{equation}
where $k$ runs over all other particles, 
such that 
all operators $H_{j,k}$ preserve the subspaces ${\cal H}_{\alpha}^j$, and moreover, 

\begin{equation}\label{eq:dec4}
H_{j,k}|_{{\cal H}_{\alpha}^j}
\in 
\bigotimes_{l\neq k} I_{{\cal H}_{\alpha}^{j.l}}
\otimes 
\mathbf{L}\left({\cal H}_{\alpha}^{j.k}\otimes {\cal H}^k \right)
\end{equation}
 
\end{lemma}

\begin{proof}
We write each Hamiltonian as a sum of tensor product terms.
$$ H_{j,k}=\sum_{\alpha} A_{\alpha}^k\otimes B_{\alpha}^k,$$
where $A^k_{\alpha}$ acts on ${\cal H}_j$ and $B^k_{\alpha}$ acts 
on ${\cal H}_k$, and
the operators $\left\{B^k_{\alpha}\right\}_\alpha$ 
are linearly independent.  
We consider the $C^*$-algebra generated by $\{A^k_{\alpha}\}_{\alpha} \cup \left\{I\right\}$, 
and denote it ${\cal A}_{j.k}$. 
The key point is that any pair of ${\cal A}_{j.k}$ algebras commute, 
due to Fact (\ref{fact:algcomm}). We can therefore apply claim 
(\ref{cl:algdec}) 
and this implies the result. 
\end{proof} 

The above lemma implies the following. 
When the interaction graph of the commuting Hamiltonian 
is a star (which is indeed the case for every qudit in the two-local 
case), the Hilbert space of the center particle ${\cal H}_j$ 
can be decomposed into a direct sum of spaces ${\cal H}_j^\alpha$, such 
that each of the terms in the Hamiltonian preserves those subspaces 
${\cal H}_j^\alpha$.   
This means that the original system has a zero eigenstate,
if and only if for every particle 
there exists some subspace index $\alpha_0$ such that 
the restriction of the terms in the Hamiltonians to
the subspaces ${\cal H}_j^{\alpha_0}$ is 
an instance which contains a zero eigenstate. 
Moreover, under this restriction, the terms in the Hamiltonian 
become {\emph disjoint}, namely, they can be described as working on 
separate particles, described by the subspaces ${\cal H}_{j,k}^\alpha$.
Thus, the interaction graph of the restricted Hamiltonian is simply a 
set of disjoint edges.  

\subsection{A proof of the two local case}
We can now prove the result of \cite{BV} stating 
that $CLH(2,d)$ is in $NP$ for any $d$. 
We present a slightly modified proof here, since a similar approach 
will be useful when we generalize the result to $CLH(3,2)$ and 
$CLH(3,3)$.

The main point is that in the two local case, for any qudit, the 
interactions involving that qudit form a star.
Lemma \ref{lem:BV} can therefore be applied. 
Merlin helps Arthur find the ground state by providing him with the 
correct index $\alpha$ in the decomposition of each particle. 

To present the proof in more detail, we use an interactive picture: a
communication protocol between Merlin and Arthur. 
In our protocol, only Merlin sends messages to Arthur, so in fact 
he can send all messages at once and this concatenation of 
messages can be viewed as the witness; but the protocol point of view 
is more convenient for our purposes. 

\begin{algorithm}
\item Input: $S$, an instance of $CLH(2,d)$.
\item Repeat until there is no node in the interaction graph 
whose degree is $2$ or more. 
\begin{enumerate}
\item 
Merlin picks a vertex $q$, and sends Arthur the description of one
subspace ${\cal H}_q^{\alpha}$ from the direct sum decomposition 
of  ${\cal H}_q$ given by Lemma \ref{lem:BV}. 
Both Merlin and Arthur can generate the decomposition so Merlin 
only needs to send Arthur the index $\alpha$.
\item 
Both Arthur and Merlin slice the qudit $q$ after restricting it
to ${\cal H}_q^{\alpha}$, according to Lemma \ref{lem:BV}. 
Accordingly, the node $q$ is replaced by at most $N_S(q)+1$ new nodes, 
each with a degree $1$ in the new interaction graph. 
\end{enumerate}

After all particles with degree more than $1$ have been removed, 
Arthur verifies that the Hamiltonian, which is now a set of non-intersecting 
terms (each term corresponds to a disjoint edge)
has a nonzero kernel. 
\end{algorithm}

\section{The three-local case for qubits}\label{sec:3lcl2}
In this section we prove Theorem \ref{MainClaim}. 
We start with a more detailed overview of the proof. 

\subsection{Proof overview} 
As mentioned before, the first step in the proof is to use 
the tools of \cite{BV} to identify and remove qubits 
that are ``separable'', namely, for which there is a 
decomposition to a direct sum of subspaces, such that 
all operators acting on the qubit preserve those subspaces. 
In the case of qubits, when a non-trivial decomposition exists, it must 
be into two subspaces of dimension one each; when restricting to one 
such subspace, the state of the qubit becomes 
some tensor product state with the rest of the system. 
This means that those qubits can in fact be removed from the system since 
Merlin can provide their state separately. 
We have thus reduced the problem to a problem in which all qubits 
are non-separable; 
This is done in Subsection (\ref{subsec:removingsep}). 

We now embark on the most important component in the proof, 
which is the characterization of the geometric properties of 
the interaction graph, after
the removal of separable qubits. We treat each connected component 
separately, so we may assume the graph is connected. 
We ask, how can we constrain the interactions of a qubit $q$, for which we know that no
single decomposition exists which all operators on $q$ agree on.

We define a $\bfly$ (butterfly) relation between two operators acting on the same qubit $q$, 
$H_1(q)$, $H_2(q)$ if $A_1\cap A_2 = \left\{q\right\}$.
We denote this by $H_1 \bfly H_2$. 
Each $\bfly$ relations yields a direct-sum decomposition 
by $(\ref{lem:BV})$ which is preserved by both operators.
As mentioned in the introduction, 
we notice as a first step, that if there are two butterflies 
with respect to $q$, 
$H_1(q)\bfly H_2(q)$ and $H_1(q)\bfly H_3(q)$, then due to the 
low dimensionality, the decompositions induced by 
both $\bfly$ relations are the same (see Claim \ref{cl:same}).  
 
This yields an important transitivity conclusion: i.e., if $H_1(q)$ and $H_2(q)$ agree on some decomposition of $q$, and $H_1(q)$ and $H_3(q)$ agree on some decomposition of $q$, then $H_2(q)$, and $H_3(q)$ agree 
on the same decomposition. We can now talk about 
two operators on $q$ which are connected by a path of such 
butterflies: two operators 
are said to be $\bfly$ connected (read this ``butterfly-connected'') 
if there is a sequence of $\bfly$ relations that connects 
them, i.e. $H_1 \bfly H_{i,1} \hdots \bfly H_{i,m} \bfly H_2$.
A basic tool in this paper 
is theorem (\ref{bflychainClaim}) proved in Section (\ref{subsec:sep}):  
\begin{theorem}\label{bflychainClaim0}
If any two operators $H_1(q),H_2(q)$ acting on a qubit $q$ are butterfly 
connected, then $q$ is separable.
\end{theorem}

This implies that an interaction graph made 
of non-separable qubits is severely limited, since its 
operators cannot be all connected by butterfly paths. 
The operators on any qubit thus cannot "fan-out" too much, 
as this would induce pairwise $\bfly$ paths and would make this 
qubit separable. 

In Sections (\ref{subsec:crowns}) 
and (\ref{subsec:limit}) 
we make this intuition more tangible, and show two important 
conclusions from Theorem (\ref{bflychainClaim}).
First, we define an ``operator crown'' on $q$, which is a set of 
three operators acting on $q$ organized as in Figure \ref{fig:crown}. 

\begin{figure}[ht]
\center{
 \epsfxsize=3in
 \epsfbox{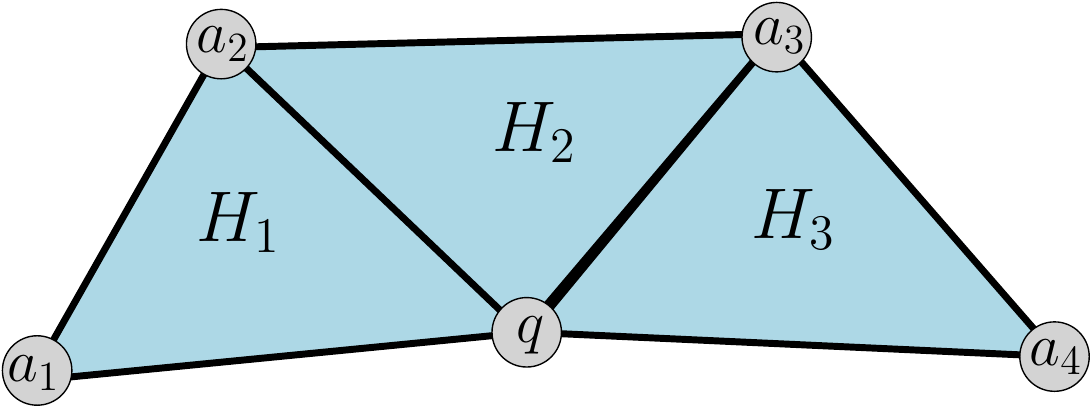}}
 \caption{An operator crown on qubit $q$.}
 \label{fig:crown}
\end{figure}

We show in Claim (\ref{cl:nosep1}) that operator crowns, 
act as "qubit traps"; This means that if $q$ has such an operator 
crown, then any operator on $q$ must also act on 
at least one other qubit of the crown. 
Second, we show in Claim (\ref{cl:nosep2}) that any 
two operators on $q$ must either intersect on one other qubit than 
$q$,  or they are connected through another operator $H_x$, 
which intersects each of them with $q$ and another qubit. 
Another claim which we call the 
Bridge claim (Claim \ref{cl:bridge}) is a slightly strengthened version of 
Claim (\ref{cl:nosep2}). 
These latter properties impose severe restrictions on the 
geometry of the interaction graph.  

Having characterized the local geometric behavior of each individual 
nonseparable qubit in the residual graph, we are ready to take one step 
further, and 
make some claims w.r.t. the global structure of the residual graph.
To this end, we define the "backbone" of the graph: 
this is 
the longest ``path of operators'' in the residual graph. 
 
\begin{figure}[ht]
\center{
 \epsfxsize=3in
 \epsfbox{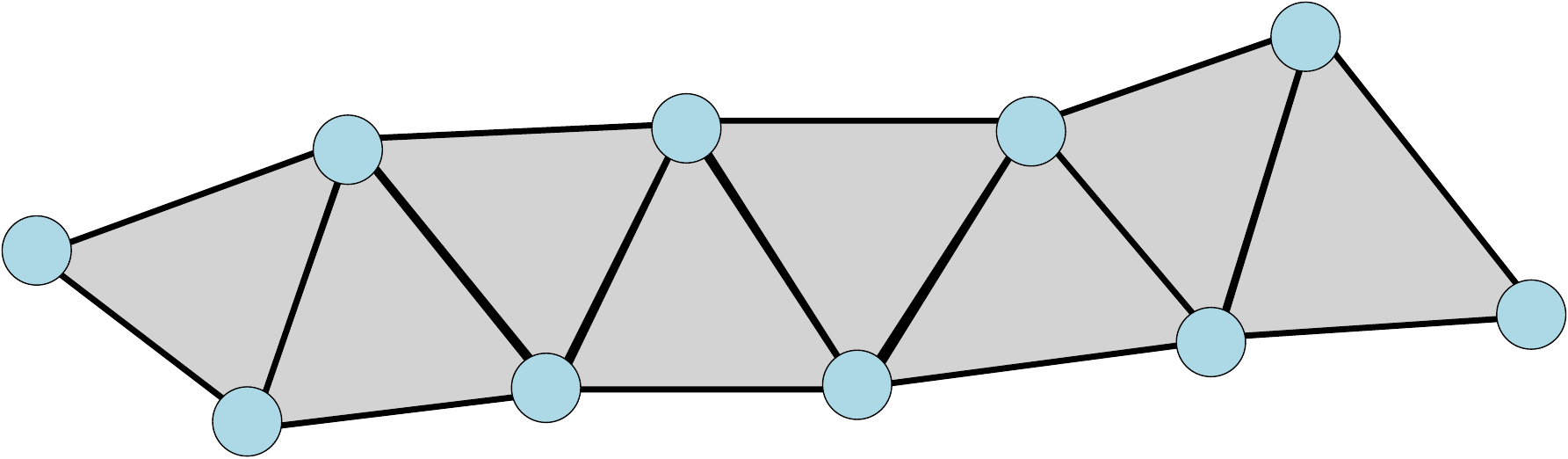}}
\end{figure}

Intuitively, the backbone constitutes the longest possible stretch 
of "operator crowns" that are attached back to back, without 
revisiting qubits that have already been visited. 
Recall that by Claim \ref{cl:nosep1} 
an operator crown on qubit $q$ essentially "traps" at 
least one other qubit of any operator acting on $q$.
This means that essentially any operator that acts on a backbone qubit, 
must act on at least one more backbone qubit which is not 
very "far" in terms of backbone edges; this is captured by Lemma 
\ref{lem:entrap} below. 

This "qubit" entrapment property alone is not sufficient for our purposes, 
as it does not handle operators that do not act on backbone qubits.
We want to show that \emph{all} operators (if we started with an instance 
whose interaction graph is connected) must examine at least 
two backbone qubits, and therefore these qubits must be close by the above 
Lemma \ref{lem:entrap}. 
 Moreover, we want to show that there are no 
``shortcuts'' between far away qubits in the backbone, 
through interactions with qubits outside the backbone: 
consider any qubit outside the backbone which interacts with 
two backbone qubits, through two different terms in the Hamiltonian. 
We want to show that even these two qubits cannot be 
too distant in terms of number of backbone edges. 
Those two properties are proved in claims (\ref{cl:bkbn1}) 
and (\ref{cl:bkbn2}) using the geometric claims above. 

The result of all this is the following. 
Consider a coarse-graining of the backbone,
in which say consecutive sets of $20$ qubits are aggregated together 
and are considered as one particle of constant dimension; 
denote those by $Q_i$.   
By the above arguments, all interactions inside the 
backbone are two-local, namely, interact only $Q_i$ and $Q_{i+1}$; 
and moreover, any qubit outside the backbone may interact only with
a specific pair of consecutive large particles $Q_i, Q_{i+1}.$

\begin{figure}[ht]
\center{
 \epsfxsize=4in
 \epsfbox{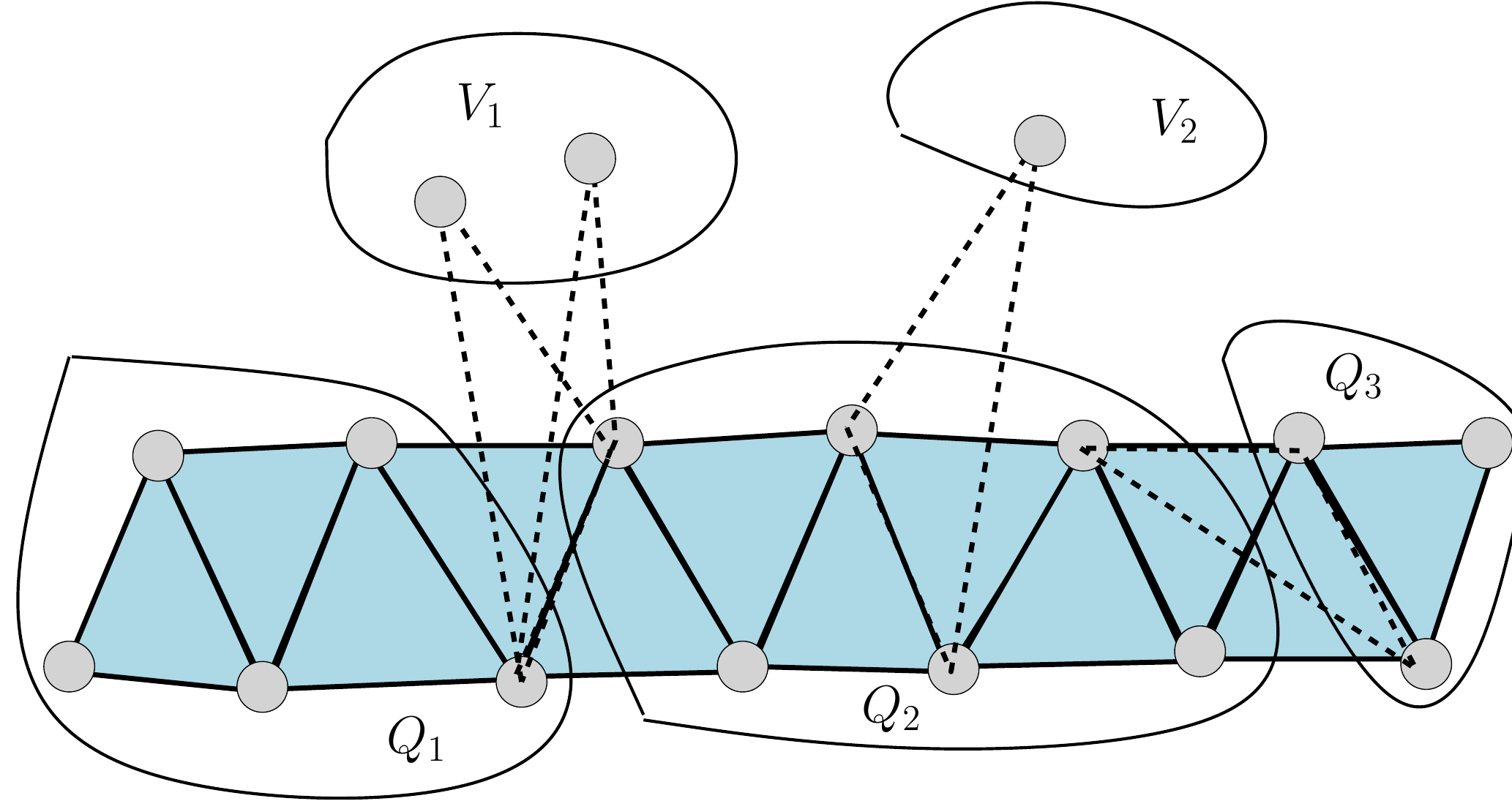}}
 \caption{\label{fig:step2}The interactions with the backbone:
a backbone of sets $Q_i$ of constantly many qubits, 
such that each operator 
acts on $V_i$ and its associated pair of qudits $Q_i,Q_{i+1}$.  
We note that while the size of $Q_i$ is constant, the size of $V_i$ can 
be a function of $n$.}
 
\end{figure}

We examine this structured problem more closely. Consider the operators
 interacting $Q_i$ with the qudit to its left, $Q_{i-1}$. Consider 
also the operators that interact $Q_i$ with the qudit to its right, 
$Q_{i+1}$. We have a $\bfly$ relation between any operator 
acting on $Q_i$ from the left and any operator acting on $Q_i$ from the right.
We can then show, very similarly to Lemma (\ref{lem:BV}), 
that there exists a decomposition of the Hilbert space of $Q_i$, such that when we restrict all operators on $Q_i$ 
to a specific subspace in this decomposition, then $Q_i$ can be written as a 
tensor product of two subparticles, the left subparticle $Q_{i,left}$ and the
right subparticle $Q_{i,right}$, and the operators acting on $Q_i$ from the left (right) 
interact only with the left (right) subparticle of $Q_i$. 
This paves the way for achieving two-locality:   
After partitioning each $Q_i$ into those two separate subparticles 
$Q_{i,left}$ and $Q_{i,right}$,  we can 
fuse the right side of one particle with the left side of the next:
 $Q_{i,right}$ with $Q_{i+1,left}$. 
 
The resulting problem is two-local:
all interactions are of the form in which one fused particle and one
particle out of the backbone interact, or they are $1$-local; hence,
we get that each fused particle is a center of a star, and the stars
are non-intersecting.
This is already a problem in NP by Lemma \ref{lem:BV}, namely, by
the methods of Bravyi and Vyalyi \cite{BV}.

We now provide the details. 
\subsection{Removing Separable Qubits}\label{subsec:removingsep}
We start by defining separable qubits, and explaining why they can 
be removed from the graph, using classical witness. 
By the end of this section, we will have removed all those qubits 
and remain with the residual problem in which all qubits are non-separable. 
We start by defining:  
\begin{definition}

\textbf{Separable qubit}

\noindent
A qubit $q$ is said to be separable if there exists a
 direct-sum decomposition of its Hilbert space 
to two one dimensional spaces, 
$$ {\cal H}_q = \bigoplus_{\alpha\in \{0,1\}}{\cal H}_q^{\alpha} $$ 
such that any operator $H(q)$  which acts 
on $q$ 
preserves this decomposition:
$$ H(q) = \bigoplus_{\alpha} H(q)|_{{\cal H}_q^{\alpha}}, $$
where $H(q)|_{{\cal H}_q^{\alpha}}$ is the restricted projector. Observe that 
the restricted projection in this case    
is also a projection. 
\end{definition}

This definition is aimed to capture a similar situation to what 
happens in the case of $CLH(2,d)$ when a particle is handled, according to lemma \ref{lem:BV}.  
Note that the fact that the above direct sum decomposition is into 
subspaces of dimension $1$, meaning that knowing the index $\alpha$, 
a separable qubit has a well defined quantum state; 
Since the verifier will receive the index $\alpha$ from the prover, 
this qubit can simply be removed from the graph. 

So, the first step in the protocol between Merlin and Arthur is as follows: 
\begin{algorithm}\label{alg:restrict}
{\bf Restrict Graph ($S$)}
\item Input: $S$, an instance of $CLH(3,2)$.
\item Iteratively change $G_S$ until there is no separable qubit in $G_S$:
\begin{enumerate}
\item 
Merlin picks a separable qubit $q$ in $S$ and
 sends Arthur an index $\alpha$ whose corresponding subspace 
contains a common zero eigenspace of $S$. 

\item Arthur restricts all terms in the Hamiltonian 
to this subspace, and removes the vertex $q$ and all 
its incident edges in $G_S$.
If after this restriction, some term acts 
trivially on one or more of its particles,
Both Arthur and Merlin replace this term (or those terms if there 
are more than one) by 
the appropriate $1$-qubit terms in the Hamiltonian.
\end{enumerate}
\end{algorithm}

The resulting instance is one without separable qubits, 
and where all terms act non-trivially on all their particles. 

\begin{claim}\label{Step1}
Given a $CLH(3,2)$ instance $S$, algorithm (\ref{alg:restrict})
generates a $CLH(3,2)$ instance $S_{nosep}$ with no separable qubits; 
If $S$ has a non-trivial common groundspace then Merlin can choose his restrictions so that $S_{nosep}$ has 
such a subspace. If $S$ has no zero eigenspace, for any choice of subspace restriction 
 the ground energy of $S_{nosep}$ is at least $1$.  
\end{claim}

\begin{proof}
Let $q$ be a separable qubit which Merlin removes. 
There exists a basis of ${\cal H}_q$, 
such that any Hamiltonian term $H_i$ is block diagonal in this basis.
$S$ has a ground energy zero if and only if one of those subspaces contains 
such a state. 
It remains to show that the new instance is still a legal $CLH(3,2)$ instance.
Indeed, if two operators
 $H_1 \in \mathbf{L}({\cal H})$ and $H_2 \in \mathbf{L}({\cal H})$ commute 
and are block diagonal in some basis of $q$, 
then the restricted operators to 
any $1$-dimensional subspace spanned by this basis also commute.
This logic is preserved, at every iteration, leading to finer and finer 
slicing of the original Hilbert space, until we exhaust all separable qubits.
\end{proof}

We note that during the process of ``slicing'' separable qubits, 
qubits which were previously nonseparable may become separable, 
yet a qubit cannot be sliced twice. 

\subsection{Geometric Constraints on the Residual Graph}\label{subsec:geom}
Here we define Butterflies, operator crowns, operator paths, 
and provide all sorts of geometric restrictions on the interactions 
between separable qubits using those notions. 
\subsubsection{Butterflies and Separability}\label{subsec:sep}
\begin{definition}

\textbf{Butterfly}

\noindent
Consider two operators $H_1 = H_1(A,B)$ and $H_2 = H_2(B,C)$, where 
$B$ is a qubit, 
and $A$ and $C$ are sets of qubits, such that $A$ does not intersect $C$. 
We say that $H_1$ and $H_2$ constitute a butterfly and denote $H_1 \bfly H_2$.
A butterfly is always with respect to the particle in the intersection, 
here $B$; most of the time the identity of the qubit $B$ 
will be clear from the context and we will 
omit specifying it. 
\end{definition}

\begin{claim}\label{cl:bfly}  
For any pair of operators acting non trivially 
on $q$, $H_1$ and $H_2$, 
with $H_1\bfly H_2$ with respect to $q$, there exists a non-trivial 
decomposition of ${\cal H}_q$ into a sum of two 
one dimensional subspaces, each of dimension one, 
which are preserved by both operators. 
\end{claim} 
\begin{proof}  
Denote by  $A$ the set of qubits which $H_1$ acts upon, excluding 
$q$. Likewise, denote by $B$ the set of qubits which $H_2$ acts on, 
excluding $q$. By the definition of the $\bfly$ relation, 
we have $A\cap B = \Phi$. 
We can consider all operators in $A$ as one  qudit. 
Similarly, we can consider all qubits in $B$ as another qudit. 
We can then apply lemma
\ref{lem:BV} and conclude that there exists a
direct-sum decomposition of $q$ that is 
preserved by both operators. 
The reason the decomposition of lemma (\ref{lem:BV}) must be 
non-trivial is that otherwise (namely, a decomposition 
to a sum of zero and two-dimensional spaces), since $dim(q)=2$, it means 
that one of the operators acts trivially on $q$, contradicting to our assumption. 
\end{proof} 

\begin{definition} 
{\bf Decomposition induced by the butterfly} 
By Claim \ref{cl:bfly}, a butterfly induces a well defined decomposition 
on its center qubit, which is called the decomposition induced by 
the butterfly
\end{definition} 

Now, what if there are several butterflies in which one qubit participates? 
Could it be that the decompositions induced on $q$ by different 
butterflies are different? The following simple claim says that 
the answer is negative. This follows from the limited dimensionality 
of the qubit. 
This clearly leads to strong transitivity relations, as 
we will soon see, but let us first state and prove the uniqueness 
of induced decomposition:

\begin{claim}\label{cl:same} 
{\bf Unique butterfly induced decomposition of $q$}
Consider two butterflies $H_1\bfly H_2$, $H_1 \bfly H_3$, both 
with respect to $q$, where all three operators act non-trivially on $q$. 
Then the decompositions induced on $q$ from both butterflies, using 
Claim \ref{cl:bfly}, must be the same.  
\end{claim}

\begin{proof}  
As in Claim (\ref{cl:bfly}), $q$ can be decomposed into a direct sum 
of two one dimensional subspaces, based on the first butterfly $H_1 \bfly H_2$.
Let 
$\Pi_q^0, \Pi_q^1$ be the projections on those subspaces of ${\cal H}_q$, so 
$\Pi_q^0 + \Pi_q^1 = I$.  
We can write 
\begin{equation}\label{eq:1}
H_1 = \Pi_q^0 \otimes \Pi_{A}^0 +  \Pi_q^1 \otimes \Pi_{A}^1 
\end{equation} 
where $\Pi_{A}^i$ are some projections on $A$.  
The operators on $A$ are projections by lemma (\ref{lem:BV}). 

\noindent
We claim 
that if another $\bfly$ relation with $H_1$ 
results in a different decomposition of $H_1$
\begin{equation}\label{eq2}
H_1 = \tilde{\Pi}_q^0 \otimes \tilde{\Pi}_{A}^0 +  \tilde{\Pi}_q^1 \otimes \tilde{\Pi}_{A}^1 
\end{equation} 
then this would imply a contradiction. 
To see this, suppose WLOG that  $\Pi_q^0$ ($\Pi_q^1$)
 projects on the state $\ket{0}$ ($\ket{1}$).   
And suppose $\tilde{\Pi}_q^0$ projects on the state
$\alpha \ket{0} + \beta \ket{1}$ with neither $\alpha$ not $\beta$ equal 
to $0$. 
Since $q$ is a qubit, we can write WLOG that $\tilde{\Pi}_q^1$ 
 projects on the state
$\beta^* \ket{0} - \alpha^* \ket{1}$. 

\noindent
Then we write:
$$ \tilde{\Pi}_q^0 = \abs{\alpha}^2 \ketbra{0} + \abs{\beta}^2 \ketbra{1}
 + \alpha \beta^* \ket{0}\bra{1} + \alpha^* \beta \ket{1}\bra{0} $$
 and 
$$ \tilde{\Pi}_q^1 = \abs{\beta}^2 \ketbra{0} + \abs{\alpha}^2 \ketbra{1}
 - \alpha \beta^* \ket{0}\bra{1} - \alpha^* \beta \ket{1}\bra{0}. $$
Plugging these terms in Equation (\ref{eq2}) for $H_1$, and
using the fact that $H_1$ is block diagonal in the computational basis,
by Equation (\ref{eq:1}), we set
to $0$ the terms in tensor with $\ket{0}\bra{1}$ and $\ket{1}\bra{0}$,
and get
$$ \tilde{\Pi}_A^0 = \tilde{\Pi}_A^1 $$
which by Equation (\ref{eq2}) means that $H_1$ can be written as 
$H_1 = I_q \otimes \tilde{\Pi}_A$.
Thus, $H_1$ acts trivially on $q$, contrary to our assumption. 
\end{proof} 

We would like now to deduce various properties from transitivity. 
We first define: 

\begin{definition}

\textbf{Butterfly path, Butterfly-connectedness}

\noindent
Consider two operators both acting on 
a qubit $q$, denoted $H_a$ and $H_b$. 
We say there is a butterfly path between $H_a$ and $H_b$ if there is a 
sequence of operators $H_a,H_1,H_2,...,H_m,H_b$ such that 
$H_a\bfly H_1, H_1\bfly H_2, H_2\bfly H_3,......,H_m\bfly H_b$. 
(Where all butterflies are with respect to $q$). 
We say that these two operators are butterfly-connected. 
\end{definition}

\begin{theorem}\label{bflychainClaim}{\bf Butterfly connectedness of 
operators implies separability}
If all pairs of operators $H_1(q),H_2(q)$ acting on a qubit $q$, 
are connected by a butterfly path on $q$, then 
$q$ is separable.
\end{theorem} 

\begin{proof}
Pick one operator acting on $q$, and now use claim (\ref{cl:same})
along the path connecting it to any other operator on $q$, to show that 
by transitivity all butterflies along the path induce the 
same decomposition on $q$, 
and thus by transitivity all operators on $q$ preserve this decomposition, 
hence by definition $q$ is separable.  
\end{proof}

\begin{corollary}\label{PartitionClaim} {\bf ``left-right'' 
Partition implies Separability:} 
If there is a partition of the operator set $S$ 
into two disjoint non-empty sets $S_{q,left}$
and $S_{q,right}$ such that for each $H_i\in S_{q,left}$ and
 $H_j\in S_{q,right}$ we have
$A_i \cap A_j \subseteq \left\{q\right\}$ 
(we call this a ``left-right'' partition), then $q$ is separable. 
\end{corollary} 

\begin{proof}
For any pair of operators $H_1$ and $H_2$ one can construct a chain of $\bfly$ relations
$H_1 \bfly H_{j_1} \bfly \hdots\bfly H_{j_m} \bfly H_2$ that goes back and forth between
$S_{q,left}$ and $S_{q,right}$ as both sets are nonempty.
Then by theorem (\ref{bflychainClaim}) we get that $q$ is indeed separable.
\end{proof}

What structures then can be present in a graph which contains only 
non-separable qubits? For that, we need to define one more 
notion of connectivity, namely operator-path connectivity. 

\subsubsection{Operator Paths}\label{subsec:limit}
\begin{definition}

\textbf{Open and Cyclic Operator Paths}

\noindent
An open operator path is an ordered set of $L$ 
distinct operators $H_1,\hdots,H_L$ such that for any pair of indices $(i,k)$
where $i\in [L],k\in [L]$ we have:
\begin{enumerate}
\item
$|A_i\cap A_{k}|=2$ for $|i-k|=1$
\item
$|A_i\cap A_{k}|=1$ for $|i-k|=2$.
\item
$|A_i\cap A_{k}|=0$ for $|i-k|>2$
\end{enumerate}
Similarly, a closed operator path, is an ordered set 
of operators that are distinct, except $H_1 = H_L$,
and the index additions above are taken modulo $L$.
The length of an operator path is defined as the number of its distinct 
operators.
\end{definition}

\begin{claim}\label{connected}
{\bf Graph Connectivity implies operator path connectivity} 

\noindent
If $S$ is a set of operators such that no qubit in $A_S$ is separable, 
and $G_S$ is connected, then $A_S$ is also operator-path-connected, i.e., 
any pair of qubits $q,v \in A_S$, are connected by an operator path 
which starts with an operator which acts on $q$ and ends with an operator 
which acts on $v$. 
\end{claim}

\begin{proof}
Let $q\in A_S$. 
Consider the set $S_q$ built as follows. 
Start with all operators acting on $q$. Add to $S_q$ any operator which 
intersects an operator in $S_q$ with an intersection of size $2$. 
Continue until it is impossible to add operators this way. $S_q$ is the final 
set of operators we get. Now, if 
$A_{S_q}$ contains  $v$, we are done (since clearly we can construct 
an operator path as desired). If not, consider a path from $q$ to $v$ 
is $G_S$. Let $w$ be the last qubit that belongs to $A_{S_q}$ 
along that path. 
We claim that $w$ is separable by corollary (\ref{PartitionClaim}) as follows:
Let us partition the operators acting on $w$ to two non-empty 
sets: The set $A$ 
which contains all operators in $S_q$, and $B$ of all other operators. 
We know that any operator in $B$ cannot intersect any operator in $A$ 
by $2$ particles, since otherwise it would have been in $S_q$; 
Also, we know that $B$ is non-empty since it contains 
the operator inducing the edge from $w$ to the next qubit on the path 
to $v$ in $G_S$, and $A$ is also not empty, since $w$ belongs to $A_{S_q}$.   
Hence, the conditions of corollary (\ref{PartitionClaim}) apply. 
\end{proof}

\subsubsection{Operator Crowns and 
Geometrically Constrained Connectivity}\label{subsec:crowns}

We are now ready to deduce various restrictions on the 
connectivity of the operators acting on a nonseparable qubit $q$.
We will often restrict attention to operator paths all of whose operators 
act on one particular qubit: 
\begin{definition}

\textbf{Operator Path on a qubit}

\noindent
An operator path on $q$ (where $q$ is a qubit) is an operator path all of whose operators act on $q$.
\end{definition}

A certain structure which will appear useful is the 
length-$3$ operator path on a qubit $q$:  

\begin{definition}

\textbf{Operator Crown}

\noindent
For a qubit $q$, an operator Crown on $q$, is a set of $3$ 
operators that act on $q$, and $4$ other distinct
qubits, which we call the crown qubits $a_1,a_2,a_3,a_4$, as follows:
$H_1(a_1,a_2,q)$,$H_2(a_2,q,a_3)$,$H_3(a_3,q,a_4)$. 
\end{definition}

See Figure \ref{fig:crown} for an illustration. 
This structure acts as an ``operator-trap'': any operator
that acts on $q$ must act also on some qubit which participates in the crown:   

\begin{claim}\label{cl:nosep1}{\bf Operator Crown as an Operator Trap:}
Let $q$ be a nonseparable qubit, and let $C$ be some operator crown on $q$. 
Then any operator on $q$ acts on some crown qubit of $C$.
\end{claim}
\begin{proof}
Suppose on the negative, that there exists some crown $C = (H_1,H_2,H_3)$ and an operator $H$, such that $C$ and $H$ intersect only on $q$.
Then there exists a subgraph on $q$ made of $3$ operators, namely $H_1,H_3,H$ such that each pair intersects only on $q$.
Thus $q$ is separable by the following argument: we prove $(*)$, 
that any operator on $q$ has a $\bfly$ path to $H_1$.
This implies that any two operators are $\bfly$-path connected, 
so by theorem (\ref{bflychainClaim}) $q$ is separable.
We now show that indeed $(*)$ is true: 
the operators $H$, and $H_3$ are $\bfly$-path connected to $H_1$ as $H\bfly H_1, H_3\bfly H_1$.
Regarding $H_2$, we have $H_2\bfly H \bfly H_1$.
Any other operator on $q$ can share two qubits with at most $2$ of the operators: $H_1, H, H_3$, so it has a $\bfly$ path to at least one of them, which implies it is also $\bfly$ connected to $H_1$. 
\end{proof}

\begin{figure}[ht]
\center{
 \epsfxsize=3in
 \epsfbox{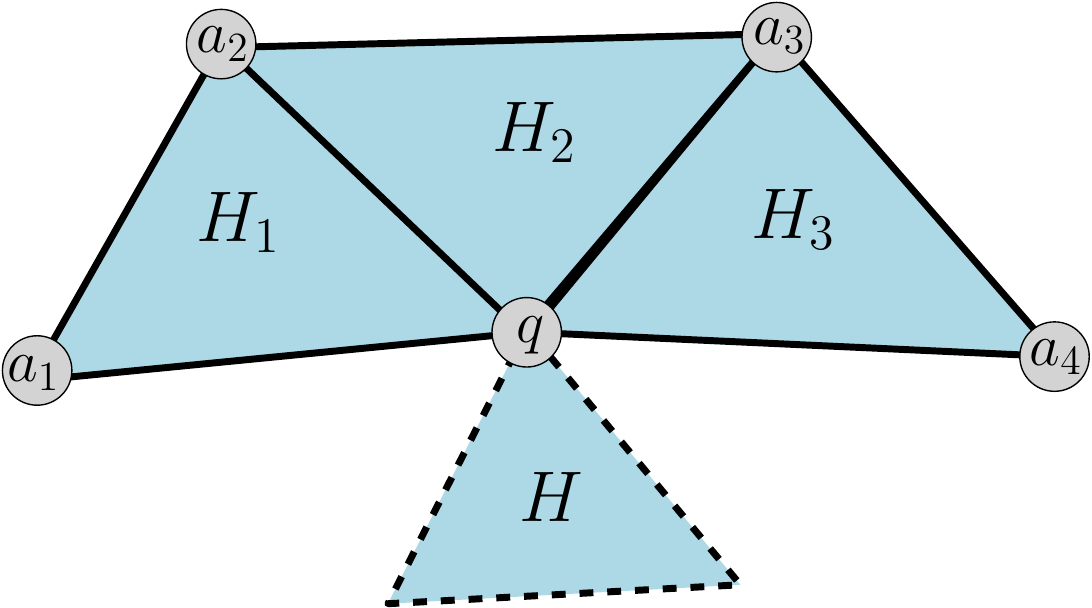}}
 \caption{A sketch of the proof of Claim (\ref{cl:nosep1})}
 \label{fig:crown2}
\end{figure}

We now use Theorem (\ref{bflychainClaim}) and Claim (\ref{connected}) 
to prove that the operators acting on one non-separable qubit $q$ 
must be close in terms of the shortest operator path connecting them: 

\begin{claim}\label{cl:nosep2}{\bf Operators on $q$ are connected by 
length-three operator paths:}
Let $q$ be some nonseparable qubit. 
Then any $2$ operators on $q$ are operator-path connected by an operator path 
on $q$ of length at most $3$.
\end{claim}

\begin{proof}
First, consider the connected component in $G_S$ which contains $q$.
By Claim (\ref{connected}) it is also operator path-connected. 
Assume on the negative, that there exist two operators  $H_1(q),H_4(q)$ such that any operator path on $q$ that connects them is of length at least $4$.
Since $H_1\neq H_4$ this must be an open operator path.
We know that any open operator path on $q$ must be of length exactly $4$, 
since if the shortest open path is of length at least $5$,
we get a structure as in Figure \ref{fig:crown2} and we contradict 
Claim (\ref{cl:nosep1}).
Let us choose such an operator path on $q$ and denote its operators as follows: 
$H_1(q,a_1,a_2)$,
$H_2(q,a_2,a_3)$,
$H_3(q,a_3,a_4)$,
$H_4(q,a_4,a_5)$.
Then any operator $H$ that examines $q$ and some qubit
 of the set $\left\{a_1,a_2\right\}$ cannot examine any qubit in the set $\left\{a_4,a_5\right\}$
as this would shorten the path between $H_1$ and $H_4$ to $(H_1,H,H_4)$.
So let $H(q)$ be some operator on $q$.
We claim that it has a $\bfly$ path to $H_1$.
Indeed, if $H$ does not examine any qubit in the set $\left\{a_1,a_2\right\}$, then it has a $\bfly$ relation with $H_1$.
Otherwise, it does not examine any qubit in the set $\left\{a_4,a_5\right\}$, so it has a $\bfly$ relation with $H_4$, yet $H_4\bfly H_1$, so $H \bfly H_1$.
Any of the operators $H_2,H_3,H_4$ is also $\bfly$ connected to $H_1$, thus by 
Theorem (\ref{bflychainClaim}) $q$ is separable.
\end{proof}

\noindent
We also give a slightly strengthened version of the above.
We examine a case where there are two operators $H_1(q),H_2(q)$ on a nonseparable qubit $q$ connected by
an operator path $P$ of length $4$ on $q$.  
By Claim \ref{cl:nosep2} this 
is not the minimal length operator path connecting these two operators.
We show that in that case, not only is there a shorter path of length $3$ 
between $H_1(q)$ and $H_2(q)$, but there exists a path of length $3$ that shortcuts $P$ itself as follows:

\begin{claim}\label{cl:bridge}{\bf (Bridge Claim)}
Let 
$H_1(q,a_1,a_2)$,
$H_2(q,a_2,a_3)$,
$H_3(q,a_3,a_4)$,
$H_4(q,a_4,a_5)$,
be a length $4$ open operator path on a nonseparable qubit $q$.
Then there exists an operator $H(q,a_2,a_4)$.
\end{claim}

\begin{proof}
We assume on the negative that no such operator exists, and show that $q$ is separable.
All operators of the path are $\bfly$ connected to $H_1$.
Any operator that shares only $q$ with $H_1,\hdots H_4$ is $\bfly$ connected to $H_1$.
Any operator that shares $q$ and just one other qubit out of $a_1,...,a_4$
is also $\bfly$ connected to $H_1$.
Let $H$ be an operator that acts on $q$ and two other qubits out of $a_1,...,a_4$.
These qubits cannot be adjacent so the possible pairs (excluding ($a_2,a_4$)) are 
$(a_1,a_3),(a_1,a_4),(a_1,a_5),(a_2,a_5)(a_3,a_5)$.
Each such pair of qubits examined by $H$, 
corresponds to a $\bfly$ relation of $H$ with the operator
 $H_4,H_2,H_3,H_3,H_1$, respectively, so $H$ is $\bfly$-connected to $H_1$.
Thus $q$ is separable by Theorem (\ref{bflychainClaim}). 
\end{proof}

\subsection{The Backbone}\label{subsec:bkbn}

Having arrived from the initial input $S$ to one with no separable qubits, using input from Merlin, we now show that with the help of Merlin, 
the new instance can be viewed as 
a $2$-local problem.
This is formalized by the following theorem:
\begin{theorem}\label{thm:2lclreduce}
For any instance $S$ of 
$CLH(3,2)$ with no separable qubits 
there exists a partition of the set of vertices of $G_S$ into disjoint sets:
$$ V_G = \bigsqcup_i Q_i $$
where each $Q_i$ is of constant size, such that any 
operator $H\in S$, acts on qudits from at most $2$ such sets $Q_i$.
\end{theorem}

As previously discussed, the way this theorem is proved, is by identifying 
a special
$1$-D structure called the backbone whose properties allow to coarse-grain 
the set of qubits
into constant-dimension sets, such that each operator is $2$-local w.r.t. 
these sets. 

\subsubsection{Properties of the backbone}
Here we define and identify the backbone, and prove 
its desired properties.  

\begin{definition}\label{def:bkbn}

\textbf{Backbone}

\noindent
For an instance $S$ of $CLH(3,2)$ we define the backbone to be 
a maximal length operator path $B$ in the 
connectivity graph $G_S$. If there are several such maximal length 
paths, we take one of them arbitrarily. 
\end{definition}

We start by proving a lemma which states that for any backbone qubit, 
any operator that acts on that qubit must also act on some 
``nearby'' backbone qubit. 

\begin{lemma}\label{lem:entrap}{\bf Short range connectivity in the backbone}
Let $B=\{H_1,H_2,...,H_L\}$ be a backbone of a connected-graph 
instance $S$ with no separable qubits, such that $L>100$. 
Let $q\in A_B$ and let $H_i(q)$ be some backbone operator.
Then for any $H=H(q)\in S$ that acts on $q$, 
there exists another qubit $p\in A_B$
examined by $H$, and a backbone operator $H_j(p)\in B$ such that $|i-j|\leq 4$
where addition is modulo $L$ for a closed operator path. 
\end{lemma}

\begin{proof} 
We say $q\in A_B$ is a "middle" qubit if it is
acted upon by $3$ consecutive backbone operators $H_{\ell-1},H_\ell,H_{\ell+1}$; 
If $B$ is a closed path, then all qubits in $A_B$ are middle qubits. 
If $B$ is an open path, all qubits are middle qubits except for the two qubits
at the beginning and the two qubits and the end of the path; 
we call those edge qubits. 

The claim easily follows for a middle qubit $q$: 
Let $q$ be such a qubit,
 acted upon by $3$ backbone operators $H_{\ell-1},H_\ell,H_{\ell+1}$.
These constitute an operator crown on $q$, so any $H(q)$ acts 
on some other qubit $p$
 in $A_{\ell-1} \cup A_\ell \cup A_{\ell+1}$; 
$p$ and $q$ are thus acted upon by $H_i(q)\in B,H_j(p)\in B$ where
$|i-j|\leq 2$.

This proves the claim for the case in which the backbone is a closed path. 
For the case of an open path, we only need to prove
the claim for its edge qubits. 
Let us consider an edge qubit which belongs to the 
``left'' edge, namely, the qubit belongs to $H_1$; the proof for the 
other side is essentially identical. 
Let $H = H(q,a,b)$ be some operator in $S$.
We discuss separately $3$ cases: at least one of $a,b$ is a middle qubit, 
at least one of $a,b$ is an edge qubit, or both of them are outside of $A_B$. 

\begin{enumerate}
\item
\textbf{Either $a$ or $b$ is a middle qubit}
Assume WLOG that $a$ is a middle qubit. 
Suppose on the negative, that the minimal value $m$ for which there exists 
some $H_m\in B$ such that $a\in A_m$ or $b\in A_m$ is $6$.
Let $C(a)\subseteq A_B$ be the set of the crown qubits of $a$.
Since $m\geq 6$ then $q\notin C(a)$, 
so by Claim (\ref{cl:nosep1}) it must be that $b\in C(a)$.
In addition, since $m\geq 6$ we know that $b$ is a middle qubit, 
which is adjacent to $a$. 
$b$ cannot be in $A_m$ for $m\le 6$ by our assumption; 
Hence we know that both $a$ and $b$ are
adjacent "middle qubits" that are distant from $q$.
Consider then the terms $H$ and $H_1$. Since both $a$ and $b$ are 
distant from $q$, we have that $H \bfly H_1$ w.r.t. qubit $q$.
By claim (\ref{cl:nosep2}) there exists an operator path of 
length exactly $3$ between $H$ and $H_1$, which 
implies the existence of an operator $H_x(q)$ that shares two qubits
with $H_1$, and hence acts on at most 
one particle from $C(a)\cup C(b)$.
However, since $H_x$ shares two particles with $H$, it must act on 
either $a$ or $b$; this contradicts
claim $(\ref{cl:nosep1})$ w.r.t either qubit $a$ or qubit $b$.

\item
\textbf{$a$ or $b$ is an edge qubit}
Suppose at least one of $a$ or $b$ are edge qubits, say qubit $a$.
Hence, either $a\in H_1$ or $a\in H_L$, in both cases we are done 
as $q\in H_1$. 

\item
\textbf{$a$ and $b$ are not in the backbone}
In this last case we assume that neither $a$ nor $b$ are in $A_B$.
We show that this case is impossible, since it implies that 
the backbone can be extended, contradictory to 
the definition of the backbone (Definition \ref{def:bkbn}). 

Denote by $H_1(a_1,b_1,a_2)$,$H_2(b_1,a_2,b_2)$, and $H_3(a_2,b_2,a_3)$ as the first $3$ operators.
So we assume that either $q=a_1$ or $q=b_1$.

If $q=a_1$, 
there exists by claim (\ref{cl:nosep2}) 
an operator path of length at most $3$ between $H$ and $H_1$. It must be of length exactly $3$ by the 
assumption that $a,b$ are not in the backbone, and so $H$ cannot intersect 
$H_1$ with two qubits to make a path of length $2$.  
 This implies the existence of $H_x$ that acts on $a_1$, and shares two qubits with $H_1$.
So if $H_x$ acts on $b_1$ (as in Figure \ref{fig:ext1}) 
then $H$ and $H_x$ can be appended to $B$ to form a $L+2$-length operator path.
\begin{figure}[ht]
\center{
 \epsfxsize=4in
 \epsfbox{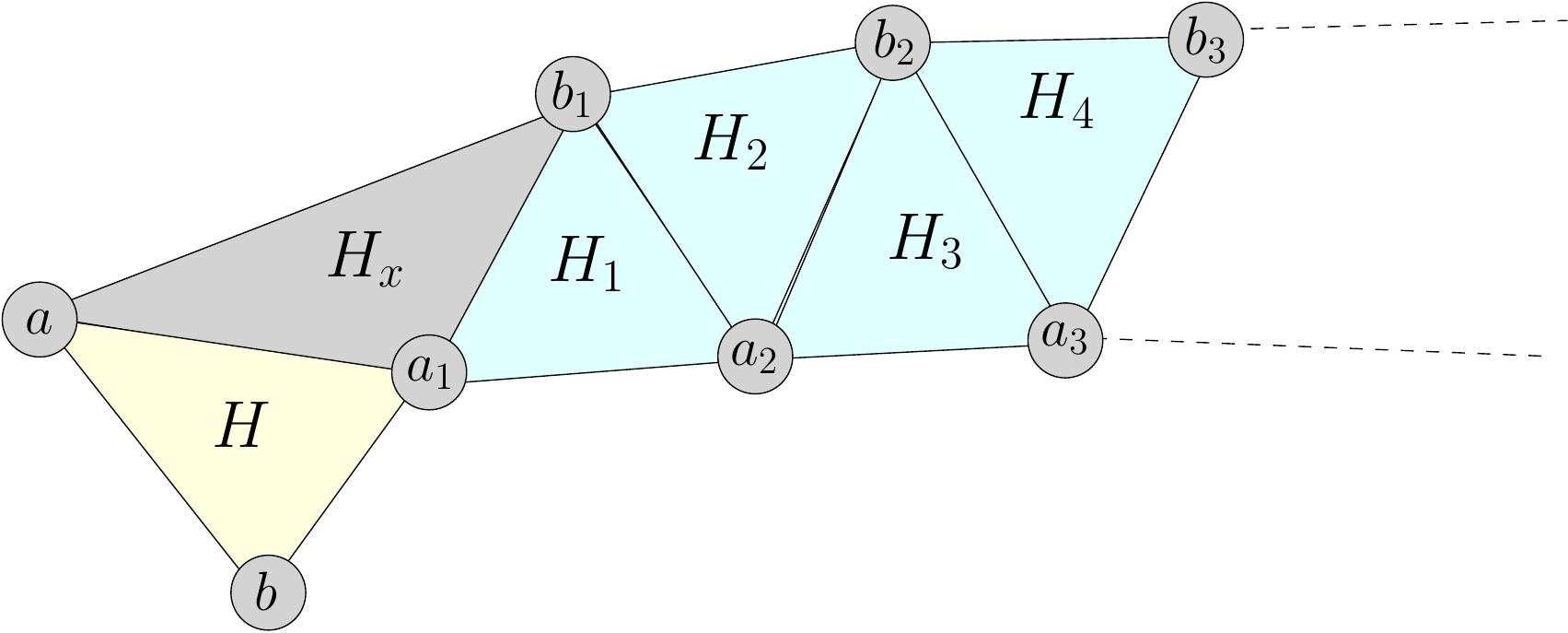}}
 \caption{$q=a_1$ and $H_x$ acts on $b_1$}
 \label{fig:ext1}
\end{figure}

Otherwise, $H_x$ acts on $a_2$ (as in Figure \ref{fig:ext2}). 
Let us examine now $H_x$ and $H_3$ which share qubit $a_2$.
There exists an open operator path on qubit $a_2$ comprised of the following $4$ operators: $H_x,H_1,H_2,H_3$.
Thus by claim (\ref{cl:bridge}) there exists an operator $H_y(a_1,a_2,b_2)$
By removing $H_1,H_2$ from $B$, and appending $H,H_x,H_y$ to $B$, we end up with a backbone of size $L+1$.
\begin{figure}[ht]
\center{
 \epsfxsize=4in
 \epsfbox{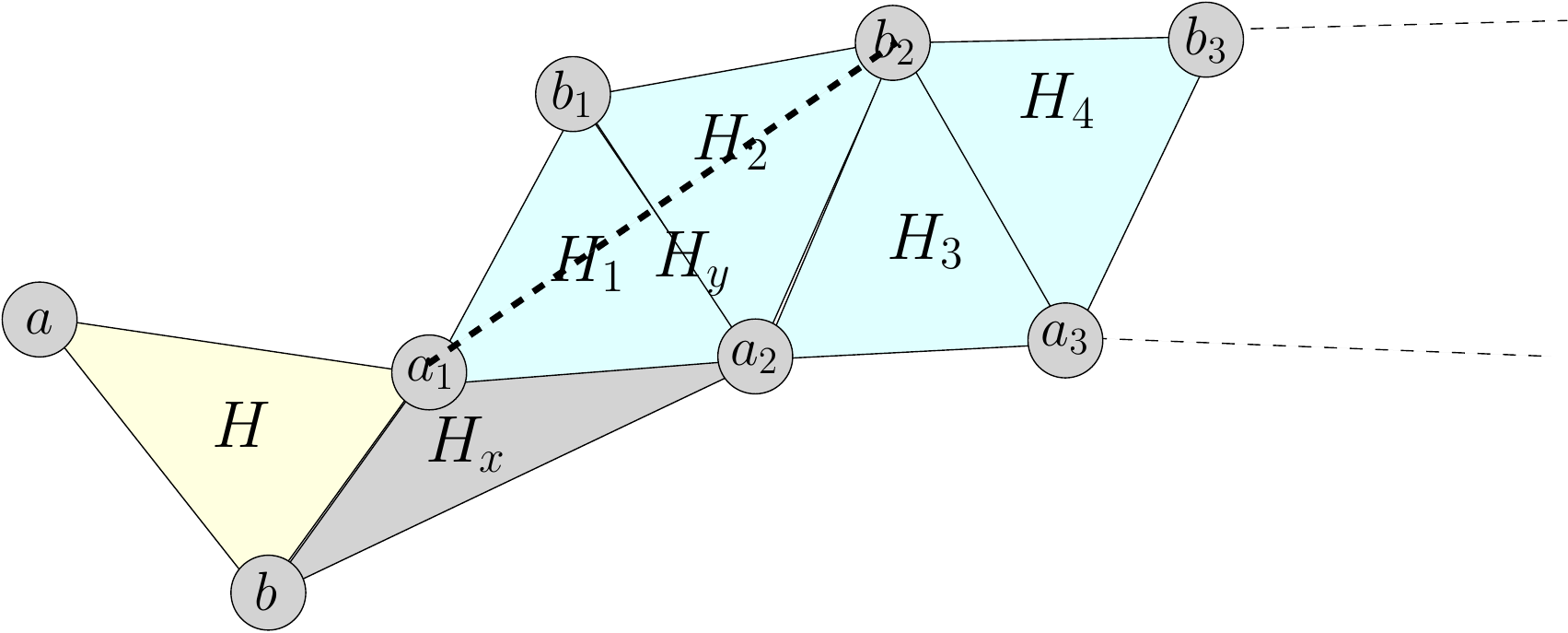}}
 \caption{$q=a_1$ and $H_x$ acts on $a_2$.}
 \label{fig:ext2}
\end{figure}

If $q=b_1$ then $H$ is connected to $H_1$ by some operator $H_x(b_1)$, 
which shares two qubits with $H_1$.
If $H_x$ acts on $a_1$ we can append $H_x$ to $B$, and increase 
its length to $L+1$ 
(Notice that we cannot append $H$ as well, since $H,H_2$ share qubit $b_1$ despite their index difference being greater than $2$ in the new path, and so this 
is not a legal operator path). 
Otherwise, $H_x$ acts on $a_2$, as in Figure \ref{fig:ext3}. 
In this case we can add $H$,$H_x$ to $B$, and remove $H_1$, thereby increasing the length of $B$ yet again by $1$.

\begin{figure}[ht]
\center{
 \epsfxsize=4in
 \epsfbox{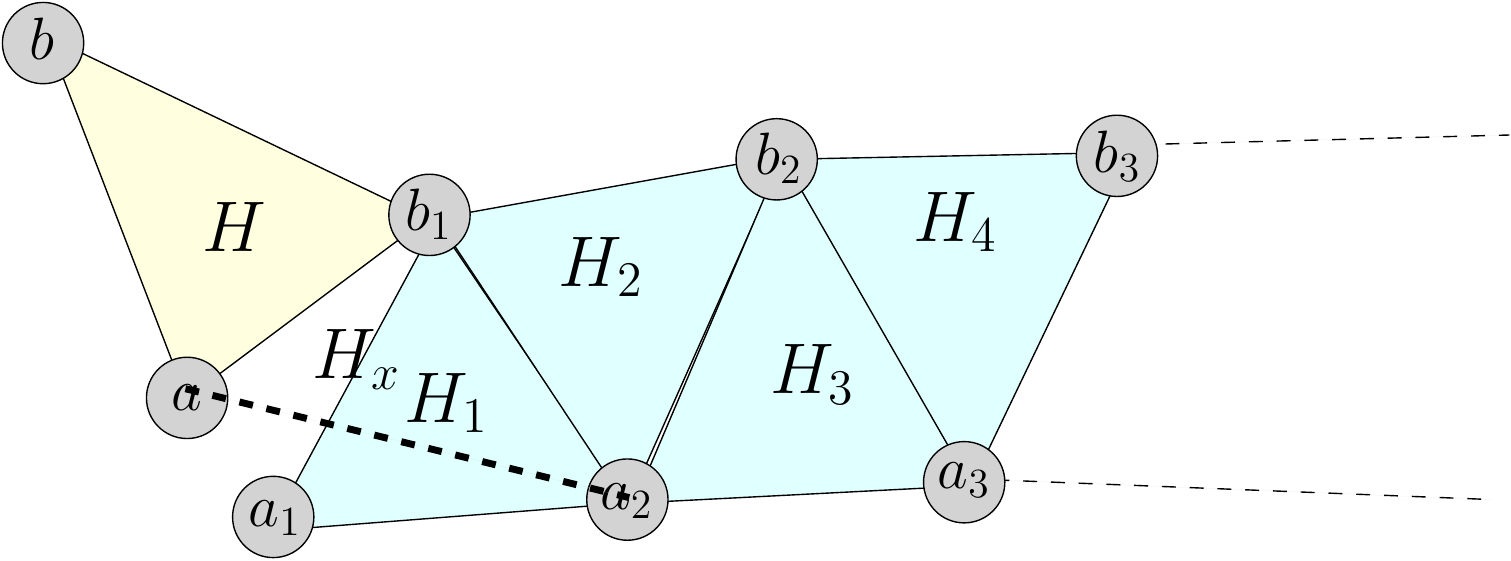}}
 \caption{$q=b_1$ and $H_x$ acts on $a_2$.}
 \label{fig:ext3}
\end{figure}

\end{enumerate}

\end{proof} 

Given a backbone $B \subseteq G_S$, we now define the combined particles.

\begin{definition} {\bf The combined particles $Q_i$}\label{def:Qi} 
Let $Q_1,\hdots, Q_M$ be a partition of the qubits in $A_B$, 
$$A_{B} = \bigsqcup_{i=1}^M Q_i$$ 
generated by grouping together contiguous sets of $20$ 
 qubits in the most natural way, from left to right along the backbone. 
If the number of qubits in $A_B$ does not divide by $20$ (but is larger than $20$), we simply 
make the last $Q_M$ larger; in any case each $Q_i$
contains at most $39$ qubits.  
If $A_B$ contains less than $60$ particles, than there is just one
particles, $Q_1$ containing all of them.  
\end{definition}

\begin{claim}\label{cl:bkbn1}

\textbf{Any operator touches the backbone in at least two close qubits}

\noindent
Assume that $G_S$ is connected. Then, 
for any operator $H(a,b,c)\in S$,
at least two of the three qubits $a,b,c$, 
must be contained in $A_{B}$, and moreover, $\{a,b,c\}\cap A_{B}$ 
is contained in two adjacent combined particles $Q_i,Q_{i+1}$ with addition 
modulo $M$ (or in one combined particle). 
\end{claim}

\begin{proof} 

We first show that at least two qubits of $\{a,b,c\}$ are in $A_{B}$.
It cannot be that exactly one of its particles is in $A_{B}$ 
since this is in contradiction to lemma (\ref{lem:entrap}). 
So it is left to rule out the case
 in which none of its particles are in $A_{B}$. 
We will show that this case is impossible. 
Since $G_S$ is connected, Claim (\ref{connected}) applies. 
Hence, for any qubit of $\{a,b,c\}$ there is an operator path starting from 
an operator acting on that qubit, to some operator $H_m\in B$. 
Let $H_l$ denote the first operator on the path for which $A_l$ intersects
 $A_B$.
By Lemma (\ref{lem:entrap}) $H_l$ examines at least 
two qubits in $A_B$.
So let us now consider $H_{l-1}$.
Since $H_{l-1}$ intersects $H_l$ 
by two qubits, at least one of the intersection qubits must be in $A_{B}$.
So $H_{l-1}$ examines at least one qubit in $A_{B}$.
This is a contradiction to the fact that $H_l$ is the first operator on the path that intersects $A_{B}$.

Now we show that $\{a,b,c\} \cap A_{B}$ is distributed among two adjacent 
$Q_i$'s. We divide to two cases: 
$\{a,b,c\} \subseteq A_{B}$ (note that this doesn't mean that 
$H_k$ is in the backbone) or only two qubits are in $A_B$. 

We start with the first case. Let $q\in \{a,b,c\}$. 
Let $H_i(q)\in B$. Then there exists some $p\in \{a,b,c\}$, 
and some $H_j(p)\in B$ s.t. $\|i-j\|\le 4$ by Lemma \ref{lem:entrap}. 
By applying this argument either once or twice we get that 
for any two qubits in $\{a,b,c\}$,
there are two backbone operators acting on them 
with indices which are at most $8$ apart. This means 
that these qubits must be contained 
in two adjacent $Q_i$'s (or just in one).  

Otherwise, two of the qubits $\{a,b,c\}$ 
are in $A_B$. 
In this case, consider the first qubit of these $2$, which belongs to some 
operator in the backbone. It follows from lemma (\ref{lem:entrap}) that the 
second qubit belongs to an operator which is at most $4$ operators away; 
Thus, these two qubits belong to two adjacent $Q_i$'s (or just in one).  
\end{proof}

\begin{claim}\label{cl:bkbn2}

\textbf{No shortcuts outside of backbone}

\noindent
Let $q\notin A_B$ be some qubit outside the backbone.
Then there exists $m$ such that for all $H_i(q)\in S$,
 $A_i\cap A_{B} \subseteq \left\{Q_m,Q_{m+1}\right\}$, with index addition 
modulo $M$.
\end{claim}

\begin{proof} 
Given a qubit $q$, let $H_k(q), H_l(q)$ be two operators on $q$. 
We first show that $(A_k\cup A_l)\cap A_B$ 
are grouped into at most two adjacent combined particles $Q_m,Q_{m+1}$ for
 some $m$. 

Since $q$ is not separable, there exists by claim (\ref{cl:nosep2}) 
an operator path on $q$ of length at most $3$ between $H_k$ and $H_l$.
So there are two cases: either $|A_k\cup A_l|=2$ or 
there exists an operator $H_x(q)$
that shares two qubits with both $H_k$ and $H_l$.

We will make the claim first for the second case. 
Since $q\notin A_B$ the two other qubits 
examined by $H_x$ 
are in $A_B$ by Claim (\ref{cl:bkbn1}).
Let us consider one of the two qubits in $A_x\cap A_B$, say $p_1$, 
and let $H_i\in B$ act on $p_1$.  
By lemma (\ref{lem:entrap}) there exists another qubit in $A_x\cap A_B$, 
say $p_2$,  
and an operator $H_j \in B$ acting on $p_2$, s.t., $|i-j|\le 4$. 
Thus the two qubits other than $q$ which $H_x$ acts on have backbone 
operators acting on them which are of distance $4$ apart. 
The same property holds for the qubits of $A_k\cap A_B$, and $A_l\cap A_B$.
Thus any pair of qubits in $(A_k\cup A_l) \cap A_B$ belong 
to operators of index difference at most $12$.
Thus, $(A_k\cup A_l)\cap A_B$ are grouped into at most 
two adjacent combined particles $Q_m,Q_{m+1}$ for some $m$.

Since the above property holds for all pairs of operators on $q$, we have that
$\bigcup_{H_i(q)} A_i \cap A_B \subseteq Q_{m}\cup Q_{m+1}$ for some $m$.
\end{proof}

\subsubsection{The Structure of
interactions with the Backbone}\label{sec:2local}

By Claims (\ref{cl:bkbn1}),(\ref{cl:bkbn2}) we realize by now that 
the input with no separable qubits $S$ and a connected interaction 
graph $G_S$ has, in fact, a very constrained 
structure. 
To describe this structure more precisely, we now define a partition 
of all qubits outside of the backbone into sets denoted by $V_i$ as follows. 

\begin{definition}{\bf The partition of the outer qubits}\label{def:Vi} 
We define the set $V =\bigsqcup_{i=1}^{M-1} V_i$ which is comprised of 
$M$ or $M-1$ (depending on whether the backbone is closed or open) 
disjoint sets of qubits $V_1,V_2,\hdots$, as follows:
$V$ is the set of all qubits not in $A_B$.
The sets $V_i$ are defined iteratively:
$V_1$ is defined as the set of all qubits $q\in V$ satisfying 
that: for any operator $H_k(q)$ $A_k\cap A_B \subseteq Q_1\cup Q_2$.  
$V_2$ is defined as the set of all qubits 
$q\in V \backslash V_1$, 
satisfying that: for any operator 
$H_k(q)$ $A_k\cap A_B \subseteq Q_2\cup Q_3$. 
And so on. 
\end{definition} 

\noindent
Note that the above definition indeed defines a partition of all qubits 
outside the backbone. 
We thus deduce that the picture of interactions is severely restricted: 

\begin{definition}{\bf Almost one dimensional structure:}\label{def:1ds}
We consider a tuple $(Q,V,S)$ consisting of 
a set of distinct qudits $Q = \left\{Q_i\right\}$, a set of non-intersecting 
qubit sets $V = \left\{V_i\right\}$ (non intersecting with the $Q_i$s as well) 
and a set of 
operators $S$ on those particles. 
We say that a tuple $(Q,V,S)$ has an ``almost one dim structure'' if 
each $H_k\in S$ has $A_k\subseteq \{Q_i\cap Q_{i+1}\cap V_i\}$ for some $i$.
\end{definition} 

\noindent
See Figure \ref{fig:step2} for a schematic example of such a one dimensional 
structure. 
By Definition (\ref{def:Vi})
and by claims \ref{cl:bkbn1} and \ref{cl:bkbn2}, 
we have that the resulting interaction after the partition into $Q_i$ 
is indeed almost one dimensional: 

\begin{claim}\label{cl:structure}
We are given $S$ with no separable qubits s.t. $G_S$ is connected. 
Consider the backbone $B$ on $S$, with $Q_i$ as defined in Definition
 \ref{def:Qi}, and $V,V_i$ as defined in Definition \ref{def:Vi}.
Then the tuple 
${Q,V,S}$ is an almost one dimensional structure, namely, 
for any interaction $H_k$ in the instance $S$, 
$A_k$ contained in $V_i\cup Q_i \cup Q_{i+1}$ for some particular $i$. 
\end{claim}
We can now use this to derive two-locality. 

\subsection{Decomposing the Backbone to get 2-locality}
We would now like to discover the two locality 
by making use of the structure of interactions with the 
backbone, which was revealed in the previous section. 
For this, we prove a lemma which is very similar to 
lemma (\ref{lem:BV}), regarding the existence of a separating decomposition 
as define in Definition (\ref{def:algdec}).

\begin{lemma}\label{SliceLemma}
Let $(Q,V,S)$ be a tuple which is an almost one dimensional structure, 
as in Definition \ref{def:1ds}. 
For any $Q_i\in Q$, there exists a direct-sum decomposition of ${\cal H}_{Q_i}$ preserved by all operators in $S$ which is a separating decomposition w.r.t. any pair of operators $H_l,H_r$ where
$A_l \subseteq (Q_{i-1},Q_i,V_{i-1})$ and $A_r\subseteq (Q_i,Q_{i+1},V_i)$.
\end{lemma}

\begin{proof}

Fix a backbone qudit $Q_i$.
Denote by $H_{l,j}\in S$ the $j$-th operator such that $A_{l,j}\subseteq (Q_{i-1},Q_i,V_{i-1})$ and
$H_{r,k}\in S$ the $k$-th operator such that $A_{r,k}\subseteq (Q_i,Q_{i+1},V_i)$.
Denote by ${\cal A}_{Q_i}^{l,j}$ the algebra of $H_{l,j}$ on $Q_i$ and by 
${\cal A}_{Q_i}^{r,k}$ the algebra of $H_{r,k}$ on $Q_i$.
Finally, denote by ${\cal A}_{Q_i}^l$ the algebra spanned by $\bigcup_j {\cal A}_{Q_i}^{l,j}$ 
and similarly, by ${\cal A}_{Q_i}^r$ the algebra spanned by $\bigcup_k {\cal A}_{Q_i}^{r,k}$.

If there is just one qudit $Q_i$ the lemma follows trivially from lemma (\ref{lem:BV}).
Otherwise, by definition (\ref{def:Qi}) we have $Q_{i-1}\neq Q_{i+1}$. 
Then $H_{l,j}\bfly H_{r,k}$ w.r.t. $Q_i$ so by fact (\ref{fact:algcomm}) the algebras 
${\cal A}_{Q_i}^{l,j}$ and ${\cal A}_{Q_i}^{r,k}$ commute for all $j,k$.
As a result ${\cal A}_{Q_i}^l$ and ${\cal A}_{Q_i}^r$ commute, because their generating sets commute.
Let us apply fact (\ref{cl:algdec}): there exists a separating decomposition of ${\cal H}_{Q_i}$ w.r.t 
${\cal A}_{Q_i}^l$ and ${\cal A}_{Q_i}^r$.
In each subspace of this decomposition, the restricted algebras act on separate subsystems.
Since ${\cal A}_{Q_i}^{l,j}\subseteq {\cal A}_{Q_i}^l$ for all $j$, then 
${\cal A}_{Q_i}^{l,j}$ preserves this decomposition, and acts only on the left subsystem 
inside each subspace.
The same holds true for ${\cal A}_{Q_i}^{r,k}$ for all $k$ w.r.t. the right subsystem.
Hence this decomposition is separating w.r.t. any pair of local algebras 
${\cal A}_{Q_i}^{l,j}$ and ${\cal A}_{Q_i}^{r,k}$.
Since $H_{l,j}$ and $H_{r,k}$ share only $Q_i$ to begin with, this is a separating decomposition
w.r.t. these operators as well.

\end{proof}

\paragraph{Reducing to a $2$-local problem}\label{red2lcl}
Merlin sends Arthur for each qudit $Q_i$ a subspace index $\alpha_i$ which by lemma (\ref{SliceLemma}) partitions all the operators on $Q_i$ into two disjoint left and right subsystems in the $\alpha_i$ subspace of $Q_i$.
Since by lemma (\ref{SliceLemma}) all operators preserve this tensor product subspace, then the tuple $(Q,V,S)$ is satisfiable if and only if there exists such a tensor product subspace in which the restricted instance is satisfiable. 

Let $\Pi_{\alpha_i}$ denote the projection on the $\alpha_i$ subspace of backbone qudit $Q_i$.
After restricting all operators $H_k\in S$ to the tensor product subspace 
$\bigotimes_i \Pi_{\alpha_i}$,
 each restricted operator $H_k|_{\bigotimes_i \Pi_{\alpha_i}}$ acts on 
${\cal H}_{Q_i}^{\alpha_i,right} \otimes {\cal H}_{Q_{i+1}}^{\alpha_{i+1},left}$
for some $i$, and possibly some qubit $q\in V_i$.
So let us now fuse each pair 
${\cal H}_{Q_i}^{\alpha_i,right},{\cal H}_{Q_{i+1}}^{\alpha_{i+1},left}$ 
into a single qudit.
We now regard the tuple $(Q,V,S)$ as the modified instance where the qudits $Q_i$ are the fused qudits, and $S$ is the restricted operator set.
It can be easily checked that all $H_k\in S$ are now at most $2$-local, acting on a single fused qudit and possibly some $q\in V_i$.
This problem can then be verified using lemma (\ref{lem:BV}).
Let $\Pi_{\beta_i}$ denote the projection on the $\beta_i$ subspace of the fused qudit 
${\cal H}_{Q_i}^{\alpha_i,right} \otimes {\cal H}_{Q_{i+1}}^{\alpha_{i+1},left}$.
Then Merlin sends Arthur a tensor-product subspace $\bigotimes_i \Pi_{\beta_i}$ - i.e. a subspace 
index for each fused qudit. 

\subsection{Putting together the proof of Theorem \ref{thm:2lclreduce}:Containment in NP}
To summarize our proof of theorem \ref{thm:2lclreduce}, 
we describe the protocol of the verifier, and recall why it is complete 
and sound. 
The witness that Merlin sends Arthur is: 
\begin{itemize}
\item
An index $\alpha$ of a subspace for each separable qubit.
\item
A description of constant-size sets 
$Q_i$ and the (possibly non-constant size) sets $V_i$, for all $i$.
\item 
An index $\alpha_i$ for each $Q_i$, 
as well as the description of the actual subspace ${\cal H}_{\alpha_i}$, 
and its tensor product structure,
${\cal H}_{Q_i}^{\alpha_i,right} \otimes {\cal H}_{Q_{i+1}}^{\alpha_{i+1},left}$. 
The subspaces are provided by providing the basis vectors.
Note that since the dimension of each $Q_i$ 
is constant, this description is efficient as long as the 
accuracy required is at most to within inverse 
exponential in the number of qubits
\footnote{Throughout 
the paper we ignore the issue of accuracy and assume that 
the subspaces are provided with infinite accuracy. The reason is that  
this issue is not a problem in the context of the $CLH$ problem, since  
we are only trying to separate a zero eigenvalue from $1$ or more.  
In fact, even inverse polynomial accuracy in each such component will suffice, 
since there are only polynomially many terms contributing to the error.}  
\item
A witness following lemma (\ref{lem:BV}) for the two-local problem following 
the merging of on the constant size sets. 
\end{itemize}

The verification procedure is performed as follows: 
First, Arthur restricts all the operators in $S$ to the subspaces 
of the separable qubits, 
provided by Merlin, and derives an instance $S_{nosep}$ with 
supposedly no separable qubits. 
Arthur then verifies that after unifying some qubits of the system into 
constant size sets $Q_i$ following a recipe by Merlin, the instance $S_{nosep}$ 
has an almost $1$-dim. structure $(Q,V,S_{nosep})$.
Then, Arthur restricts all operators on ${\cal H}_{Q_i}$ to the subspace
${\cal H}_{Q_i}^{\alpha_i}$ as provided by Merlin for each $i$.
This is a tensor product of subspaces, one for each $i$.
He then fuses $right$ and $left$ remnants of the qudits 
(i.e. ${\cal H}_{Q_i}^{\alpha_i,right}$ with ${\cal H}_{Q_{i+1}}^{\alpha_{i+1},left}$ for each $i$) and achieves a $2$-local problem. 
Finally, Arthur uses a witness provided by Merlin for 
the $2$-local problem as in Lemma (\ref{lem:BV}). 

We now argue for completeness and soundness: 
By Claim (\ref{Step1}) the first step
results in an equivalent instance $S_{nosep}$, which is a positive instance if and only if $S$ is a positive instance.
The completeness and soundness of the third step
follows exactly the proof of lemma (\ref{SliceLemma}): i.e. Merlin can find
a subspace index $\alpha$ for each $Q_i$, that separates the interaction on $Q_i$ into left-right components, such that the restricted Hamiltonian 
is a positive instance, if and only if the original Hamiltonian is a positive instance.
The correctness of the last step of the protocol, namely, for the two local residual problems, follows from the correctness of \cite{BV}.

We end with the following corollary:
\begin{corollary} {\bf Constant Depth Diagonalizing Circuit}
For any $CLH(3,2)$ Hamiltonian $H$ there exists a quantum circuit of depth
$3$, where each gate acts on at most $100$ qubits that diagonalizes $H$.
\end{corollary} 

\begin{proof}
A depth $3$ circuit is induced in a straightforward manner from the protocol above.
All separable qubits can be diagonalized separately from all other particles using $1$-local qubit gates.
Let us now consider just the nonseparable qubits.
We work our way from the end of the protocol in (\ref{red2lcl}) backwards: the first two layers handle the $2$-local star-topology instance, and the third layer embeds this instance into the backbone.
We recall again the following notation:  the backbone qudits are denoted by ${\cal H}_{Q_i}$, and restricted to some subspace $\alpha$ we have ${\cal H}_{Q_i}^{\alpha,left}$ and ${\cal H}_{Q_i}^{\alpha,right}$ for each $i$.
Then when Merlin proves to Arthur the satisfiability of the star-topology instance, he provides a subspace 
$\Pi_{\beta_i}$ 
 for each fused qudit ${\cal H}_{Q_i}^{\alpha,right} \otimes {\cal H}_{Q_{i+1}}^{\alpha,left}$.
To write the circuit, we actually go in reverse order 
as follows:
\begin{enumerate}
\item 
\textbf{Layer 1: Disjoint Hamiltonian Diagonalization} 

A diagonalizing gate for each disjoint Hamiltonian interaction $(q,p)$ where $q\in V_i$ and $p$ is some fused qudit $\left({\cal H}_{Q_{i+1}}^{\alpha,right} \otimes {\cal H}_{Q_i}^{\alpha,left}\right)|_{{\beta}_i}$ 
restricted to some subspace ${\beta}_i$.
Each such gate is a $2$-local unitary acting on the qubit $q$ and some qudit 
of dimensionality less than
 the product of dimensions of two $Q_i$'s, hence at most a constant.
\item
\textbf{Layer 2: Embedding of disjoint Hamiltonians into Star Topology}
An isometric embedding $V$ of the individual Hilbert spaces $p$ from the item above (which may require adding some ancilla qubits), into a fused qudit that is now still restricted to some subspace $\alpha$.
We have
$$ V_2: 
\left({\cal H}_{Q_i}^{\alpha,right} \otimes {\cal H}_{Q_{i+1}}^{\alpha,left}\right)|_{{\beta}_i}
\mapsto 
{\cal H}_{Q_i}^{\alpha,right} \otimes {\cal H}_{Q_{i+1}}^{\alpha,left}.$$
The isometry $V_2$ acts on possibly numerous individual particles such that the product of their sizes is at most $100$.
\item
\textbf{Layer 3: Embedding of Star Topology instance into backbone}
An isometric embedding of each pair of adjacent fused subparticles:
$$ V_3:
\left({\cal H}_{Q_{i}}^{\alpha,left} \otimes {\cal H}_{Q_i}^{\alpha,right}\right)
\mapsto 
{\cal H}_{Q_i}.$$
This may also require adding some ancillary qubits.
\end{enumerate}

\end{proof}

{~}

\noindent{\bf Remark}: 
We remark that in fact, a depth $2$ quantum circuit suffices. 
We will defer the proof to a later version of the paper. 

\section{The 3-local case for qutrits}\label{sec:3lcl3}

In this section we extend the result of the previous section for qutrits, presuming the interaction graph
can be embedded on the plane in a special way.

\subsection{Formal Definitions}\label{sec:3lcldef}
For a planar connected interaction graph $G_S$ of a $CLH(3,3)$ instance $S$, embedded in $R^2$ we define the following:

\begin{definition}\label{def:op}

\textbf{Op/Noop Faces}

\noindent
Let $f$ be some face of $G$.
$f$ is said to be an "op" face if it has $3$ vertices,
 and there exists an operator $H\in S$ that acts on its $3$ vertices 
(qudits). If there is no operator that acts on its vertices it 
is called a "noop" face.
In this respect, the face extending from the boundary of $G$ to 
infinity is a "noop" face.
\end{definition}

\noindent
We now define a notion of planar $CLH$: 

\begin{definition}

\textbf{Planar $CLH$}

\noindent
An instance $S$ of  $CLH(k,d)$  
is said to be ``planar  $CLH$''
if its interaction graph is planar, connected, and moreover,  
can be associated with an embedding in the plane, such that 
every operator in $S$ is associated 
with an ``op'' face. In other words, the three vertices corresponding to the 
particles on which the operator acts, are connected by edges in the graph, 
such that the area surrounded by those edges does not contain any other 
vertex.  
\end{definition}

In Figure \ref{fig:nonplanar2} we provide examples which explain the difference between the mere requirement that the interaction graph is planar, and our definition of planarity. 
\begin{figure}[ht]
\center{
 \epsfxsize=4in
 \epsfbox{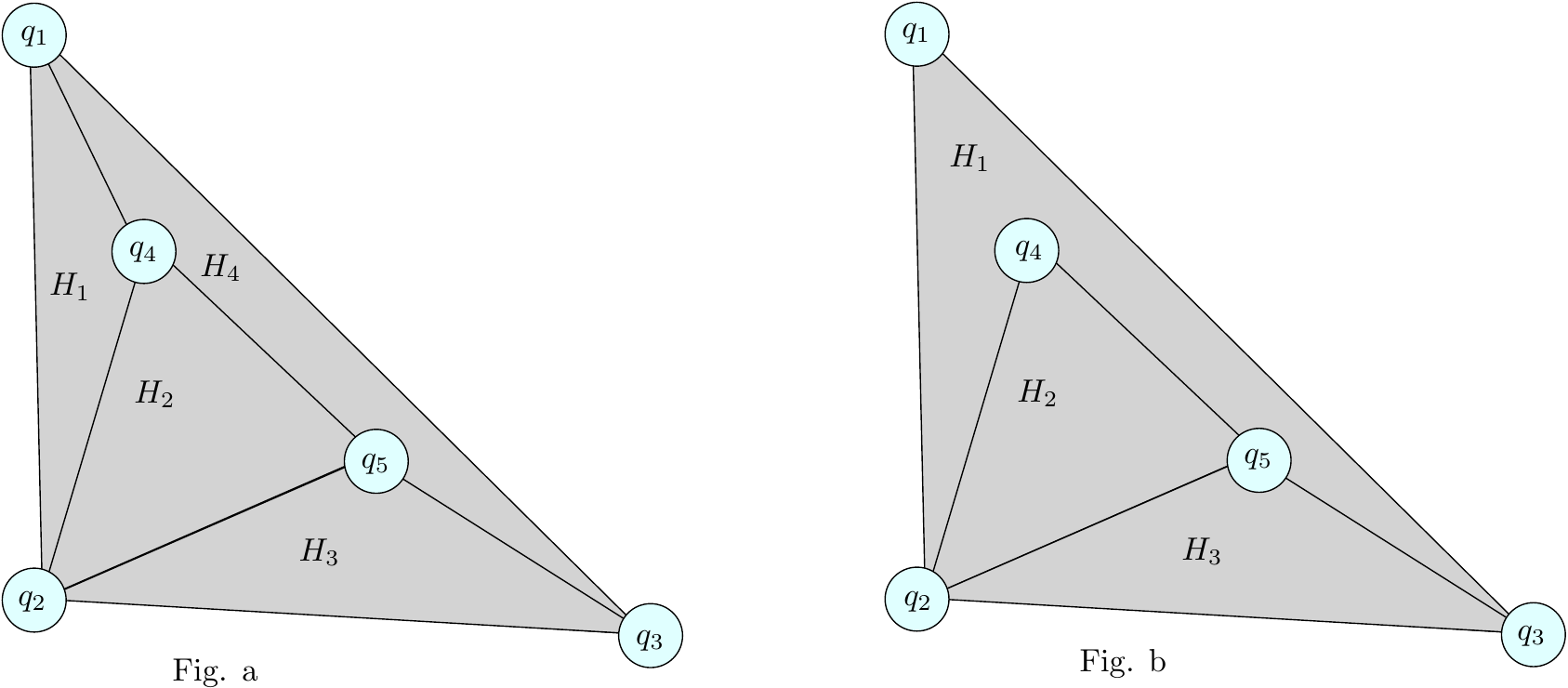}}
 \caption{\label{fig:nonplanar2}
 Examples of instances of $CLH(k,d)$ whose interaction graphs are embedded on the plane, yet are not instances of planar $CLH$ as defined above:
 \newline
 Fig. a: Suppose that in addition to the operators $H_1,H_2,H_3,H_4$ there also exists an operator $H_5(q_1,q_2,q_3)$.
 Then there is no "op" face corresponding to $H_5$.
 \newline
 Fig. b: 
 Suppose there are only $3$ operators $H_1,H_2,H_3$.  
 There exists a face corresponding to each operator, yet the one corresponding to $H_1$ is not an "op" face since it is not triangular.
 Note however, that there exists a different embedding of the interaction graph of $H_1,H_2,H_3$ so that it is a legal planar $CLH$ instance.
}
 
\end{figure}
So far, these are very mild and natural restrictions on top of planarity.  
We now add one additional ``natural'' requirement. 
For this, let us define another notion: 

\begin{definition}

\textbf{Nearly-Euclidean Triangulation of a Polygon}

\noindent
A finite planar graph is said to be a Nearly Euclidean (NE) 
triangulation of a polygon 
if every face except the infinite face has three edges, 
the edges are straight lines, and moreover, the ratio between the shortest and longest edge is bounded from above by some overall constant, 
and the angle between any two incident edges is bounded from below 
by some overall constant angle. 

\end{definition}

Notice that this definition ensures that there cannot be areas in which the 
density of vertices is much higher than in other places, 
and that the length of a path in the graph is not too far from that 
of the Euclidean distance between the two end points. 
In short; propagation along paths 
in the graph behaves more or less like it would on a periodic lattice, or 
in Euclidean space. 

\noindent
We now define NE $CLH$ instances: 

\begin{definition}\label{def:NEdef}
\textbf{Nearly Euclidean $CLH$}

\noindent
An instance $S$ of planar $CLH(k,d)$ is said to be Nearly Euclidean (NE)
if there exists a NE 
planar triangulation $T_S$ of a convex polygon in the plane,
such that the interaction graph of 
$S$, $G_S$, is a subgraph of this triangulation, 
i.e. $G_S\subseteq T_S$ and the number of vertices of $T_S$ 
is at most polynomial in the number of particles (equivalently, 
vertices) in $S$.
\end{definition}

\noindent
Our main result in this section is:
\begin{theorem}\label{thm:d3}
The problem of $CLH(3,3)$ restricted to NE planar instances is in NP.
\end{theorem}

\noindent
From our proof, we also derive the following corollary: 
\begin{corollary}\label{cor:circuit}
Given an instance of NE $CLH(3,3)$, there exists a quantum circuit 
of depth $3$, involving $2$-local gates acting on nearest neighbor
qutrits in the interaction graph of the original instance, 
which diagonalizes the Hamiltonian (and thus generates a basis of eigenstates 
from a basis of input computational basis states). 
\end{corollary} 

\subsection{Proof Overview}

Similar to the proof for qubits, our proof of theorem (\ref{thm:d3}) above is along the following strategy:
identify "classically"-behaving qutrits, and eliminate them using input from Merlin.
Then, show that the residual instance can be coarse-grained
into a $CLH(2,d)$ instance for a constant $d$.
The proof that this coarse graining can be done, however, is entirely
different, from the one in the previous chapter. 

The first difference is in identifying classically-behaving qutrits; 
this is slightly more involved than for qubits.
The larger qudit dimension allows qutrits to interact in certain configurations that generate qubit separability, but do not generate qutrit separability.
Consider for example the following set of $4$ operators: $H_1(q,a,b),H_2(q,b,c),H_3(q,c,d),H_4(q,d,e)$ where
$$ H_1 = \ketbra{0}_q \otimes \Pi_a \otimes \Pi_b $$
$$ H_2 = \ketbra{0+1}_q \otimes (I-\Pi_b) \otimes \Pi_c $$
$$ H_3 = \ketbra{1+2}_q \otimes (I-\Pi_c) \otimes \Pi_d $$
$$ H_4 = \ketbra{2}_q \otimes (I-\Pi_d) \otimes \Pi_e. $$
These operators commute yet do not agree on any single decomposition of $q$, despite the fact that they are all $\bfly$ connected.
The fact that $\bfly$-connectedness does not imply separability (i.e., 
no equivalent of Theorem \ref{bflychainClaim} holds) 
destroys the basis for most of the geometric structure we managed to 
prove in the case of qubits.  

From an algebraic perspective the above phenomenon occurs because an operator on a qubit can be written in block-diagonal form w.r.t. only one basis of that qubit (see lemma (\ref{cl:same})).
Alternatively stated, any two decompositions of $H(q)$ w.r.t. the 
qubit $q$ preserve each other (i.e., the projections on the subspaces 
in the decomposition preserve the subspaces in the other 
decomposition) and are hence identical.
Yet for qutrits, this is obviously not the case.
Consider for example the following operator:
$$ H = \ketbra{0}^q\otimes \Pi_0^E + \ketbra{1}^q\otimes \Pi_1^E + \ketbra{2}^q \otimes \Pi_2^E$$
where the first term of each summand acts on a qutrit $q$, 
and the second term acts on the rest of the system.
If the $3$ projections $\Pi_i^E$ are linearly independent, then any non-trivial direct-sum decomposition of $H$ w.r.t. $q$ must preserve the subspaces corresponding to the decomposition above ($\ketbra{0},\ketbra{1},\ketbra{2}$).
If, however, $\Pi_0^E = \Pi_1^E$ then $H$ may also be written as
$$ H = \ketbra{+}^q\otimes \Pi_0^E + \ketbra{-}^q\otimes \Pi_1^E + \ketbra{2}^q \otimes \Pi_2^E$$
where $\ketbra{+}$ projects on the vector $\frac{1}{\sqrt{2}}\left(\ket{0}+\ket{1}\right)$ and $\ketbra{-}$ projects on the state $\frac{1}{\sqrt{2}}\left(\ket{0}-\ket{1}\right)$.
So this decomposition of $H$ does not preserve the first decomposition.

Nevertheless, notice in the above example that every decomposition must 
preserve the subspace spanned by $\ket{2}$, and its orthogonal complement.
We call this one dimensional subspace a 
\textbf{Critical Subspace}.
More formally, a critical subspace is a rank $1$ projection in the center of the induced algebra of the operator on the target qutrit.
For the above example, if $\Pi_0^E,\Pi_1^E,\Pi_2^E$ are all linearly 
independent, $H$ has $3$ critical subspaces spanned by the states 
$\ket{0},\ket{1},\ket{2}$, whereas if $\Pi_0^E = \Pi_1^E \neq \Pi_2^E$ 
$H$'s only critical subspace is spanned by $\ket{2}$.
We show (facts (\ref{fact:reduce}) and (\ref{fact:bflycrit})) that 
a $\bfly$ relation between a pair of operators acting non-trivially on a qutrit $q$ implies
that each of these operators has a critical subspace on $q$,
and these subspaces are either identical or orthogonal, and moreover each of 
the operators preserves each other's critical subspace.

This behavior of critical subspaces is what 
replaces the notion of unique decomposition in qubits, though 
it is weaker. We can prove the following, when restricting the interactions 
to act on the plane. 
Consider all operators on a qutrit $q$; each operator has its own critical 
subspace in the Hilbert space of $q$. We prove that if the number 
of operators is large enough, 
any assignment of critical subspaces to the operators on $q$ forces all of 
the operators to preserve at least one of the assigned critical subspaces, 
and so the qutrit becomes separable. 
This means that there cannot be more than a small number of operators 
acting on $q$; this is proved in Lemma (\ref{lemma:cq2}). 

An easy implication of Lemma (\ref{lemma:cq2})
is that in an 
instance with no separable qutrits, 
all vertices in the interaction graph must be of degree at most $5$.  
(we state this in Corollary (\ref{cor:deg})). 
This is a crucial point. 
We show that planar-embedded 
Hamiltonians in which all the vertices are of degree at most $5$, 
must exhibit an intriguing characteristic, which is in fact entirely 
geometrical (see Claim (\ref{cl:outerbound})). Consider 
a planar embedding of a graph, whose faces are colored black and white. 
Moreover, only 
$3$-vertex faces can be colored black 
(black regions correspond to terms in the Hamiltonian).   
Then there must be a constant density of white "holes"; i.e.,  
any point in the plane is within a constant distance (in terms of number 
of faces) from such a white hole - i.e., a region where no interaction 
acts.  
The proof of this geometric fact uses the Euler Characteristic but  
is fairly involved and we delay its overview to Subsection 
(\ref{subsec:regspac}) where it is proven.  

The main point is that the existence of those regularly spaced holes 
allows us, in the case the interaction graph is NE (and this is the only place 
where we use the NE property in the proof)  
to coarse grain the set of particles, and by this derive a 
$2$-local instance. 
This is done in Claim (\ref{cl:cutting}): 
the rough idea is to lay down on the plane a ``net'' that partitions 
the plane in such a way that in each region, there are only constantly 
many particles, while making sure that the 
junctions of the net fall precisely inside those white "holes".
If we combine the particles in each region together, then   
each term in the Hamiltonian acts on at most $2$ of 
the combined particles. 
We now proceed to the details. 

\subsection{Removing Separability}\label{subsec:removesep}
We begin by defining qutrit separability in exactly the same way as qubit separability:

\begin{definition}\label{def:qutritsep}
A qutrit $q$ is called separable in an instance $S$ of $CLH(3,3)$ if there exists a non-trivial direct-sum decomposition of ${\cal H}_q$
$$ {\cal H}_q = \bigoplus_{\alpha} {\cal H}_q^{\alpha}$$
such that any operator $H(q)\in S$ preserves this decomposition, i.e. $H(q) = \bigoplus_{\alpha} H(q)|_{\alpha}$.
\end{definition}

Similar to handling $CLH(3,2)$ the NP protocol begins by Merlin helping 
Arthur remove all separable qutrits from the system.
Let $S$ be some instance of a NE $CLH(3,3)$, and $q$ be some qutrit.
If $q$ is separable then following input from Merlin, its dimension is
reduced to at most $2$, using restriction to some subspace.
After the restriction, if the dimension is $1$, the particle essentially "vanishes" from the input, so suppose we are left with a $2$-dimensional particle, a qubit.
This qubit may now be either separable or non-separable.
If it is separable, by additional input from Merlin, it
"vanishes" from the input.
Thus, after all qubit/qutrit separability has been exhausted, 
we are left with a Hamiltonian acting on qutrits and qubits, all of which 
are non-separable.
We note that the remaining interaction graph is still NE, 
even though some ``op'' faces may have now turned into ``noop'' faces. 

\subsection{Critical Subspaces and Separability for general CLH(3,3)}\label{subsec:sepqutrit}

We would now like to develop tools that will allow us to prove 
geometrical restrictions on the residual problem after removing all 
separable qudits. 

\begin{definition}\label{def:critsub}

\textbf{Critical Subspace}

\noindent
Let $H(q)$ be an operator acting on $q$. 
Denote by ${\cal A}_q^H$ the algebra induced by $H$ on $q$.
Any rank-$1$ projection in the center of this algebra 
$\mathbf{Z}\left({\cal A}_q^H\right)$, induces a 
a critical subspace of $H$ on $q$, which is the one dimensional 
subspace which is the 
image of the rank-$1$ projection.  
\end{definition}

\noindent
We now prove several useful facts regarding critical 
subspaces. The following facts hold
for all qudits of dimension $d\leq 3$.

\begin{fact}\label{fact:reduce}
Let $H_1(q),H_2(q)$ be two operators that act non-trivially on $q$ such that $H_2(q) \bfly H_1(q)$.
We claim that each of operator has at least one critical subspace on $q$.
Also, each critical subspace of $H_1$ is preserved by $H_2$ and vice versa.
\end{fact}

\begin{proof}
Let ${\cal A}_q^1$, ${\cal A}_q^2$ denote the algebras induced by
 $H_1(q),H_2(q)$ on the qudit $q$, respectively. 
First, we show that both algebras ${\cal A}_q^1$, ${\cal A}_q^2$ are reducible.
Suppose on the negative that it is not the case, and that say, 
 ${\cal A}_1(q)$ is irreducible. By Fact (\ref{fact1}) it is isomorphic 
to $\mathbf{L}({\cal H}_q^1) \otimes I_{{\cal H}_q^2}$. 
Since $d=3$ or $d=2$, then one
of the Hilbert spaces ${\cal H}_q^1$ or ${\cal H}_q^2$ is of dimension $1$. 
It cannot be that  ${\cal H}_q^1$ is of dimension $1$, 
since this would imply that the algebra ${\cal A}_q^1$ is trivial, 
and this means that the operator $H_1$ acts trivially on $q$, contrary 
to our assumption. Thus, it must be that 
${\cal H}_q^2$ is of dimension $1$. 
Since ${\cal A}_q^1,{\cal A}_q^2$ commute, this implies that ${\cal A}_q$ 
is trivial, and so $H_2(q)$ acts trivially on $q$ contrary to our
assumption. 
 
Thus by fact (\ref{fact2}) each of the two algebras can be non-trivially 
decomposed into a direct-sum of algebras, following an orthogonal 
decomposition of ${\cal H}_q$ into a direct-sum of subspaces 
${\cal H}_q^{\alpha}$.
Since $d\leq 3$ at least one such subspace is of dimension $1$, 
so the center of each algebra has a rank $1$ projection, 
and so both operators have a critical subspace on $q$. 

We now show that the critical subspaces of one operator
 are preserved by the other. 
Let $S_c \subseteq {\cal H}_q$ be a critical subspace of $H_1(q)$.
Let $\Pi_c$ be the projection on $S_c$; 
By Definition (\ref{def:critsub}), $\Pi_c
 \in \mathbf{Z}\left({\cal A}_q^1\right)$ and thus it is contained in 
${\cal A}_q^1$. Since ${\cal A}_q^1$
 commutes with ${\cal A}_q^2$ by Fact \ref{fact:algcomm}, 
then $\Pi_c$ commutes with any operator $A\in {\cal A}_q^2$, thus by 
Fact \ref{fact:trivial2} ${\cal A}_q^2$ preserves $S_c$, 
and thus so does $H_2(q)$. 
\end{proof}

\begin{fact}\label{fact:bflycrit}
Let $H_1(q) \bfly H_2(q)$ be two operators on $q$.
Let $S_c^1$ be a critical subspace of $H_1(q)$ on $q$. 
Then any critical subspace of
 $H_2(q)$ on $q$ is either $S_c^1$ or a subspace orthogonal to it.
\end{fact}

\begin{proof}
Let ${\cal A}_q^1$, ${\cal A}_q^2$ denote the algebras of $H_1(q),H_2(q)$ on qubit $q$.
The algebra ${\cal A}_q^2$ commutes with ${\cal A}_q^1$, and so it commutes with any element in the center of ${\cal A}_q^1$, in particular, the projection on $S_c^1$.
By Fact \ref{fact:trivial2} any operator in ${\cal A}_q^2$ preserves $S_c^1$. 
Let $S_c^2$ be a critical subspace of $H_2(q)$, and let 
$\Pi_c^2\in \mathbf{Z}({\cal A}_q^2)$ be the projection on it. 
Since this projection 
is contained in ${\cal A}_q^2$, it too preserves $S_c^1$. 
We have that a projection on a one-dim subspace $S_c^2$ 
preserves a one dimensional subspace $S_c^1$ 
and so either the two are equal or they are orthogonal. 
\end{proof}

\begin{fact}\label{fact:tight}
Let $H_1(q)$ be an operator with a critical subspace on $q$, 
and suppose that $H_1$ preserves two $1$-dim. subspaces of $q$: $S_0,S_1$ such 
that $S_0$ and $S_1$ are neither equal nor orthogonal.
Then $S_c = S_0^c \cap S_1^c$ is a critical subspace of $H_1$. 
\end{fact}

\begin{proof}
Since $H_1$ preserves both $S_0$ and $S_1$, it also 
preserves their two dimensional span, which we shall
 denote by $S_{0,1}$, and by Fact \ref{fact:trivial1}, it also preserves 
$S_{0,1}^c$ which we denote by $S_2$. 

Let us examine the operator $H_1$ restricted to $S_{0,1}$: $H_1|_{S_{0,1}}$.
This operator too preserves $S_0$ and $S_1$, and so 
since $S_0$ and $S_1$ are neither equal nor orthogonal, 
we get that this operator is block-diagonal w.r.t. 
two different orthogonal bases : $S_0,S_0^c$ and $S_1,S_1^c$.
Since $S_{0,1}$ is a two dimensional subspace, 
we can proceed in a similar way to the proof of Claim (\ref{cl:same}), 
and conclude that $H_1$ restricted to $S_{0,1}$ is trivial on $q$, 
and can be written as 
$$ H_1|_{S_{0,1}} = I^q \otimes \Pi_{0,1}^E $$
where $\Pi_{0,1}^E$ is a projection on the system not including $q$.
All in all we have:
$$ H_1 = \Pi_{0,1}^q \otimes \Pi_{0,1}^E + \Pi_2^q \otimes \Pi_2^E$$
where $\Pi_{0,1}^q$ projects on $S_{0,1}$ and $\Pi_2^q$ projects on 
its orthogonal complement $S_2$. 

It cannot be that $\Pi_{0,1}^E$ and  $\Pi_2^E$
are linearly dependent, since this would mean that they are 
in fact equal, which would imply that $H_1$ is trivial, and thus 
does not have a critical subspace, contradicting the assumption 
of the statement. 

Hence, $\Pi_{0,1}^E$ and  $\Pi_2^E$ are linearly independent, and so by 
Definition \ref{def:induced} the algebra induced by $H_1$ on $q$, 
 ${\cal A}_q^1$, is spanned by $\Pi_{0,1}^q,\Pi_2^q$. 
Its center thus included exactly one rank-$1$ projection, the projection on 
$S_2$. 
\end{proof}

\begin{fact}\label{fact:neithernor}
Let $H(q)$ act non-trivially on a qudit $q$, 
and have an induced algebra ${\cal A}_q$ on $q$ which 
is reducible. By the notation of definition (\ref{def:qutritsep}), denote the 
subspaces in the decomposition of the algebra by ${\cal H}_q^{\alpha}$ and the 
corresponding projections $\Pi_{\alpha}\in \mathbf{Z}({\cal A}_q)$. 
 Consider an unrelated one-dimensional subspace of ${\cal H}_q$ denoted 
$S_0$. Then if there exists an $\alpha$ for which $\Pi_\alpha$ 
does not preserve $S_0$, 
then $H(q)$ has a critical subspace on $q$, denoted $S+$, 
and  $S_0$ and $S_+$ are neither orthogonal nor equal.
\end{fact}

\begin{proof} 
If $\Pi_\alpha$ does not preserve $S_0$, then 
this means that also its complement $I-\Pi_{\alpha}$ does not preserve $S_0$. 
Since $\mathbf{Z}({\cal A}_q)$ is trivially closed under complement with 
$I$, also $I-\Pi_{\alpha}\in \mathbf{Z}({\cal A}_q)$. 
Then either $\Pi_{\alpha}$ or $I-\Pi_{\alpha}$ is a
rank-$1$ projection in $\mathbf{Z}({\cal A}_q^3)$
 not preserving $S_0$.
The image of this rank-$1$ projection is a critical 
subspace by definition (\ref{def:critsub}). Let us denote it by $S_+$.
Then $S_0$ and $S_+$ are neither orthogonal nor equal.
\end{proof}

\subsection{All vertices in a planar CLH(3,3) instance are of degree at most 5}
\label{subsec:planar}

Using the results of the previous subsection, 
we now show that all vertices are of degree at most $5$.
We prove the following claim for that purpose. 

\begin{claim}\label{cl:divide}{\bf Left-Right Partition implies consensus on decomposition}
Consider a set of operators on a qudit $q$ with $d\leq 3$, 
which are separated into two non-empty sets, such 
that any two operators from these two sets, one from each side, 
share only $q$. 
Then there is a non-trivial direct-sum decomposition of $q$ which is 
preserved by all those operators on $q$. 
\end{claim} 

\begin{proof}
For a qubit $q$ the lemma is the same as corollary (\ref{PartitionClaim}), so now we focus on $q$ being a qutrit.
Let us denote the sides of the division as side $1$ and side $2$.
We would like to show that there exists a direct-sum decomposition, preserved by all operators on both sides.
First, for any pair of operators $H_1(q)$ from side $1$ and $H_2(q)$ from side $2$ we have $H_1\bfly H_2$ so by Fact (\ref{fact:reduce}) 
both $H_1$ and $H_2$ have critical subspaces.
Let us take some operator on side $1$, $H_1(q)$ whose 
critical subspace is $S_0$.
We consider three possible cases and show that in each case, 
a "consensus" decomposition emerges. 
\begin{enumerate}

\item Case $1$: 
All operators on side $2$ have the same critical subspace of $q$, denoted 
$S_c$.
Then all operators on side $1$ must preserve this 
subspace by Fact (\ref{fact:reduce}), and so 
$S_c$ is preserved by all operators on $q$.

\item Case $2$: 
There exists an operator in side $2$ which has a critical subspace 
equal to $S_0$.
In this case on each side of the devision $S_0$ is a critical subspace, 
so all operators (on both sides) must preserve this subspace by 
Fact (\ref{fact:reduce}). 

\item Case $3$: Neither of the two first cases hold. 
Then there are at least two distinct critical subspaces on 
side $2$, namely $S_1$ and $S_2$ and none of them is equal to $S_0$.
Since $S_0$ is a critical subspace of an operator in side $1$
then by fact (\ref{fact:bflycrit}) $S_1,S_2$ are both orthogonal to $S_0$. 
Since they are not the same, they span the entire orthogonal subspace to $S_0$.
Any operator on side $1$ must preserve both $S_1$ and $S_2$, by Fact 
(\ref{fact:reduce}) and thus 
it preserves the span of $S_1$ and $S_2$, and so must also preserve
the orthogonal subspace of this two dimensional subspace, namely $S_0$.
By Fact (\ref{fact:reduce}), any operator on side $2$ is $\bfly$
with $H_1$ and thus must preserve 
$S_0$ too. 
\end{enumerate}
\end{proof} 
 
\begin{corollary}\label{cor:divide3}
Let $q$ be a qubit of dimension $d\le 3$. 
If all operators on $q$ can be divided to two non-empty sets, such that 
any operator from one set intersects any operator in the other set in 
$q$ alone, then $q$ is separable.
\end{corollary}

\begin{lemma}\label{lemma:cq2} {\bf No open operator paths of length $>4$} 
Let $S$ be an instance of planar $CLH(3,3)$, and let $q$ with $d\leq 3$ 
be a qudit in that instance.  Assume
 there exists a subset of the operators on $q$ that 
constitute an open operator path on $q$ of length at least $5$,
then $q$ is separable.
\end{lemma}

\begin{proof}
Let $H_1,H_2,H_3,H_4,H_5$ denote a length $5$ open operator path on $q$, 
so $H_1$ and $H_5$ intersect only on $q$.
If $q$ is a qubit then by claim (\ref{cl:nosep1}) $q$ is separable, since 
 two qubits of $H_1$ are not part of the operator crown $\{H_3,H_4,H_5\}$.
So we now treat the case where $q$ is a qutrit.
We first show that all these operators agree on some non-trivial decomposition.
The subset $H_1,H_2,H_4,H_5$ yields a division $S_1 = \left\{H_1,H_2\right\}$ and $S_2 = \left\{H_4,H_5\right\}$ 
such that any pair of operators, one from each set, intersect only on $q$.
Thus, by Claim (\ref{cl:divide}) all these $4$ operators agree on some non-trivial decomposition of $q$, 
denoted by $S_0$ and its orthogonal complement $S_0^c$.

Let ${\cal A}_q^3$ be the algebra of $H_3(q)$ on $q$.
Let ${\cal H}_q = \bigoplus_{\alpha} {\cal H}_q^{\alpha}$ be the direct-sum decomposition of ${\cal H}_q$ corresponding to the algebra ${\cal A}_q^3$ whose existence is promised by (\ref{fact2}).
Since $H_3\bfly H_1$, then by fact (\ref{fact:reduce}) $H_3$ has a reducible algebra on $q$.
So the decomposition ${\cal H}_q^{\alpha}$ is non-trivial.  

Let us examine the behavior of the projections $\Pi_{\alpha}$ on the subspaces ${\cal H}_q^{\alpha}$ w.r.t. $S_0$.
If $S_0$ is preserved by $\Pi_{\alpha}$ for all $\alpha$ we are done, since 
by Fact (\ref{fact2}) $H_3$ preserves $S_0$, and so all $5$ operators preserve the decomposition $S_0,S_0^c$.
So suppose this is not the case.

Then there exists a subspace ${\cal H}_q^{\alpha}$ whose corresponding projection $\Pi_{\alpha}\in \mathbf{Z}({\cal A}_q^3)$ does not preserve $S_0$.
We are now in the situation handled by Fact \ref{fact:neithernor}; 
we conclude that either $\Pi_{\alpha}$ or $I-\Pi_{\alpha}$ is a
rank-$1$ projection in $\mathbf{Z}({\cal A}_q^3)$ not preserving $S_0$.
The image of this rank-$1$ projection is a critical subspace denoted by $S_+$.
and $S_0$ and $S_+$ are neither orthogonal nor equal.

Since $H_1\bfly H_3$ and $H_1\bfly H_5$ then by fact (\ref{fact:reduce}) $H_1$ and $H_5$ preserve $S_{+}$ and so by our 
assumption they preserve both $S_0$ and $S_{+}$.
Since their algebras on $q$ are reducible by Fact (\ref{fact:reduce}), 
we conclude by Fact (\ref{fact:tight}) 
that they must both have the subspace $S_2 = S_0^c \cap S_{+}^c$
as a critical subspace. 
We get that $H_2,H_3,H_4$ all of
 which intersect either $H_1$ or $H_5$ 
only on $q$, must preserve $S_2$ by fact (\ref{fact:reduce}).
Therefore, all $5$ operators preserve $S_2$.

To complete the proof we now need to show not only $H_1,.., H_5$ but all 
operators on $q$ can agree on some preserved non-trivial subspace. 
We will prove this by induction, "adding back" the other operators on $q$
one by one. Let $H_6,...,H_L$ be all other operators on $q$. 
We will assume that all operators $H_1,...,H_j$ (for $j\ge 5$) 
agree on a non-trivial 
decomposition (namely, preserve its subspaces), and in particular, 
preserve a one dimensional subspace $S_0$, 
and prove that 
$H_1,...H_{j+1}$ must also agree on one such decomposition (which might 
be different). 

We divide to two cases. Either $H_{j+1}$ preserves $S_0$, in which case 
we are done, or it doesn't. Let us therefore assume it doesn't. 
Since $S$ is planar, it must be that all operators $H_6,...,H_L$ satisfy 
$\bfly H_2,H_3,H_4$, and so 
$H_{j+1}$ has at least one $\bfly$ relation. By Fact (\ref{fact:reduce})
its algebra 
on $q$ ${\cal A}_q^{j+1}$ is reducible, and so it has 
a non-trivial decomposition of this algebra; using the notation of Fact 
(\ref{fact2}) we denote the subspaces in the decomposition 
by ${\cal H}_q^{\alpha}$. Since $H_{j+1}$ does not preserve $S_0$, 
this means that there exists a subspace ${\cal H}_q^{\alpha}$ whose 
corresponding projection does not preserve $S_0$.
We are now again in the situation of Fact (\ref{fact:neithernor})
and we conclude that  
the center of ${\cal A}_q^{j+1}$ has a rank-$1$ projection 
on a one dimensional space which is neither equal nor orthogonal to $S_0$.
This is a critical subspace of $H_{j+1}$ denoted by $S_{j+1}$.

By assumption, $H_2,H_3,H_4$ preserve $S_0$, and by 
Fact (\ref{fact:reduce}) $H_2,H_3,H_4$ preserve $S_{j+1}$. 
We can now apply Fact (\ref{fact:tight}) to deduce that 
the critical subspace of all $3$ 
operators must be orthogonal to the space spanned by $S_0$ and $S_{j+1}$. 
Let us denote this space by $S_c$.
All other operators we mentioned have a $\bfly$ relation to at least
one of these $3$ operators, thus by fact (\ref{fact:reduce}) 
all operators $H_1,...,H_{j+1}$ preserve $S_c$. 
\end{proof}

\begin{corollary}\label{cor:deg}
If an instance of planar $CLH(3,3)$ has no separable qudits, then 
all the vertices of its interaction graph are of degree at most $5$.  
\end{corollary}

\begin{proof} 
On the negative, let $q$ be a nonseparable qudit 
in an instance $S$ of NE planar $CLH(3,3)$, with degree is at least $6$.
Since $q$ is nonseparable then by corollary (\ref{cor:divide3}) 
all operators on $q$ are operator-path connected.
Since $S$ is a planar instance, it means that there exists an operator 
path on $q$ comprised of all operators on $q$.
If the degree of $q$ is at least $6$, this
means that $q$ is acted upon by an open operator path of length at least $5$ 
so by lemma (\ref{lemma:cq2}) it is separable, contrary to our assumption.
\end{proof} 

\subsection{Regularly-spaced holes in triangle tilings with no high degree vertices}\label{subsec:regspac}
This subsection is entirely geometrical, and considers tiling of  
the plane with triangles. 
We show that if the tiling avoids vertices of degree 
$6$ and above, it must be that the tiling must contain 
regularly spaced holes, with constant density. 

\subsubsection{Background and definitions}

\begin{definition}\label{def:dual}

\textbf{Dual Graph} 

\noindent
Given a planar graph $G$, the dual graph of $G$ as follows:
The set of vertices of $\hat{G}$ are comprised of all faces in $G$. 
An edge between vertices of $\hat{G}$ exists, if and only if 
the respective faces share an edge in $G$.
\end{definition}

\begin{definition}

\textbf{Tessellation} 

\noindent
Given an instance $S$ planar $CLH$, 
a {\it tessellation} is a subset of operators $T\subseteq S$, whose induced  
embedding, i.e., the embedding of $G_T$ in $R^2$, 
is a triangulation of a polygon,  
which contains no "noop" faces. 
For a tessellation $T$ we define an external vertex as one which belongs to the infinite face, and otherwise it is an internal vertex.
\end{definition}

\paragraph{Average Degree and the Euler Characteristic}

A well known formula connects the average degree of a vertex in a planar 
graph, with the average number of edges per face.
Given a tessellation $T$, and its interaction graph $G_T$ with its embedding 
in $R^2$,  
let $a$ be the average number of edges per face (including the infinite face), 
and $b$ be the average number of edges incident on a vertex.
Let $F$ denote the number of faces of $T$ (including the infinite face),
 $V$ the number of vertices, and $E$ the number of edges.
Recall the following definition:
\begin{definition}{\bf Euler Number}
The Euler number $\chi$ is defined to be $\chi = V-E+F$. 
\end{definition}
Assuming that $G_T$ is a connected planar graph,
then by Euler's theorem we have $\chi=2$.

One can easily check that the following relations hold:
$$ a\cdot F = 2\cdot E, b \cdot V = 2\cdot E,$$
since counting the number of edges 
of each face (including the infinite face) counts twice the number of edges, and similarly, counting the number of edges incident on each vertex also counts twice the number of edges.
Using these two, together with the expression for $\chi$ yields the following formula:

\begin{equation}\label{eq:qutrit1}
(a-2)\cdot (b-2) = 
4 \left(1-\frac{\chi}{F}\right)\left(1-\frac{\chi}{V}\right) =
4 \left(1-\frac{2}{F}\right)\left(1-\frac{2}{V}\right).
\end{equation}
where we have used the fact that $\chi=2$. 
\subsubsection{Constant density of holes: Statement and Overview of Proof}

We prove the following claim: 
\begin{claim}\label{cl:outerbound}
Let $\eta \ge 500$ be some constant.
We are given a tessellation $T$, such that there exists an 
operator $w\in T$ for which the following holds: 
a) any face in $G_T$ which is at $\hat{G}_T$-distance at most $\eta$ from $w$ 
is an "op" face, and 
b) all "op" faces in $G_T$ are at $\hat{G}_T$-distance at most $\eta$ from $w$.
Then there exists a vertex $q$ in $T$ whose degree is at least $6$.
\end{claim}

We will assume that all vertices in $T$
have degree at most $5$, and arrive at a contradiction. 
To do that, we will lower bound $b$, the average 
number of edges incident on a vertex, by a number larger than $5$. 
This would imply that there must be a vertex of degree $6$ or more.  

In order to extract interesting information about $b$, we in fact turn to 
 $a$, the average number of edges per face, and prove it
is very close to $3$. In other words, we would like to show that
$a$ is dominated by the "internal" faces which by our assumption are all 
"op" faces (i.e., their number of edges is  $3$), 
whereas the outside multi-edge face is negligible in 
determining the average number of edges of a face. 
We can then use Equation \ref{eq:qutrit1};  
we will be able to deduce that $b$ must be strictly larger than $5$, 
since $F,V$ are large.  

Our main effort is thus to bound $a$ from above. To this end,  
\begin{definition}
Let $\hat{a}$ be the number of edges of the outside face of $T$
(there is only one such face since $T$ is a tessellation). 
\end{definition}  
and we can write
\begin{equation}\label{eq:ahat}
a = \frac{1}{F} (3 \cdot (F-1) + \hat{a})
\end{equation}
We will show (and that's the main effort in the proof)
that $\hat{a}$ is bounded by a small constant, $13$, 
so when $F$ is large enough, we will get an upper bound on 
$a$ which is very close to $3$. 

The proof that $\hat{a}$ is bounded by a constant goes as follows.  
We devise a procedure that generates a 
sequence of tessellations, starting from the tessellation $T_1=w$, namely 
the single operator given in the statement of the claim. 
The sequence of operator sets is denoted 
$T_1 = w,\hdots, T_i,\hdots, T_m$, for $m=\eta/5$, the first one being $w$ itself, 
where for each $i$ $T_i\subseteq T_{i+1}$, and for all $i$, $T_i\subseteq T$.
Each $T_i$ includes at most $4$ additional operators from $T$, 
compared with $T_{i-1}$.
We show that despite the growth in number of faces of the tessellation, 
the number of external edges cannot cross a certain bound.

To this end, the choice of how to 
construct the tessellation $T_{i+1}$ from $T_i$ is of particular 
importance.
Suppose that given a tessellation one would construct the $T_i$'s by adding 
the operators corresponding to some arbitrary path in the dual graph, 
one by one. 
Then the number of edges of $\hat{a}$ will grow 
proportionally to the number of faces $F$; that would be a bad choice.
Our scheme thus relies on constructing the $T_i$'s in a {\it spiral} 
fashion, by choosing a ``special'' external vertex of $T_i$
at each step, and then adding all operators acting on it to 
``close'' an operator path on that vertex. 
We can show that such a process can continue 
for a large number of 
steps ($\eta/5$) without failing, i.e., 
a ``special vertex'' is well defined for $i<\eta/5$, 
and thus the $T_i$ are strictly increasing, 
while $\hat{a}$ cannot increase beyond some constant value. 

We now proceed to the details. 

\subsubsection{Detailed proof of Claim \ref{cl:outerbound}} 
As mentioned, we assume by contradiction that all vertices in the tessellation 
$T$ have degree at most $5$. 
\begin{definition} 
Consider an external vertex $v$ in $T$. We 
denote by $n_T(v)^k$ the vertex located $k$ vertices away from $v$ on the
path that traverses the external vertices of $T$ in 
a counterclockwise direction.
\end{definition} 

The procedure "Tessellate" is given as input $T$ and an operator in $T$, $w$
and returns a sequence of tessellations $T_i$ is as follows. 

\noindent
\begin{algorithm}\textbf{Tessellate(T,w)}
\noindent
\begin{enumerate}
\item
Init ($i=1$): set $T_1$ to be $w$. 

\item
($i=2$) Choose arbitrarily an external vertex $v_1$ of $T_1$ (all vertices are external at this point) and close an operator path on it, by adding to $T_1$ the necessary operators from $T$ that close the operator path on $v_1$.
Denote the new set of operators as $T_2$.

\item
For $2<i\leq \eta/5$, set $v_i$ as the vertex $n_{T_{i-1}}(v_{i-1})^k$ that is external in $T_i$ with minimal $k$, i.e.
the closest left neighbor of $v_{i-1}$ in the external 
path of $T_{i-1}$, which remained external 
following the closing of the operator path on $v_{i-1}$.
If all external vertices of $T_{i-1}$ were closed following the appending of operators during the creation of $T_i$, then set an arbitrary external vertex in $T_i$ as $v_i$.
Then close an operator path on $v_i$, again, by adding to 
$T_i$ the operators from $T$ that close an operator path on $v_i$ (namely, 
all operators $H(v_i)\in T\backslash T_i$. 
Denote the resulting set as $T_{i+1}$, and so on. 
\end{enumerate}
\end{algorithm}

We claim 
\begin{fact}\label{fact:tes1}
Given $T$ and $w$, the steps of the procedure "Tessellate" above
are well defined for all $i\in [1,\eta/5]$, and the procedure generates strictly increasing 
tessellations $T_1\subset T_2 \subset\cdots \subset T_{\eta/5}\subseteq T$.  
\end{fact}

\begin{proof}
We show that at each step $i$ for $i\in [1,\eta/5]$ there exists a 
closed operator path in $T$ on the vertex selected as $v_i$:
By our assumption, all vertices have degree at most $5$. 
This means that each $T_i$ contains at most $4$ more operators than 
$T_{i-1}$, and so $T_i$ contains at most $4(i-1)+1$ operators. 
This means that the furthest operator in $T_i$ from $w$ 
(where distance is measured in number of edges in the dual graph) 
is within distance $4(i-1)$. 
Therefore, for $i\le \eta/5$, all operators are within distance at most
$4\eta/5$
from $w$, and thus are all ``op'' operators by assumption, 
and also are not adjacent to  ``noop'' faces.
Thus, for any vertex $v$ in $T_i$,  
for $i\in [1,..,\eta/5]$, the operators $H(v)\in T$
form a closed operator path on $v$. 
\end{proof}

\begin{figure}[ht]
\center{
 \epsfxsize=4in
 \epsfbox{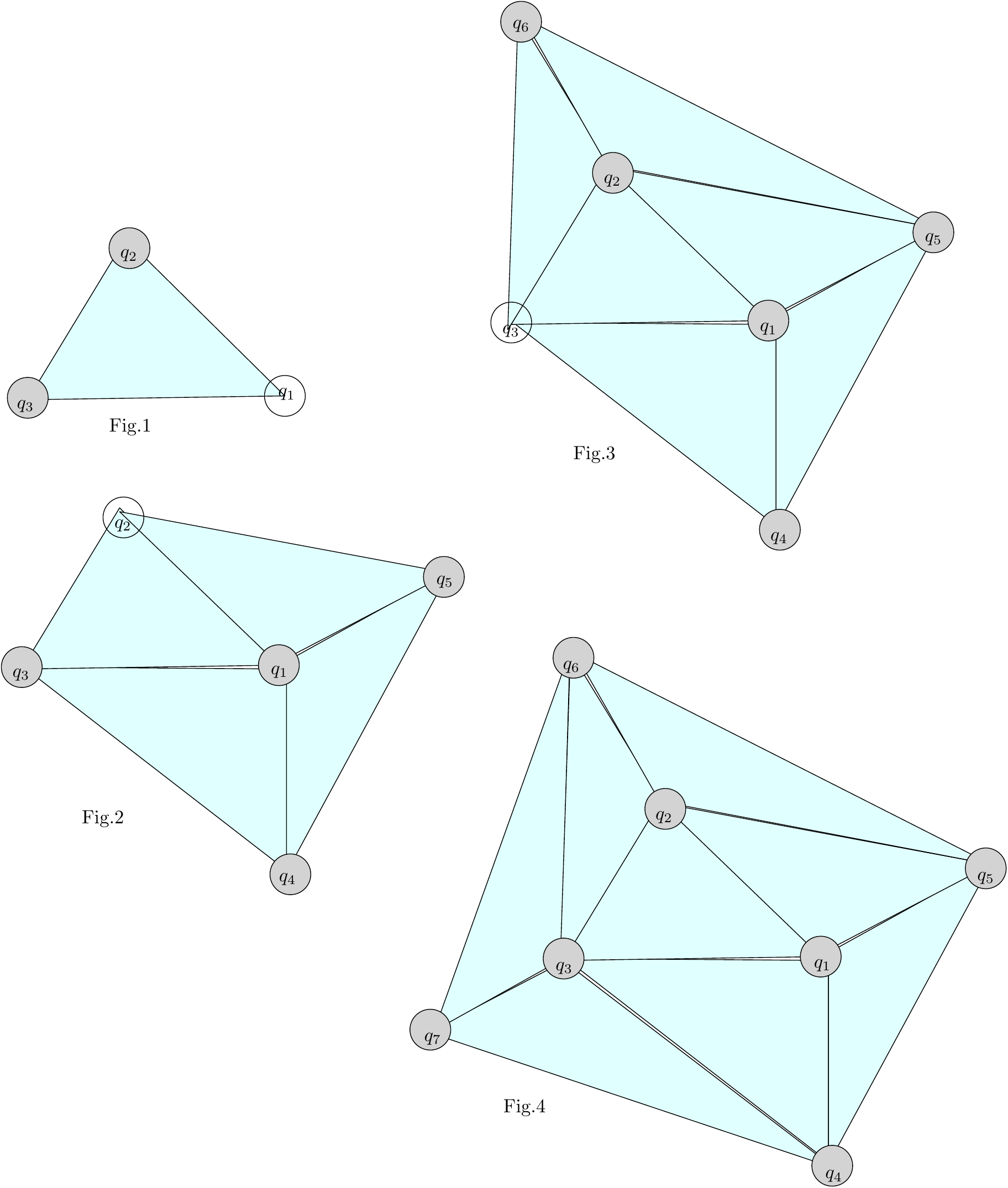}}
 \caption{\label{fig:tes1} The figure shows the first $4$ steps of a possible run of the algorithm. The unfilled qutrit, is the one chosen to close the operator path on.} 
\end{figure}

We would now like to show the important property for which the procedure was created.
Denote by $\hat{a}_i$ the number of external edges of $T_i$, and by $\hat{a}_{max} = max_{i=1}^{\eta/5} (\hat{a}_i)$
\begin{claim}\label{cl:tes2}
$\hat{a}_{max} \leq 13$. 
\end{claim}

\begin{proof}
We begin by showing the following properties:
\begin{enumerate}

\item
If $\hat{a}_i\geq 10$, then there exists some $k$ such that $n_{T_{i}}(v_{i})^k$ is external in $T_{i+1}$.

\noindent
We want to show that if $\hat{a}_i\geq 10$, then closing an operator path on any external vertex of $T_i$ cannot close simultaneously an operator path on all other external vertices of $T_i$.
Take some external vertex $q$, and close an operator path on it by appending to $T_i$ at most $4$ additional operators from $T$.
Each of these operators examines, in addition to $q$, two additional qutrits so at most $8$ external qutrits of $T_i$, aside from $q$, are examined by these operators.
Therefore, closing the operator path on $q$ in $T_i$ can close additional operator paths on at most $8$ other external qutrits of $T_i$.
Thus, if $\hat{a}_i \geq 10$, then since $\hat{a}_i$ is equal to the number of external qutrits, there will always be an external vertex of $T_i$ that remains without a closed operator path following the closing of the operator path on $q$.

\item
For each $i\in [1,\eta/5]$ if $\hat{a}_i\geq 10$ then closing the operator path on $v_i$ increases the number of operators acting on $v_{i+1}$ by at least $1$:

\noindent
Denote by $P = v_i,n_{T_i}(v_i)^1,n_{T_i}(v_i)^2,\hdots n_{T_i}(v_i)^m$ the path traversing the external vertices of of $T_i$ counterclockwise, starting from $v_i$.
Since $\hat{a}_i\geq 10$, then by item ($1$), $v_{i+1}$ is selected by "Tessellate(T,w)" as $v_{i+1} = n_{T_i}(v_i)^k$ for some $k$.
Then for all $j<k$ the vertex $n_{T_i}(v_i)^j$ is an internal qutrit in $T_{i+1}$ - i.e. a vertex such that closing the operator path on $v_i$ has also closed its operator path.
Specifically, for $j=k-1$ the operator path on $n_{T_i}(v_i)^j$ was closed at step $i$.
Yet, prior to closing, $n_{T_i}(v_i)^{k-1}$,$n_{T_i}(v_i)^k$ are two external neighboring qutrits in $T_i$, so they share an edge which is part of the external face in $T_i$, and which must be part of a newly-added operator on $n_{T_i}(v_i)^{k-1}$.
Thus, closing the operator path on $n_{T_i}(v_i)^{k-1}$ must entail adding at least one operator which acts on both $n_{T_i}(v_i)^{k-1}$ and $n_{T_i}(v_i)^k$.
Thus the operator count on $v_{i+1} = n_{T_i}(v_i)^k$ increases by at least $1$ at the end of step $i$.

\item
For each $i\in [4,\eta/5]$ if $\hat{a}_i\geq 10$, then the number of operators in $T_{i+1}$ acting 
on the selected vertex $v_{i+1}$, before closing its operator path, is at least $3$:

\noindent
We close the operator paths on the $3$ qudits of the input $w$ after at most $3$ steps.
After that, all $v_i$'s, are ones which are 
added during previous "closure" processes.
When a new vertex is added to $T_i$, it is added as part of a closed path; 
hence, when such a $v_i$ was added, it was added together with 
two operators acting on it. 
By item ($2$), when $\hat{a}_i\geq 10$ closing the operator path on $v_i$ increases number of operators acting on $v_{i+1}$ by at least $1$. 
Hence, when $v_{i+1}$ is chosen at step $i+1$, prior to closing its operator path, it is acted upon by at least $3$ Hamiltonians.
\end{enumerate}

Now we are ready to show that $\hat{a}_{max}\leq 13$.

We will first show that
for all steps $i$ for which $i\geq 5$ and $\hat{a}_{i-1}\geq 10$ 
we have $\hat{a}_i \ge \hat{a}_{i+1}$.

Recall that we assume that the degree of each vertex in $T$ is at most 
$5$, hence, since the graph is planar, there can be at most 
$5$ operators acting on each vertex.  
By item ($3$) above, before 
closing the path on $v_i$ for $i\geq 5$ and $\hat{a}_{i-1}\geq 10$ 
there were at least $3$ operators acting on $v_i$; Hence, the closing 
could have added either one or two operators that act on $v_i$. 
Let us see how each of the cases affects $\hat{a}_i$: 

\begin{enumerate}

\item
Add one face/operator to act on $v_i$ to close the path on $v_i$: 
In this case, in order to close an operator path we must connect 
existing vertices, and not add any new ones.
Therefore, $\hat{a}_i$ "looses" two edges which 
are replaced by $1$, and decreases overall by $1$.

\item
Add two faces/operators to act on $v_i$ to close the path on $v_i$: 
In this case $\hat{a}_i$ gains 
at most $2$ new edges in exchange for at least $2$ previous edges so 
it does not increase.
\end{enumerate}

Thus we have 
$(*)$ for all steps $i$ for which $i\geq 5$ and $\hat{a}_{i-1}\geq 10$ 
we have $\hat{a}_i \ge \hat{a}_{i+1}$.
On the other hand, at each step $i$ the value of $\hat{a}_i$ 
can increase by at most $2$: exactly two edges are removed from 
the external path, and at most $4$ are added; the latter case 
occurs if $4$ operators are added together with $2$ new vertices.

We now claim that for all $i$, $\hat{a}_i\le 13$. 
We know that $\hat{a}_1=3$, and 
at each one of the first $4$ steps,  
closing the operator path on each qutrit 
can increase $\hat{a}$ by at most $2$. 
Hence, $\hat{a}_i \le 3+4\cdot 2=11$ for $i\in \{1,...,5\}$.
Now, assume by contradiction 
that $\hat{a}_{j_0}\ge 14$ for some index $j_0\ge 6$. 
WLOG let $j_0$ be the first such index for which 
$\hat{a}$ is strictly larger than $13$.  
Then $\hat{a}_{j_0-1}$ must be either $13$ or $12$; 
This means that one step before that, 
 $\hat{a}_{j_0-2}$ must have been $\ge 10$, and notice that 
$j_0-2\ge 4$. 
This means that we have an index 
we have an index $i=j_0-1\ge 5$, such that $\hat{a}_{i-1}\ge 10$ and $\hat{a}_{i+1}>a_i$, in contradiction to $(*)$.  
\end{proof} 

We can now prove Claim \ref{cl:outerbound}: 
\begin{proof} 
Consider the tessellation $T_{\eta/5}$ generated by the algorithm Tessellate. 
It is comprised of at least $\eta/5\ge 100$ faces
since each step increases the number of operators by at least one, 
and so $F\ge 100$.
We also have that since $V-E+F=2$ and $E\le 5V$, then $F-2\le 4V$, 
and so $V\ge 98/4$.     
Plugging this into Equation \ref{eq:qutrit1} we have 
that $$(a-2)\cdot (b-2)\ge 3.6$$
and since $\hat{a}_{max}\leq 13$, we have $a\leq 3.1$ by 
Equation \ref{eq:ahat}. 
This implies that $b>5$, which means there must be a vertex of degree $>5$. 
\end{proof}

\subsection{Proof of regularly spaced holes} 

We can now deduce that the interaction graph of a planar 
$CLH(3,3)$ instance, with no separable qudits, 
which is a graph whose vertices are all of degree smaller 
than $6$, must have constant ``density'' of ``holes'' in it, namely, 
``noop'' faces. 
  
\begin{claim}\label{thm:reghole}
Consider the interaction graph $G_S$ 
of a planar instance $S$ of $CLH(3,3)$ with no separable qudits.
Then for any operator $w\in T$ there exists a "noop" face within distance (in the dual graph) 
at most $\eta$.
\end{claim}

\begin{proof}
We assume on the negative that there exists a vertex $w$ with no
 "noop" faces within distance $\eta$ from $w$.  
This means that any face $u$ whose distance from face $w$ is at most $\eta$ 
is an "op" face. Consider the set of all operators of distance
at most $\eta$ to $w$. This set has no gaps (otherwise there would be a 
``noop'' face of distance less than $\eta$ to $w$, and so it
is a tessellation, which we can denote $T$. 
The conditions of Claim (\ref{cl:outerbound}) now hold; 
therefore there exists a vertex
$v$ in $T$ whose degree is at least $6$ contrary to the fact that $S$ has no separable qudits.
\end{proof}

\subsection{From regularly spaced holes to two-locality}\label{subsec:alg}

We thus have a constant density of ``noop'' faces. 
Here (and only here) we use the fact that the planar embedding is NE, and
devise an NP protocol to reduce the $CLH(3,3)$ instance 
into a $CLH(2,d)$ instance for constant $d$.  
This is formalized in the following claim:  
  
\begin{claim}\label{cl:cutting}
For an instance $S$ of NE planar $CLH(3,3)$ with no separable qudits,  
there exists additional classical input from Merlin
that allows to reduce $S$ into an instance of $CLH(2,d)$ for constant $d$.
\end{claim}

\subsubsection{Overall approach} 
We now shift our attention to the 
NE planar triangulation of a convex polygon, which contains the interaction 
graph of $G_S$ as a subgraph. The existence of such a triangulation  
is guaranteed by the definition of
NE planarity (Definition \ref{def:NEdef}). 
The triangles in this triangulation are also associated with ``noop'' or ``op'' 
faces, which inherent their nature (``noop'' or ``op'') from the faces 
they are contained in, in the original embedding of the graph $G_s$ 
in the plane. From now on, we refer to this triangulation $T$ as our graph. 
We note that the density of ``noop'' faces in this graph (namely, 
Claim \ref{thm:reghole}) is still as before, 
since only ``noop'' faces were partitioned to triangles - ``op'' faces were 
already triangles. 

The basic idea is to partition the polygon 
to constant size sections, using ``cuts''
along paths in the dual graph, which are non-intersecting, 
except at ``noop'' faces.  
The reader can convince herself that if this can be done, 
this will ensure that the interactions 
between the different sets of vertices corresponding to the 
different sections are two-local, since each interaction term 
involves only particles from at most two different such areas. 

The difficulty in the proof is that due to the lack of 
regularity of the graph, it is non-trivial to construct 
explicitly such a net of paths between ``noop'' faces. 
Here is where we use the fact that  
NE planar triangulations of polygons 
obey nice characteristics, which resemble those of a periodic lattice.  
The NE requirements imply that the noop faces are scattered in a
more of less periodic fashion in terms of Euclidean distance 
between the noops. We note that no periodicity is assumed here regarding $T$,  
and the NE restriction is significantly more general, 
and cover, essentially, all ``reasonable'' physical scenarios. 

The overall idea is to lay down on the 
plane a ``Brick-Wall'', (namely, a square lattice in which the even 
rows are shifted by a half-square w.r.t. the odd rows - see Figure 
\ref{fig:bwall}), and find near each vertex of the brick wall 
a noop face, which we know exists due to Claim (\ref{cl:outerbound}); 
We essentially perturb the brick wall junctions 
so that the junctions fall inside those 
``noop'' faces.  

\begin{figure}[ht]
\center{
 \epsfxsize=2in
 \epsfbox{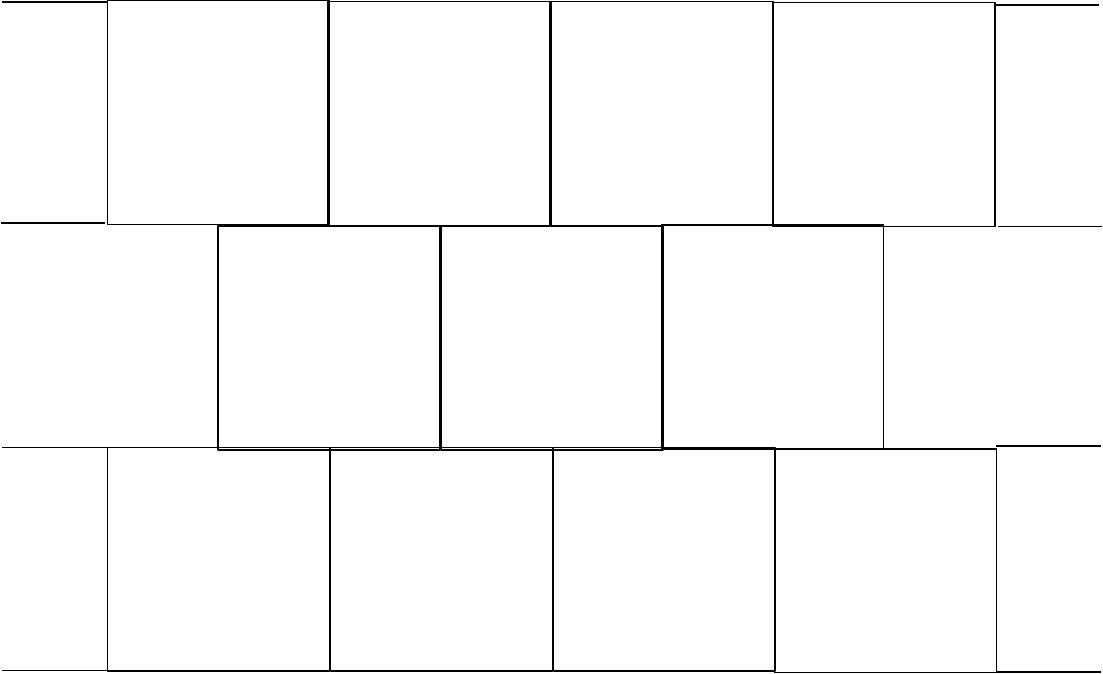}}
 \caption{\label{fig:bwall} A sample section of a brick-wall grid.} 
\end{figure}

A problem arises though from the fact that even if all the junctions
of the partition are located at noop faces, locally there may still be qudit
triplets that manage to interact across $3$ subdivisions, 
without intersecting the
``noop'' itself; this can happen
since the sizes of the edges are not all the same. 

We thus employ a trick - we in fact find two noops close to each 
vertex of the brick wall, rather than one. The two noops 
are required to be not too close to each other.  
We use the area between the two noops to create what we call a 
``noop zone''. 
The advantage is that this noop-zone is 
larger in diameter than any edge in the graph,
and so it prevents the local problems mentioned before; 
more explicitly, 
no ``op'' face can cross and connect all three bricks adjacent to this 
junction. 

To create the noop-zone, we consider a path connecting the two 
noops, and {\it contract} all edges along a path 
between the two noops, until the two noops intersect in a vertex,
which is now larger in dimension but still a constant.  
This enables us to apply Lemma \ref{lem:BV} and decompose the new 
(constant size) vertex to two, 
using input from Merlin. Separating these two vertices to two, 
we create a ``bridge'' between the two noop faces; this causes
the two noops to merge into a larger noop face, which is the  
``noop-zone''. The junction of the 
brick wall is moved to a location inside this noop-zone. 

The proof is somewhat technical, but the above ``noop-zone'' idea 
is the only non-trivial idea used in it. 
We will now proceed to the details. 

\subsubsection{Useful Characteristics of NE triangulations} 
Let us denote by $l_{max}, l_{min}$ the sizes of the maximal length and minimal 
length edges in 
the planar triangulation.  
Let us denote by $\theta_{min}$ the minimal angle between any two 
adjacent edges in this embedding. 
By our assumptions, both the ratio between $l_{max}, l_{min}$ and the value of 
$\theta_{min}$ are constants bounded away from $0$.  
For a NE CLH(3,3) instance $S$, set

\begin{equation}\label{eq:c0}
c_0 =  l_{max} \cdot 2\eta
\end{equation}
where 
$\eta$ is the constant from claim (\ref{cl:outerbound}), 
i.e., the bound on the distance from any op face to a 
noop face (distance is 
measured in terms of number of edges in the dual graph). 
We note that $c_0$ upper bounds the \emph{Euclidean} 
distance between any face and the closest ``noop'' face to it, since edges in the dual graph 
(except those connecting the infinite face) are of 
length at most $2l_{max}$. 
Here, we define the distance between two faces
 as the minimal distance of any pair of points contained in each of them.

\begin{claim}\label{cl:constdense}
Given is $T$, a NE planar triangulation of a polygon. 
Consider a rectangle $R$ in the plane.  
Then the number of vertices of $T$ that lie inside
the rectangle $R$ is at most a constant times the area of the rectangle, 
where the constant depends on $l_{min}$ and $\theta_{min}$.   
\end{claim} 

\begin{proof} 
Let $q$ be some vertex of $T$.
Let $P(q)$ be the polygon induced by the union of faces on $q$, and
$C(q)$ be the intersection of half-planes generated by the edges of $P(q)$.
Let $B(q)\subseteq C(q)$ be the largest disk, centered around $q$ that is contained in $C(q)$.
Since $B(q)\subseteq C(q)$ then $B(q)$ contains a single vertex - $q$ at its center:
on the negative, if $B(q)$ contains another vertex $p_1$ then it also intersects an edge $(q,p_1)$ of some face $(q,p_1,p_2)$, and so it intersects both half-planes of the line $(p_1,p_2)$, contrary to definition.

Also, by definition, there exists some edge $e$ in a face containing $q$, such that $e$ is opposite to $q$, and $B(q)$ is tangent to the line containing $e$.
Thus the radius of $B(q)$ over all $q$ is at least $r_{min}$ where $r_{min} = l_{min} \sin(\theta_{min})$.
Since every vertex contained in $R$ must have at least $1/4$-th of $B(q)$ contained in $R$, then the number of vertices in $R$ is constant.
\end{proof} 

\begin{claim}\label{cl:pathRect}
Let $A$ and $B$ be two faces in a NE triangulation $T_S$, which are at least 
$6l_{max}$ distance apart. 
Then there exist two vertices, $a\in A, b\in B$, such that 
the following holds. First, there exists a straight line $\ell$ 
crossing both $A$ and $B$, such that $a\in A,b\in B$ are the two vertices 
whose projections on $\ell$ are the closest among all pairs of points in 
$A$ and $B$. Denote the distance between the projections of those 
two points by $|a-b|_\ell$. There exists a path in $G$ 
between $a,b$, denoted by $P_{ab}$, which is fully contained in a rectangle 
$R_{ab}$ of area $l_0 \times w_0$, 
where $l_0 = |a-b|_\ell$,$w_0 = 
4l_{max}+4l_{max}/sin(\theta_{min})l_{min}$, and this rectangle 
contains no other vertex of $A$ or $B$. Moreover, the number of 
vertices in the path is a bounded 
function of $|a-b|_l$ and the NE parameters
$l_{min}, l_{max}, \theta_{min}$.   
\end{claim}

\begin{proof}
Consider a line which crosses both faces $A, B$. 
Consider the projections of the vertices of those faces on this line. 
If there is more than one pair which is the closest, then perturb $\ell$ 
until $a,b$ are unique and all other projections are further away. 
Now, define the rectangle $R_{ab}$ to be the rectangle two of whose edges 
are parallel to $\ell$, and are of length which is the same as the distance 
between the projection of $a,b$ on the line 
(see Figure \ref{fig:rect}); 

\begin{figure}[ht]
\center{
 \epsfxsize=2in
 \epsfbox{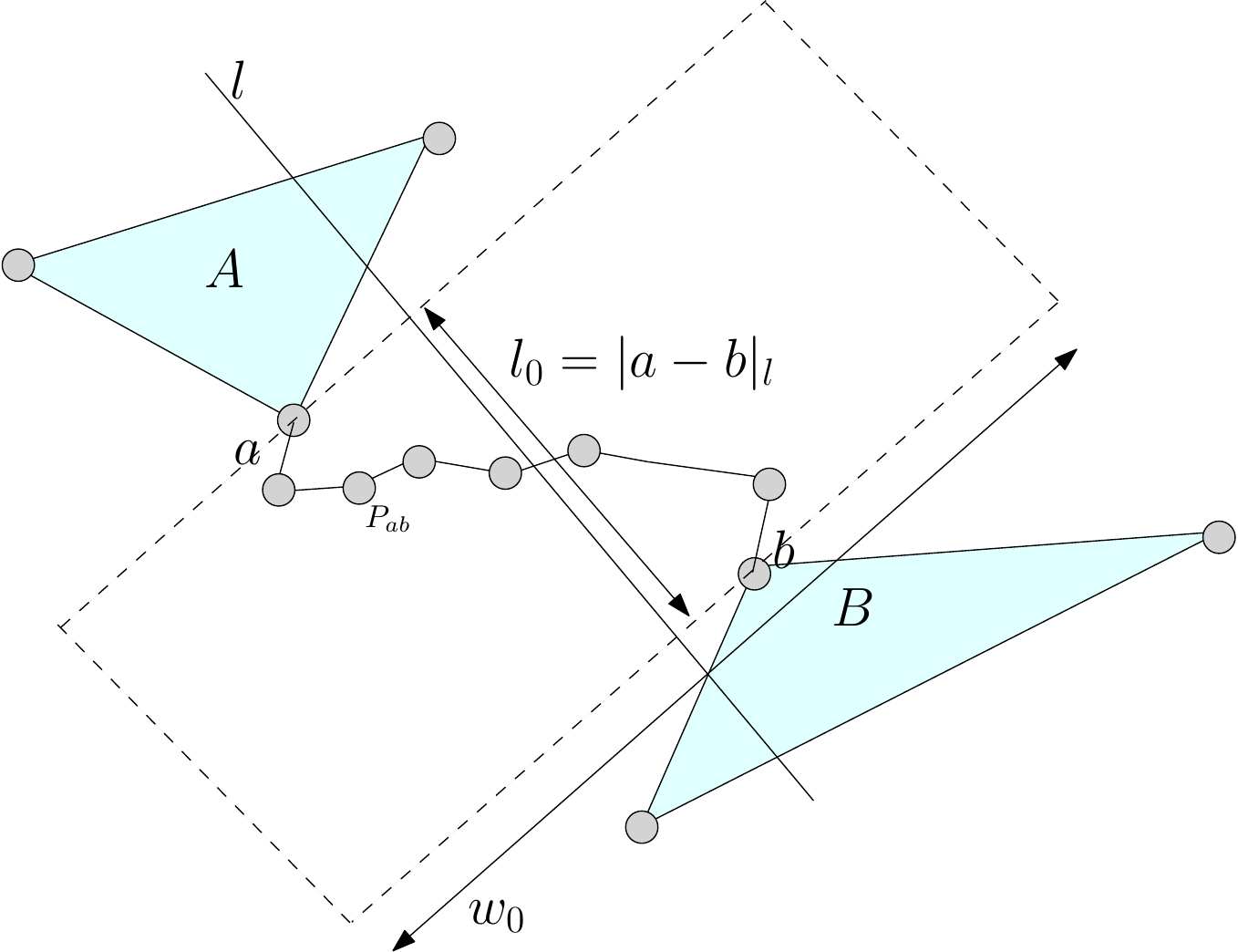}}
 \caption{\label{fig:rect} A sketch of $R_{ab}$ and $P_{ab}$.} 
\end{figure}

The other two edges of $R_{ab}$ are 
perpendicular, and are taken to be $w_0$ as defined in the claim.  
By construction, $R_{ab}$ contains no other vertices of $A$ or $B$.

We now construct the path. WLOG, let us assume that $\ell$ is parallel to 
the $y$ axis for ease of presentation. 
Starting from vertex $a$ we find a neighboring vertex to $a$ whose $y$ 
coordinate is at least 
$cos(\pi/2-\theta_{min})l_{min}=sin(\theta_{min})l_{min}$. 
We proceed this way until we reach the first vertex $a'$ 
whose $y$ coordinate
is larger than
$2l_{max}$. This will be the first segment of the path; 
denote it by $P_a$. We claim that this path $P_a$ 
is fully contained in the 
rectangle $R_{ab}$; This is because the path contains at
 most $2l_{max}/sin(\theta_{min})l_{min}$ 
vertices, and so in terms of the $x$-axis, we are still within the 
$2l_{max}$ distance from the edge of the rectangle; 
in terms of the $y$ axis, the largest $y$ coordinate is 
smaller than $3l_{max}$ by construction. Notice that 
the final point of $P_a$ is more than $2l_{max}$ away from 
the boundary of the rectangle. 
We generate in a similar way a path from $b$, $P_b$, that is contained in 
$R_{ab}$ to a vertex $b'$ with a $y$ coordinate at most that of $b$
minus $2l_{max}$. We notice that if we connect $a'$ and $b'$ by a 
straight line $l_{a'b'}$, the line $l_{a'b'}$ are not 
only contained inside $R_{ab}$ but every point of it is
also at least $2l_{max}$ from the boundary of the rectangle. 
We now finish the construction of the path by connecting $a',b'$ by a path 
inside the rectangle, as follows. 

Consider all faces crossed by this line. Find a simple path $P_{a'b'}\in T_S$ 
between $a'$ and $b'$ that uses only edges of these faces - 
There exists such a path $P_{a'b'}$ because the graph $T_S$ 
restricted to the faces crossed by $l_{a'b'}$ is a connected graph.
Also, by $NE$ restrictions, the path $P_{a'b'}$ cannot visit 
vertices whose Euclidean distance from $l_{a'b'}$ is more than $2l_{max}$.
Hence $P_{a'b'}$ is contained in $R_{ab}$.

Connecting the paths $P_a,P_b,P_{a'b'}$ and removing edges as necessary to 
derive a simple path, we arrive at the desired 
simple path $P_{ab}$. 
Since by claim (\ref{cl:constdense}), $R_{ab}$ contains at most a 
number of vertices proportional to its area, 
the length of $P_{ab}$ is a bounded 
function of $|a-b|_\ell$ and the NE parameters. 
\end{proof}

\subsubsection{Creating the partition}
We are now ready for the proof of the claim. 

\begin{proof}(of claim \ref{cl:cutting})

\paragraph{Defining the Grid}
We are given an instance $S$ of NE $CLH(3,3)$ with an interaction graph $G_S$.
Let $T_S$ denote the NE triangulation of the convex 
polygon of which $G_S$ is a subgraph: $G_S\subseteq T_S$.
Let $\alpha = 7c_0 + 2w_0$ where $w_0$ is the constant from claim (\ref{cl:pathRect}).
Let us lay down on the plane a brick-wall ${\cal B}$ with block-size of constant length $10\alpha$.
With small perturbations (shifting and rotating) 
we can make sure that no edge is parallel to 
the edges of the planar embedding of $T_S$, and no vertex of the brick wall 
lands on an edge or vertex of $G$ in this embedding.
By Claim (\ref{cl:constdense}) we know that each cell contains a constant 
number of vertices of 
$T_S$, and thus of $G_S$. 
We note that interactions that connect vertices from different cells 
are obviously two-local (in terms of number of cells participating in 
one such interaction) as long as they are 
far away from the junctions of the brick-wall. 
Unfortunately, we can still have three-local interactions among those 
close to the brick-wall junctions. We thus make small modifications of 
the interaction graph close to every brick-wall junction; 
we will make sure that those modifications are confined to small 
disks around those junctions, which do not intersect each other.  
From now on we focus our attention on 
one such junction,  and explain how to modify the interactions 
close to it to eliminate the 3-local 
interactions. 

\paragraph{Creating a ``noop-zone'' near a brick-wall junction}

For each face $F_i\in G_S$ that contains the $i$-th junction of ${\cal B}$, the point $p_i$, let us denote by $N_i\in T_S$ the closest noop operator to the junction.

Given $N_i$, choose a noop $N_{i}'$ within Euclidean distance
 $\in [c_0 ,6c_0]$ to 
it (Euclidean distance between faces is the minimum 
such distance between any two points in those faces).  
$N_i'$ exists since we can consider 
a face whose Euclidean distance 
 to $N_i$ is $\in [3c_0, 4c_0]$, (if no such face exists, the 
problem is of constant size anyway) and apply Claim (\ref{cl:outerbound})
to find a noop close to this face within Euclidean distance $c_0$ 
\footnote{we need to add a length of an edge to 
account for the fact that those distances are taken 
from different points in the intermediate face, 
which is why we take loose bounds}.
Since the distance between $N_i$ 
and $N_i'$ is at least $c_0> 6l_{max}$, Claim (\ref{cl:pathRect}) holds, 
and we can find $a\in N_i,b\in N_i'$,
and $R_{ab}$ be a rectangle containing a path $P_{ab}$, 
as is guaranteed by the claim.

Consider the half-infinite line starting from $a$, 
and going to infinity in a direction normal to the edge of the rectangle
containing $a$. Consider also the half-infinite line starting from $b$ 
and going to infinity in the parallel direction, normal to the 
edge of the rectangle containing $b$. 
These half-infinite lines, augmented with the path $P_{ab}$, partition 
the plane to two infinite regions, one on each side of $P_{ab}$; 
we call the sides arbitrarily ``left'' and ``right'' (these names 
may be completely unrelated to true ``left'' and ``right''.) 
Denote $O_{left}$ as the set of faces containing at least one vertex 
of $P_{ab}$, possibly vertices on the left of $P_{ab}$,
but no vertex to its right.  
Denote $O_{right}$ to be the set of all other faces containing
vertices of $P_{ab}$ (by definition, those triangles contain
at least one vertex to the right of $P_{ab}$). 
By the planarity of $T_S$ we have that $O_r\in O_{right}$ cannot 
contain any ``left'' vertex:
a face $O_r$ supported on both sides of the divide must have an edge that crosses at least one edge of the path $P_{ab}$ or one of the edges of $N_i$ or $N_i'$.
Thus any two operators, one from 
$O_{right}$ and one from $O_{left}$ can only share vertices 
of $P_{ab}$. 

We would now like to modify the interaction graph, and 
``merge'' all vertices along the path $P_{ab}$ to one vertex, 
whose Hilbert space is the tensor product of all those particles, 
creating a large though constant dimensional new particle, 
such that $O_{left}$ and $O_{right}$ intersects only through 
this single particle.   
This is done as follows. 
First, we erase all nodes on $P_{ab}$ except $a$. 
A new large particle, of the dimensionality 
which is the tensor of all particles in $P_{ab}$, 
is embedded where previously $a$ was embedded on the plane.  
For any vertex $v$ that was merged into $a$, and $v'$ a vertex that 
was previously connected to $v$ by an edge, replace the original edge 
$(v,v')$ by a straight edge connecting $v'$ to $a$. 
Also, replace each of the  
interactions that contained a particle on the path, by the corresponding 
interaction on the appropriate subparticle in the new merged particle $a$. 
The remaining particles in that interaction are not changed. 
We note that due to the change of geometry, some edges may now cross others, 
but this non-planarity is confined only to an area close to the path, 
and as we shall see, will not matter to our argument. 

The partition of the operators on the vertices of $P_{ab}$ into 
the sets $O_{left}$ and $O_{right}$ now induces a partition of all 
the operators on $a$, which only intersect on $a$. 
Hence, the operators on the merged qudit $a$ 
can be separated, following Claim (\ref{cl:algdec}), 
into two qudits, $a^l$ and $a^r$,  
using an isometry sent by Merlin.
Note that this isometry acts on a constant dimensional particle $a$, 
since the length of $P_{ab}$ is constant following claim (\ref{cl:pathRect}).
We embed $a^l$ and $a^r$ at some tiny nonzero distance,
each to one side of the previous location of $a$, and slightly modify 
all edges connected to $a$ to fit this small perturbation, keeping them 
straight for simplicity.

We have arrived at an equivalent instance, except now
the two noop faces $N_i$ and $N'_i$ are connected into one region;  
we denote the unified noop face as $\bar{N}_i$ and call it 
the ``noop-zone''. 
We note that if initially, either the left or the right side (or both) 
of the operators on $a$ contains no operator - 
then we can trivially unite $N_i$ and $N'_i$ into $\bar{N}_i$, 
by removing edges of $T_S$, and no merging of particles is required.

Let $a_i$ be the middle point on the face not containing $a$ of $N_i$,
and $b_i$ the same for $N'_i$ respectively.
Since we have chosen the distance between $N_i$ and $N_i'$ to be at 
least $c_0$, and since the edges that contain $a_i$ and $b_i$ are not 
changed during the merging process,
then 
\begin{equation}\label{eq:abmax}
|a_i-b_i|>2l_{max}. 
\end{equation}
We denote 
$l_i$ as the straight line connecting $a_i$ and $b_i$.

\begin{figure}[h]
\center{
 \epsfxsize=4in
 \epsfbox{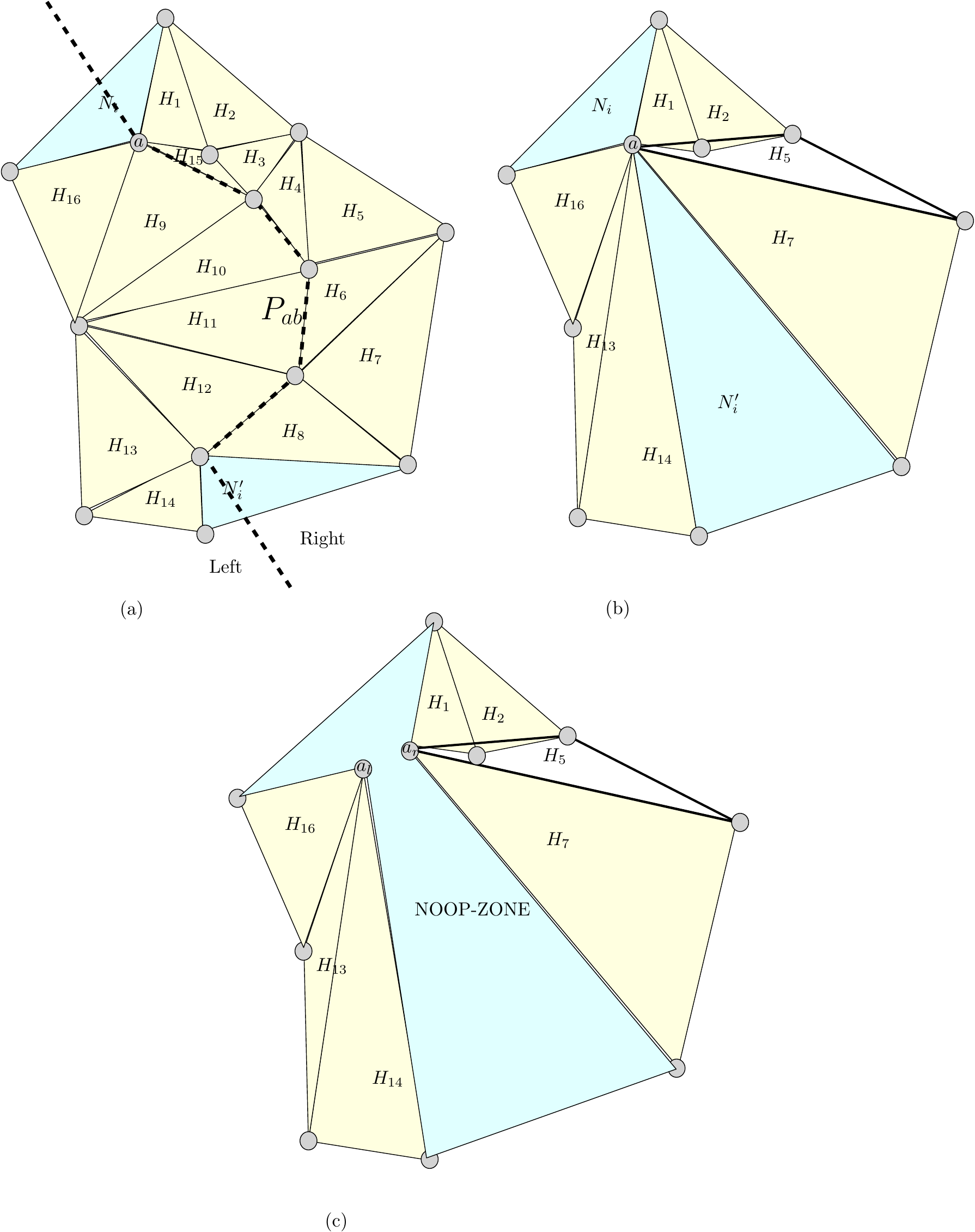}}
 \caption{\label{fig:noopzone}
 An example of a merging step: 
 Figure (a) is the embedding prior to merging - the path $P_{ab}$ (dashed line) induces a left-right separation of the plane, such that any operator on vertices of the path is supported on only one of these divisions.
 Figure (b) is the embedding after merging - note that not all faces are drawn for clarity reasons.
 Note also that some edges cross others - like the edges of $H_5$ and $H_2$.
 Figure (c) is after application of the separating isometry on vertex $a$.  It is separated into two vertices, and a large noop one is generated.} 
\end{figure}

We explain now that the graph has not changed except for in a constant 
size disk around the noops. 
Let $B_{\alpha}(p)$ denote a disk around a point $p$ with radius $\alpha$.
Since for all $i$ $N_i$ is within distance $c_0$ from $F_i$ by 
(\ref{cl:outerbound}), and since by 
Claim (\ref{cl:pathRect}) the merging described above of the vertices on the 
path between $N_i$ and $N_i'$, affect edges which are confined 
to a radius $max\left\{w_0,l_0\right\}+ 2l_{max} <6c_0+ 2w_0$
around $N_i$, 
then all edges changed by this merging 
are contained within $B_{\alpha}(p_i)$ for all $i$, 
so they are decoupled from each other and the disks do not intersect.

\paragraph{Rerouting the grid to prevent $3$-locality}
Let us reroute the junction so that the "horizontal" line of the $T$ junction when restricted to $B_{\alpha}(p_i)$ 
is tilted to the straight line going through $b_i$ that is perpendicular to $l_i$.
The "vertical" line of $T$ is tilted so that when restricted to $B_{\alpha}(p_i)$ it is the line 
starting from $b_i$, going to a point in the middle of the short 
line connecting 
$a^l$ and $a^r$, and then connecting to $a_i$ by a straight line, from which we continue 
upwards in the same direction as $l_i$, as in Figure \ref{fig:reroute}. 
Outside of the disk, these lines are connected directly
to the edges of the brick-wall. 
This can be easily achieved using the detours as in Figure \ref{fig:reroute}. 
Once again, we make sure using small rotations and shifts of the lines 
that none of these lines pass through 
any vertex. 
\begin{figure}[h]
\center{
 \epsfxsize=4in
 \epsfbox{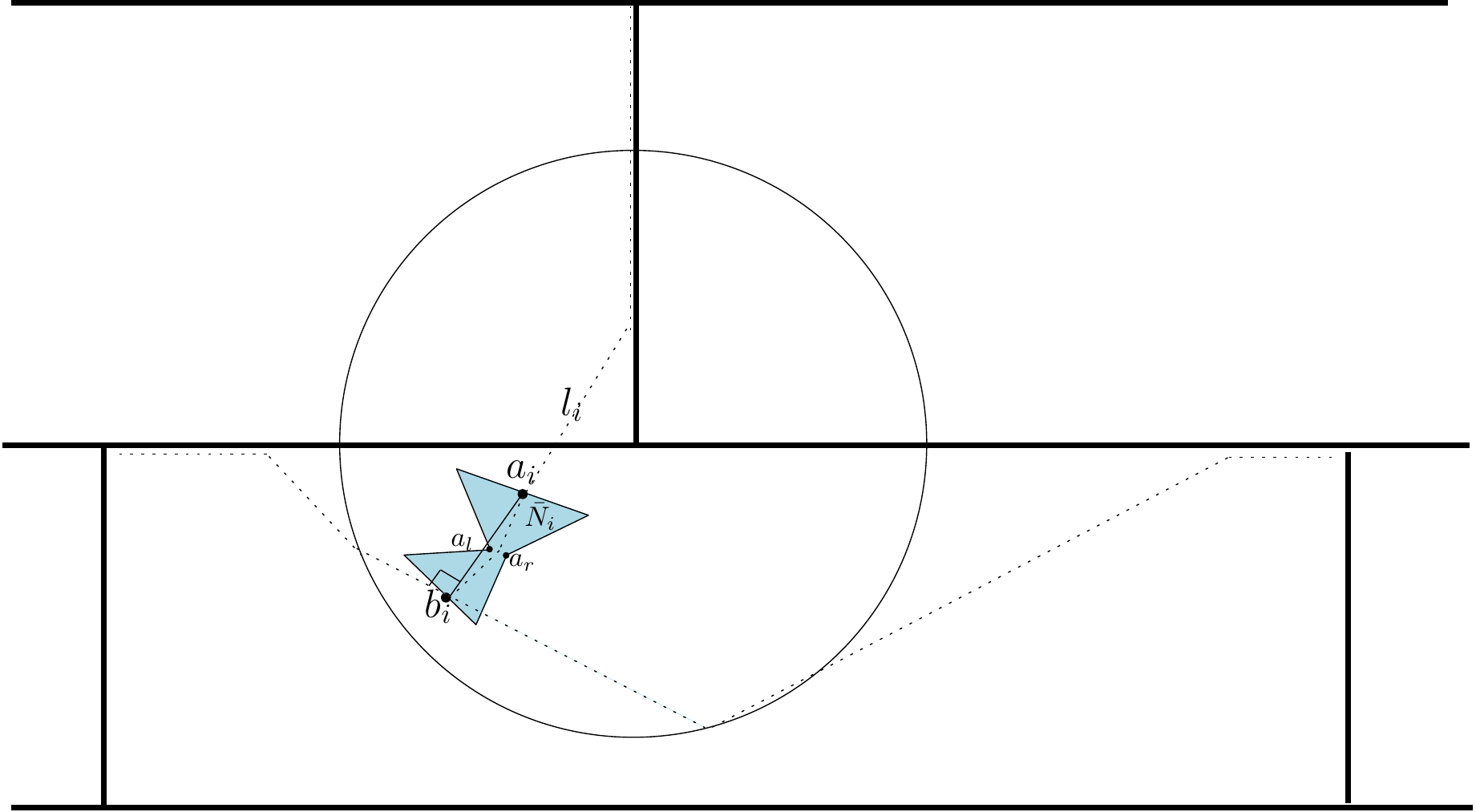}}
 \caption{\label{fig:reroute}
 Rerouting the divisions of ${\cal B}$ to pass through "noop" faces: the dotted lines denote the rerouting of the division so that the intersection $p_i$ is located at a noop $N_i$.} 
\end{figure}
\noindent

Finally, we show that no local term can act on three-bricks:
Let $O$ be an operator in $O_{left}$, then it does not act on the brick to the right of the partition line between $a_i$ and $b_i$.
Hence it is at most $2$-local.
A similar claim holds for any operator in $O_{right}$.
Any other operator has edges of length at most $l_{max}$ that cannot cross $\bar{N}_i$, since it was not
modified by the merging process.
Since we have that $|a_i-b_i|>2l_{max}$ by Equation \ref{eq:abmax},
no such operator can act on the three bricks.
\end{proof}

\section{Tight conditions on Topological Order}\label{sec:TO}
Let us now describe more formally the implications of our results to 
Topological order. 
We define a Topological Order system as follows:
\begin{definition}

\textbf{Topological Order}

\noindent
A quantum state $\ket{\psi_1}$ is said to exhibit Topological Order w.r.t. 
a lattice $L$, (with one particle on each of its edges) 
if there exists another state $\ket{\psi_2}$ 
orthogonal to $\ket{\psi_1}$ with the following properties, for any locally 
confined observable $O$ on $L$:
\begin{enumerate}
\item 
$\bra{\psi_1} O \ket{\psi_2} = 0$
\item
$\bra{\psi_1} O \ket{\psi_1} = \bra{\psi_2} O \ket{\psi_2}$
\end{enumerate}
In this paper, we will 
say that a quantum system exhibits TO if all states in its groundspace 
do. 
\end{definition}

Note that this definition can also be extended to NE 
graphs - since the notion of locality makes sense in those cases too. 

It is shown in \cite{BHV} that Topological order states cannot be generated 
by small depth nearest neighbor circuits:
\begin{theorem}\label{thm:BHV}
Let $U = U_r \cdots U_1$ be a quantum circuit which is a product of local unitaries $U_i$ that generates a $TO$ state $\ket{\psi}$ from the all-zero state, i.e. $\ket{\psi_1} = U \ket{0}^{\otimes n}$ on a $2D$ lattice $L$ of $n$ qubits.
Then $r = \Omega(\sqrt{n})$.
\end{theorem}

In order to conclude that there exist 
NE commuting Hamiltonians with $d\geq 4$ for which 
no constant-depth diagonalizing circuit exists, 
we recall 
the Toric Code due to Kitaev \cite{Kit2}. 
The Toric Code is comprised of a finite square $2$-dimensional lattice of qubits (edges), such that the top and bottom rows are identified as one, and so are the left and right columns.
The system is stabilized by two types of $4$-local operators, one which acts on a "plaquette" of $4$ edges of the unit square of the grid (denoted as the set $P$), and one which acts on a "vertex" of $4$ edges at each crossing at the grid (denote as the set $V$).
The "vertex" operators are $A = X^{\otimes 4}$ and the
"plaquette" operators are $B = Z^{\otimes 4}$.
The complete Hamiltonian is given by 
$$ H = -\sum_{v\in V} A(v) -\sum_{p\in P} B(p). $$ 
Let us now consider grouping together the qubits of the Toric Code lattice, 
by grouping together the two top-right qubits (edges) of every "plaquette". 

\begin{figure}[ht]
\center{
 \epsfxsize=2in
 \epsfbox{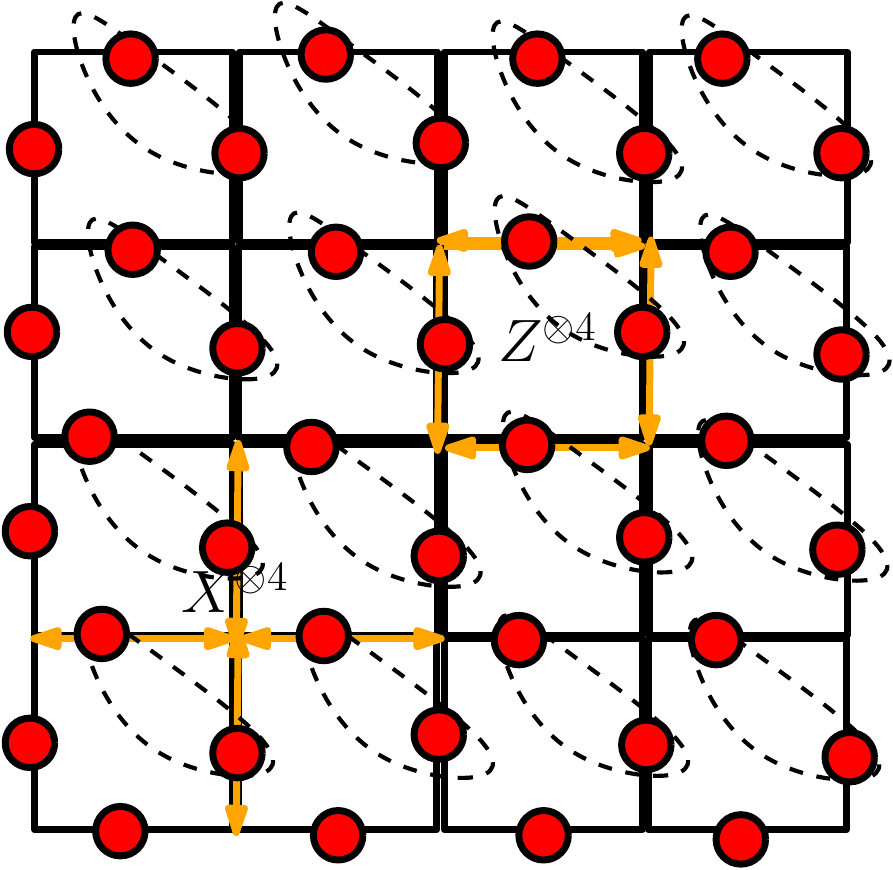}}
 \caption{\label{fig:TC} The figure shows the structure of the Toric Code and the pairing of qubit edges (using the dashed circles) to arrive at a $CLH(3,4)$ instance.} 
\end{figure}

It can be easily seen that, up to a global energy-level shift, 
we arrive at an instance of $CLH(3,4)$.
Thus, the Toric Code is a special case of $CLH(3,4)$.
Since all states in the Toric Code groundspace are of 
Topological Order (\cite{Kit2}) no state in the code can be generated by any constant-depth circuit; 
Hence, in this case there cannot be a constant-depth nearest-neighbor 
quantum circuit which diagonalizes the Toric code Hamiltonian,   
and our conclusion follows.

This conclusion settles an intriguing gap regarding commuting Hamiltonian 
systems on lattices (more generally, on NE graphs).  
On the one hand for $CLH(k,d)$ with $k\geq 3$ and $d\geq 4$, 
or $k>3$ no constant-depth diagonalization exists, and Topological order is possible, whereas for $CLH(k,d)$ with $k<3$ for all $d$, or $k=3$ and $d<4$ constant-depth diagonalization exists, so Topological Order is impossible.
Thus, the Toric Code construction is optimal in terms of 
commuting Hamiltonian dimensionality and size of interactions. 
This implies Theorem \ref{thm:TO}. 

A clarifying remark is in place. Theorems \ref{MainClaim} and \ref{MainClaim2}
hold for planar instances, whereas the 
Toric codes we have just mentioned are defined on a Torus. 
How can we compare the two? 
We recall that the Toric codes 
can actually be generalized to a lattice with a boundary
embedded on the plane, 
as was shown by Bravyi and Kitaev \cite{BK}, 
and so this proves the tight boundary provided by 
Theorem \ref{thm:TO} between TO and local entanglement, 
for NE {\it planar} instances. 
In fact, Theorems \ref{MainClaim} and \ref{MainClaim2}
can be extended to handle higher genus surfaces, 
as long as the graphs are "locally NE planar", 
i.e., large disks around any vertex are NE planar.  
From the other side, the proof of Theorem \ref{thm:BHV} holds 
for such instances too, since
locally NE planar embeddings support the Lieb-Robinson bound used in the proof
of Theorem \ref{thm:BHV}).  
Hence our implications regarding the tight conditions on TO 
hold for higher genus surfaces too; we will not provide here the 
details of this generalization. 

\section{Acknowledgements}
We would like to thank Sergey Bravyi for several important contributions 
to this paper. First, for pointing us to that
our early results on the complexity of the commuting Hamiltonian problem, 
may have interesting implications 
on the conditions for the existence of topological order. 
Second, for showing us the beautiful argument that any 
circuit generating the Toric Code groundstate is of non-constant depth, 
and finally for pointing to us to a mistake in an early draft of this paper.
We also thank Zeph Landau for useful remarks on an early draft of this 
paper. 



\begin{thebibliography}{1}

\bibitem{Aha}
D. Aharonov, D.Gottesman, S. Irani, J. Kempe.
\newblock The power of quantum systems on a line
\newblock {\em Comm. Math. Physics, vol. 287, no. 1, pp. 41-65 (2009)}

\bibitem{Aha2}
D. Aharonov, I. Arad, Z. Landau, U. Vazirani
\newblock The Detectability Lemma and Quantum Gap Amplification     
\newblock {\em quant-ph/0811.3412}, 2008.
    
\bibitem{AV}
M. Aguado, G. Vidal.
\newblock Entanglement Renormalization and Topological Order
\newblock {\em Phys. Rev. Lett. 100, 070404} (2008). 

\bibitem{BSAT}
S. Bravyi.
\newblock Efficient algorithm for a quantum analogue of 2-SAT
\newblock {\em 	quant-ph/0602108}, 2006.
 	
\bibitem{BH}
S. Bravyi, M. B. Hastings.
\newblock A short proof of stability of topological order under local perturbations
\newblock {\em arXiv:1001.4363v1}, 2010. 	
 	
\bibitem{BHV}
S. Bravyi, M. B. Hastings, F. Verstraete.
\newblock Lieb-Robinson Bounds and the Generation of Correlations and Topological Quantum Order
\newblock {\em Physical Review Letters 97.050401}, 2006. 

\bibitem{BK}
S. Bravyi, A. Kitaev
\newblock Quantum codes on a lattice with boundary
\newblock {\em quant-ph/9811052}, 1998.

\bibitem{BV}
S. Bravyi, M. Vyalyi.
\newblock Commutative version of the k-local Hamiltonian problem and common eigenspace problem
\newblock {\em quant-ph/0308021}, 2004.

\bibitem{Has}
M. B. Hastings.
\newblock An area law for one-dimensional quantum systems 
\newblock {\em J. Stat. Mech. (2007) P08024}

\bibitem{Has2}
M. B. Hastings.
\newblock Lieb-Schultz-mattis in higher dimensions
\newblock {\em Phys. Rev. B, vol. 69, p. 104431, Mar 2004.}

\bibitem{KKR}
J. Kempe, A. Kitaev, O. Regev.
\newblock The Complexity of the Local Hamiltonian Problem
\newblock SIAM Journal of Computing, Vol. 35(5), p. 1070-1097 (2006), conference version in Proc. 24th FSTTCS, p. 372-383 (2004)

\bibitem{Kit2}
A. Kitaev.
\newblock Fault-tolerant quantum computation by anyons
\newblock {\em Annals Phys. 303 (2003) 2-30}

\bibitem{Kit1}
A. Yu. Kitaev, A. H. Shen, M. N. Vyalyi.
\newblock Classical and Quantum Computation (Graduate Studies in Mathematics)
\newblock {\em American Mathematical Society, (2002)}

\bibitem{Oli}
R. Oliveira, B. M. Terhal.
\newblock The complexity of quantum spin systems on a two-dimensional square lattice
\newblock {\em Quant. Inf, Comp. Vol. 8, No. 10, pp. 0900-0924 (2008)}
	
\bibitem{TDV}
B. M. Terhal, D. P. DiVincenzo.
\newblock Adaptive Quantum Computation, Constant Depth Quantum Circuits and Arthur-Merlin Games
\newblock arXiv:quant-ph/0205133v6	

\bibitem{Vid}
G. Vidal.
\newblock Entanglement Renormalization
\newblock {\em Phys. Rev. Lett. 99, 220405} (2007).

\end{thebibliography}
\end{document}